\documentclass{CSML}

\def\dOi{12(4:12)2016}
\lmcsheading%
{\dOi}
{1--38}
{}
{}
{Dec.~15, 2015}
{Dec.~31, 2016}
{}

\ACMCCS{[{\bf Theory of computation}]: Models of
  Computation---Concurrency---Process Calculi; Semantics and
  reasoning---Program semantics} 

\subjclass{Probabilistic computation,
  Process calculi, Algebraic language theory, Operational semantics,
  Program specifications}

\usepackage{hyperref}
\usepackage{ifdraft}
\usepackage{comment}
\usepackage{txfonts}
\usepackage{latexsym}
\allowdisplaybreaks

\newtheorem{example}[thm]{Example}
\newtheorem{notation}[thm]{Notation}
\newtheorem{remark}[thm]{Remark}

\newcommand{\SOSrule}[2]{\frac{\displaystyle #1}{\displaystyle #2}}

\newcommand{\D}{\textsf{D}}
\newcommand{\T}{\textsf{T}}
\newcommand{\N}{\mathbb{N}} 

\newcommand{\Rgez}{\mathbb{R}_{\ge 0}}

\newcommand{\copyOp}{\textsf{cp}}

\newcommand{\trans}[1][]{\xrightarrow{\, {#1} \, }}
\newcommand{\ntrans}[1][]{\mathrel{{\trans[#1]}\makebox[0em][r]{$\not$\hspace{2ex}}}{\!}}

\newcommand{\openDTerms}{\mathbb{DT}(\Sigma)}
\newcommand{\closedDTerms}{\textsf{DT}(\Sigma)}
\newcommand{\rank}{\mathop{\sf r}}

\newcommand{\openST}{\mathbb{T}}
\newcommand{\openTerms}{\openSTerms \cup \openDTerms}
\newcommand{\openSTerms}{\openST(\Sigma)}
\newcommand{\closedSTerms}{\T(\Sigma)}

\newcommand{\closedTerms}{\T(\Sigma)}
\newcommand{\VarTerm}{\mathop{\textit{Var}}}

\newcommand{\prem}[1]{\textrm{prem}(#1)}
\newcommand{\conc}[1]{\textrm{conc}(#1)}

\newcommand{\SVar}{\mathcal{V}\!_s}
\newcommand{\DVar}{\mathcal{V}\!_d}
\newcommand{\Var}{\mathcal{V}}

\newcommand{\parop}[1]{\mathop{||_{#1}}}
\newcommand{\Act}{A}
\newcommand{\tick}{{\surd}}

\DeclareMathOperator{\der}{der}

\newcommand{\topd}{\ensuremath{\mathbf{1}}}
\newcommand{\botd}{\ensuremath{\mathbf{0}}}
\newcommand{\bisimd}{\ensuremath{\mathbf{d}}}
\newcommand{\bisimddisc}{\ensuremath{\bisimd}}

\DeclareMathOperator{\Kantorovich}{\mathbf{K}}
\DeclareMathOperator{\Hausdorff}{\mathbf{H}}
\DeclareMathOperator{\Bisimulation}{\mathbf{B}}

\newcommand{\dext}{\sqsubseteq}
\newcommand{\nonrecPASig}{\Sigma_{\text{PA}}}
\newcommand{\nonrecPARules}{R_{\text{PA}}}
\newcommand{\nonrecPASpec}{P_{\text{PA}}}
\newcommand{\recPASig}{\Sigma_{\text{PA}^\circlearrowleft}}

\newcommand{\recPASpec}{P_{\text{PA}^\circlearrowleft}}

\newcommand{\procBRP}{\mathit{BRP}}
\newcommand{\procRC}{\mathit{RC}}
\newcommand{\procTV}{\mathit{TV}}
\newcommand{\procCH}{\mathit{CH}}
\newcommand{\procCHprime}{\mathit{CH'}}
\newcommand{\procCHH}{\mathit{CH}_2}

\newcommand{\actAck}{\mathit{ack}}
\newcommand{\actLost}{\mathit{lost}}
\newcommand{\actResOk}{\mathit{res(OK)}}
\newcommand{\actResNok}{\mathit{res(NOK)}}

\begin{document}

\title[Compositional bisimulation metric reasoning with Probabilistic Process Calculi]{Compositional bisimulation metric reasoning with \\ Probabilistic Process Calculi\rsuper*}

\author[D.~Gebler]{Daniel Gebler\rsuper a}
\address{{\lsuper a}VU University Amsterdam (NL)}
\email{e.d.gebler@vu.nl} 

\author[K.~G.~Larsen]{Kim G. Larsen\rsuper b}
\address{{\lsuper b}Aalborg University (DK)}	
\email{kgl@cs.aau.dk}  
\thanks{This research is partially supported by the European FET projects SENSATION and CASSTING and the Sino-Danish Center IDEA4CPS.}	

\author[S.~Tini]{Simone Tini\rsuper c}	
\address{{\lsuper c}University of Insubria (IT)}	
\email{simone.tini@uninsubria.it} 


\titlecomment{{\lsuper*}A preliminary version of this paper appeared as~\cite{GLT14}}

\begin{abstract}
We study which standard operators of probabilistic process calculi allow for compositional reasoning with respect to bisimulation metric semantics. We argue that uniform continuity (generalizing the earlier proposed property of non-expansiveness) captures the essential nature of compositional reasoning and allows now also to reason compositionally about recursive processes. We characterize the distance between probabilistic processes composed by standard process algebra operators. Combining these results, we demonstrate how compositional reasoning about systems specified by continuous process algebra operators allows for metric assume-guarantee like performance validation.
\end{abstract}

\maketitle

\section{Introduction} \label{sec:introduction}

Probabilistic process algebras, such as probabilistic CCS~\cite{JYL01,Bar04,DD07}, CSP~\cite{JYL01,Bar04,DGHMZ07b,DL12} and ACP~\cite{And99,And02}, are languages that are employed to describe probabilistic concurrent communicating systems, or probabilistic processes for short. 
Nondeterministic probabilistic transition systems~\cite{Seg95a} combine labeled transition systems~\cite{Kel76} and discrete time Markov chains~\cite{Ste94,HJ94}. They allow us to model separately the reactive system behavior, nondeterministic choices and probabilistic choices.

Behavioral semantics provide formal notions to compare systems. Behavioral equivalences are behavioral semantics that allow us to determine the observational equivalence of systems by abstracting from behavioral details that may be not relevant in a given application context. In essence, behavioral equivalences equate processes that are indistinguishable to any external observer. The most prominent example is bisimulation equivalence~\cite{LS91,SL95,Seg95a}, which provides a well-established theory of the behavior of probabilistic nondeterministic transition systems. 

Recently it became clear that the notion of behavioral equivalence is too strict in the context of probabilistic models. The probability values in those models originate either from observations (statistical sampling) or from requirements (probabilistic specification). Behavioral equivalences such as bisimulation equivalence are binary notions that can only answer the question if two systems behave precisely the same way or not. However, a tiny variation of the probabilities, which may be due to a measurement error or limitations how precise a specified probabilistic choice can be realized in a concrete system, will make these systems behaviorally inequivalent without any further information. 
In practice, many systems are approximately correct. This leads immediately to the question of what is an appropriate notion to measure the quality of the approximation.
The most prominent notion is behavioral metric semantics~\cite{DGJP04,BW05,DCPP06} which provides a behavioral distance that characterizes how far the behavior of two systems is apart.
Bisimulation metrics are the quantitative analogue to bisimulation equivalences and assign to each pair of processes a distance which measures the proximity of their quantitative properties. 
The distances form a pseudometric\footnote{A bisimulation metric is in fact a pseudometric. For convenience we use the term bisimulation metric instead of bisimulation pseudometric.} with bisimilar processes at distance $0$. 

In order to specify and verify systems in a compositional manner, it is necessary that the behavioral semantics is compatible with all operators of the language that describe these systems. For behavioral equivalence semantics there is common agreement that compositional reasoning requires that the considered behavioral equivalence is a congruence with respect to all language operators.
For example, consider a term $f(s_1,s_2)$ which describes a system consisting of subcomponents $s_1$ and $s_2$ that are composed by the binary operator $f$. When replacing $s_1$ with a behaviorally equivalent $s_1'$,  and $s_2$ with a behaviorally equivalent $s_2'$, congruence of the operator $f$ guarantees that the composed system $f(s_1,s_2)$ is behaviorally equivalent to the resulting replacement system $f(s_1',s_2')$.
This implies that equivalent systems are inter-substitutable: Whenever a system $s$ in a language context $C[s]$ is replaced by an equivalent system $s'$, the obtained context $C[s']$ is equivalent to $C[s]$.
The congruence property is important since it is usually much easier to model and study (a set of) small systems and then combine them together rather than to work with a large monolithic system.

However, for behavioral metric semantics there is no satisfactory understanding of which property an operator should satisfy in order to facilitate compositional reasoning. 
Intuitively, what is needed is a formalization of the idea that systems close to each other should be approximately inter-substitutable: Whenever a system $s$ in a language context $C[s]$ is replaced by a close system $s'$, the obtained context $C[s']$ should be close to $C[s]$. 
In other words, there should be some relation between the behavioral distance between $s$ and $s'$ and the behavioral distance between $C[s]$ and $C[s']$.
This ensures that any limited change in the behavior of a subcomponent $s$ implies a smooth and limited change in the behavior of the composed system $C[s]$ (absence of chaotic behavior when system components and parameters are modified in a controlled manner).
Earlier proposals such as non-expansiveness~\cite{DGJP04} and non-extensiveness~\cite{BBLM13b} are only partially satisfactory for non-recursive operators and even worse, they do not allow at all to reason compositionally over recursive processes.
More fundamentally, those proposals are kind of `ad hoc' and do not capture systematically the essential nature of compositional metric reasoning. 

In this paper we consider uniform continuity as a property that generalizes non-extensiveness and non-expansiveness and captures the essential nature of compositional reasoning w.r.t.\ behavioral metric semantics. 
A uniformly continuous binary process operator $f$ ensures that for any non-zero bisimulation distance $\epsilon$ (understood as the admissible tolerance from the operational behavior of the composed process $f(s_1,s_2)$) there are non-zero bisimulation distances $\delta_1$ and $\delta_2$ (understood as the admissible tolerances from the operational behavior of the processes $s_1$ and $s_2$) such that the distance between the composed processes $f(s_1,s_2)$ and $f(s_1',s_2')$ is at most $\epsilon$ whenever the component $s_1'$ (resp.\ $s_2'$) is in distance of at most $\delta_1$ from $s_1$ (resp.\ at most $\delta_2$ from $s_2$). 
Uniform continuity ensures that a small variance in the behavior of the parts leads to a bounded small variance in the behavior of the composed processes.
Since uniformly continuous operators preserve the convergence of sequences, this allows us to approximate composed systems by approximating its subsystems. In summary, uniform continuity allows us to investigate the behavior of systems by disassembling them into their components, analyze at the component level, and then derive properties of the composed system. 
We consider the uniform notion of continuity (technically, the $\delta_i$ depend only on $\epsilon$ and are independent of the concrete systems $s_i$) because we aim at universal compositionality guarantees.
As important notion of uniform continuity we consider Lipschitz continuity which ensures that the ratio between the distance of composed processes and the distance between its parts is bounded.

Our main contributions are as follows:
\begin{enumerate}
	\item We develop for many non-recursive and recursive process operators used in various probabilistic process algebras tight upper bounds on the distance between processes combined by those operators (Sections~\ref{sec:distance_nonrec_procs} and~\ref{sec:distance_rec_procs}). 
	\item We show that non-recursive process operators, esp.\ (nondeterministic and probabilistic variants of) sequential, alternative and parallel composition, allow for compositional reasoning w.r.t.\ the compositionality criteria of non-expansiveness and hence also w.r.t.\ both Lipschitz and uniform continuity (Section~\ref{sec:nonrec_procs}). 
	\item We show that recursive process operators, e.g.\ (nondeterministic and probabilistic variants of) Kleene-star iteration and $\pi$-calculus bang replication, allow for compositional reasoning w.r.t.\ the compositionality criterion of Lipschitz continuity and hence also w.r.t.\ uniform continuity, but not w.r.t.\ non-expansiveness and non-extensiveness (Section~\ref{sec:rec_procs}). 
                \item We discuss the copy operator proposed in~\cite{BIM95,FvGdW12} to specify the fork operation of operating systems as an example of operator allowing for compositional reasoning w.r.t.\ the compositionality criterion of uniform  continuity, but not w.r.t.\ Lipschitz continuity.
	\item We demonstrate the practical relevance of our methods by reasoning compositionally over a network protocol built from uniformly continuous operators. In detail, we show how to derive performance guarantees for the entire system from performance assumptions about individual components. In reverse, we show also how to derive performance requirements on individual components from performance requirements of the complete system (Section~\ref{sec:application}).
\end{enumerate}

\section{Preliminaries}\label{sec:preliminaries}

\subsection{Probabilistic Transition Systems}
We consider transition systems with process terms as states and labeled transitions taking states to distributions over states. 
Process terms are inductively defined by process combinators. 

\begin{defi}[Signature]
A \emph{signature} is a structure $\Sigma = (F, \rank)$, where 
\begin{enumerate}[]
	\item $F$ is a countable set of \emph{operators}, and
	\item $\rank \colon F \to \N$ is a \emph{rank function}.
\end{enumerate}
\end{defi}
The rank function gives by $\rank(f)$ the arity of operator $f$.  
We call operators with arity $0$ \emph{constants}.
If the rank of $f$ is clear from the context we will use the symbol $n$ for $\rank(f)$. 
We may write $f\in\Sigma$ as shorthand for $\Sigma=(F,\rank)$ with $f\in F$.

Terms are defined by structural recursion over the signature.
We assume an infinite set of \emph{state variables} $\SVar$ disjoint from $F$.

\begin{defi}[State terms] \label{def:state_terms}
The set of \emph{state terms} over a signature $\Sigma$ and a set $V \subseteq \SVar$ of state variables, notation $\T(\Sigma, V)$, is the least set satisfying: 
\begin{itemize}
	\item $V \subseteq \T(\Sigma, V)$, and
	\item $f(t_1, \ldots, t_n) \in \T(\Sigma, V)$ whenever $f \in \Sigma$ and $t_1, \ldots, t_n \in \T(\Sigma, V)$.
\end{itemize}
\end{defi}
We write $c$ for $c()$ if $c$ is a constant.
The set of \emph{closed state terms} $\T(\Sigma, \emptyset)$ is abbreviated as $\closedSTerms$. 
The set of \emph{open state terms} $\T(\Sigma, \SVar)$ is abbreviated as $\openSTerms$. 
We may refer to operators in $\Sigma$ as \emph{process combinators}, to state variables in $\SVar$ as \emph{process variables}, and to closed state terms in $\closedSTerms$ as \emph{processes}.

A probability distribution over the set of closed state terms $\closedSTerms$ is a mapping $\pi \colon \closedSTerms \to [0,1]$ with $\sum_{t \in \closedTerms} \pi(t) = 1$ that assigns to each closed term $t \in \closedSTerms$ its respective probability $\pi(t)$. The probability mass of a set of closed terms $T \subseteq \closedSTerms$ in some probability distribution $\pi$ is  given by $\pi(T)=\sum_{t\in T}\pi(t)$.
We denote by $\Delta(\closedSTerms)$ the set of all probability distributions over $\closedSTerms$. 
We let $\pi,\pi'$ range over $\Delta(\closedSTerms)$.

\begin{notation}[Notations for probability distributions]
We denote by $\delta(t)$ with $t \in \closedSTerms$ the \emph{Dirac distribution} defined by $(\delta(t))(t)=1$ and $(\delta(t))(t')=0$ for all $t' \in \closedTerms$ with $t \neq t'$. 
The convex combination $\sum_{i \in I} p_i \pi_i$ of a family $\{\pi_i\}_{i \in I}$ of probability distributions $\pi_i \in \Delta(\closedSTerms)$ with $p_i \in (0,1]$ and $\sum_{i \in I} p_i = 1$ is defined by $(\sum_{i \in I} p_i \pi_i)(t) = \sum_{i \in I} (p_i \pi_i(t))$ for all terms $t \in \closedTerms$. 
The expression $f(\pi_1,\dots,\pi_n)$ with $f \in \Sigma$ and $\pi_i \in \Delta(\closedSTerms)$ denotes the product distribution of $\pi_1,\ldots,\pi_n$ defined by $(f(\pi_1,\dots,\pi_n))(f(t_1,\ldots,t_n)) = \prod_{i=1}^n \pi_i(t_i)$ and $(f(\pi_1,\dots,\pi_n))(t) = 0$ for all terms $t \in \closedTerms$ not in the form $t = f(t_1,\dots,t_n)$. 
For binary operators $f$ we may use the infix notation and write $\pi_1 \, f \, \pi_2$ for $f(\pi_1,\pi_2)$.
\end{notation}

Next, we introduce a language to describe probability distributions.
We assume an infinite set of \emph{distribution variables} $\DVar$ and let $\mu,\nu$ range over $\DVar$. 
We denote by $\Var$ the set of state and distribution variables $\Var = \SVar \cup \DVar$ and let $\zeta,\zeta'$ range over $\Var$.

\begin{defi}[Distribution terms] \label{def:distribution_term}
The set of \emph{distribution terms} over 
a signature $\Sigma$, 
a set of state variables $V_s \subseteq \SVar$ and 
a set of distribution variables $V_d \subseteq \DVar$, 
notation $\D\T(\Sigma, V_s, V_d)$, 
is the least set satisfying: 
\begin{enumerate}
	\item \label{def:DT:var}
		$V_d \subseteq \D\T(\Sigma, V_s, V_d)$, 
	\item \label{def:DT:dirac}
		$\{\delta(t) \mid t \in \T(\Sigma, V_s)\} \subseteq \D\T(\Sigma, V_s, V_d)$, 
	\item \label{def:DT:sum} 
		${\textstyle \sum_{i\in I} p_i \theta_i \in \D\T(\Sigma, V_s, V_d)}$ whenever $\theta_i \in \D\T(\Sigma, V_s, V_d)$ and $p_i \in (0,1]$ with $\sum_{i\in I} p_i = 1$, and
	\item \label{def:DT:prod} 
		$f(\theta_1,\ldots,\theta_n) \in \D\T(\Sigma, V_s, V_d)$ whenever $f \in \Sigma$ and $\theta_1,\ldots,\theta_n \in \D\T(\Sigma, V_s, V_d)$.
\end{enumerate}
\end{defi}
Distribution terms have the following meaning. A \emph{distribution variable} $\mu \in \DVar$ is a variable that takes values from $\Delta(\closedSTerms)$. An \emph{instantiable Dirac distribution} $\delta(t)$ is an expression that takes as value the Dirac distribution $\delta(t')$ when state variables in $t$ are substituted such that $t$ becomes the closed term $t'$. 
Case~\ref{def:DT:sum} allows us to construct convex combinations of distributions. 
Case~\ref{def:DT:prod} lifts structural recursion from state terms to distribution terms.

The set of \emph{closed distribution terms} $\D\T(\Sigma, \emptyset, \emptyset)$ is abbreviated as $\closedDTerms$.
The set of \emph{open distribution terms} $\D\T(\Sigma, \SVar, \DVar)$ is abbreviated as $\openDTerms$.
We write $\theta_1 \oplus_p \theta_2$ for $\sum_{i=1}^2 p_i \theta_i$ with $p_1 = p$ and $p_2=1-p$.
Furthermore, for binary operators $f$ we may use the infix notaion and write $\theta_1\,f\,\theta_2$ for $f(\theta_1,\theta_2)$.

\begin{defi}[Substitution]
A \emph{substitution} is a mapping $\sigma \colon \Var \to \openTerms$ such that $\sigma(x) \in \openSTerms$, if $x\in \SVar$, and $\sigma(\mu) \in \openDTerms$, if $\mu\in \DVar$. A substitution $\sigma$ extends to a mapping from state terms to state terms by $\sigma(f(t_1,\ldots,t_n)) = f(\sigma(t_1),\ldots,\sigma(t_n))$. A substitution $\sigma$ extends to a mapping from distribution terms to distribution terms by 
\begin{enumerate}[label=(\roman*)]
\item $\sigma(\delta(t))=\delta(\sigma(t))$, 
\item $\sigma(\sum_{i\in I} p_i \theta_i) = \sum_{i\in I} p_i \sigma(\theta_i)$, and
\item $\sigma(f(\theta_1,\ldots,\theta_n)) = f(\sigma(\theta_1),\ldots,\sigma(\theta_n))$. 
\end{enumerate}
\end{defi}\medskip

\noindent A substitution $\sigma$ is \emph{closed} if $\sigma(x) \in \closedSTerms$ for all $x \in \SVar$ and $\sigma(\mu) \in \closedDTerms$ for all $\mu \in \DVar$. 
Notice that closed distribution terms denote distributions in $\Delta(\closedSTerms)$.

Probabilistic nondeterministic labelled transition systems~\cite{Seg95a}, PTSs for short, extend labelled transition systems by allowing for probabilistic choices in the transitions.
As state space we will take the set of all closed terms $\closedSTerms$.

\begin{defi}[PTS, \cite{Seg95a}] \label{def:pts}
A \emph{probabilistic nondeterministic labeled transition system} (\emph{PTS}) over the signature $\Sigma$ is given by a triple $(\closedSTerms,\Act,{\trans})$, where:
\begin{itemize}
	\item $\closedSTerms$ is the set of all closed terms over $\Sigma$,
	\item $\Act$ is a countable set of \emph{actions}, and 
	\item ${\trans} \subseteq {\closedSTerms \times \Act \times \Delta(\closedSTerms)}$ is a \emph{transition relation}. 
\end{itemize}
\end{defi}
We call ${(t,a,\pi)} \in {\trans}$ a \emph{transition} from state $t$ to distribution $\pi$ labelled by action $a$.
We write $t \trans[a] \pi$ for ${(t,a,\pi)} \in {\trans}$.
Moreover, we write 
$t \trans[a]$ if there exists some distribution $\pi \in \Delta(\closedSTerms)$ with $t \trans[a] \pi$, and
$t \ntrans[a]$ if there is no distribution $\pi \in \Delta(\closedSTerms)$ with $t \trans[a] \pi$. 
For a closed term $t \in \closedTerms$ and an action $a \in \Act$, let $\mathit{der}(t,a) = \{\pi \in \Delta(\closedSTerms) \mid t \trans[a] \pi \}$ denote the set of all distributions reachable from $t$ by performing an $a$-labeled transition. We call $\mathit{der}(t,a)$ also the \emph{$a$-derivatives} of $t$.

We say that a PTS is \emph{image-finite} if $\der(t,a)$ is finite for each closed term $t$ and action $a$.
In the rest of the paper we assume to deal with image finite PTSs.

\subsection{Bisimulation metric}

Bisimulation metric\footnote{A bisimulation metric is in fact a pseudometric. In line with the literature we use the term bisimulation metric instead of bisimulation pseudometric.}~\cite{DGJP04,BW05,DCPP06} provides a robust semantics for PTSs.
It is the quantitative analogue to bisimulation equivalence and assigns to each pair of states a distance which measures the proximity of their quantitative properties.
The distances form a pseudometric where bisimilar processes are at distance $0$. 

\begin{defi}[Pseudometric over $\closedTerms$]
A function $d \colon \closedSTerms \times \closedSTerms \to [0,1]$ is a \emph{1-bounded pseudometric} if
\begin{itemize}	
	\item $d(t,t)= 0$ for all $t \in \closedSTerms$,
	\item $d(t,t') = d(t',t)$ for all $t,t' \in \closedSTerms$ (symmetry), and
	\item $d(t,t') \le d(t,t'') + d(t'',t')$ for all $t,t',t'' \in \closedSTerms$ (triangle inequality).
\end{itemize}
\end{defi}
We will define later bisimulation metrics as $1$-bounded pseudometrics that measure how much two states disagree on their reactive behavior and their probabilistic choices. 
Note that a pseudometric $d$ permits that $d(t,t')=0$ even if $t$ and $t'$ are different terms (in contrast to a metric $d$). This will allow us to assign distance $0$ to different bisimilar states.
We will provide two (equivalent) characterizations of bisimulation metrics in terms of a coinductive definition pattern and in terms of fixed points.

Both characterizations require the following lattice structure.
Let ${([0,1]^{\closedSTerms \times \closedSTerms},\sqsubseteq)}$ be the complete lattice of functions 
$d \colon \closedSTerms \times \closedSTerms \to [0,1]$
ordered by $d_1 \sqsubseteq d_2$ iff $d_1(t, t') \le d_2(t, t')$ for all $t, t' \in \closedSTerms$. 
Then for each $D \subseteq [0,1]^{\closedSTerms \times \closedSTerms}$ the supremum and infinimum are 
$\sup(D)(t,t') = \sup_{d \in \D}d(t,t')$ and 
$\inf(D)(t,t') = \inf_{d \in \D}d(t,t')$ for all $t, t' \in \closedSTerms$.
The bottom element is the constant zero function $\botd$ given by $\botd(t,t')=0$, and 
the top element is the constant one function $\topd$ given by $\topd(t,t')=1$, for all $t,t' \in \closedSTerms$.

\subsubsection{Metrical lifting}

Bisimulation metric is characterized using the quantitative analogous of the bisimulation game, meaning that two states $t,t' \in \closedSTerms$ at some given distance can mimic each other's transitions and evolve to distributions that are at distance not greater than the distance between the source states. 
Technically, we need a notion that lifts pseudometrics from states to distributions (to capture probabilistic choices).

A $1$-bounded pseudometric on terms $\closedSTerms$ is lifted to a $1$-bounded pseudometric on distributions $\Delta(\closedSTerms)$ by means of the Kantorovich pseudometric~\cite{Den09}. This lifting is the quantitative analogous of the lifting of bisimulation equivalence relations on terms to bisimulation equivalence relations on distributions~\cite{BW01b}.

A \emph{matching}
for a pair of distributions $(\pi,\pi') \in \Delta(\closedSTerms) \times \Delta(\closedSTerms)$ is a distribution over the product state space $\omega \in \Delta(\closedSTerms \times \closedSTerms)$ with left marginal $\pi$, i.e.\ $\sum_{t'\in \closedSTerms} \omega(t,t')=\pi(t)$ for all $t \in \closedSTerms$, and right marginal $\pi'$, i.e.\ $\sum_{t\in \closedSTerms} \omega(t,t')=\pi'(t')$ for all $t' \in \closedSTerms$. 
Let $\Omega(\pi,\pi')$ denote the set of all matchings for $(\pi,\pi')$. 
Intuitively, a matching $\omega \in \Omega(\pi,\pi')$ may be understood as a transportation schedule that describes the shipment of probability mass from $\pi$ to $\pi'$. Historically this motivation dates back to the Monge-Kantorovich optimal transport problem~\cite{Vil08}.

\begin{defi}[Kantorovich lifting] \label{def:KantorovichLifting}
Let $d\colon \closedSTerms \times \closedSTerms \to [0,1]$ be a $1$-bounded pseudometric. The \emph{Kantorovich lifting} of $d$ is a $1$-bounded pseudometric $\Kantorovich(d)\colon \Delta(\closedSTerms) \times \Delta(\closedSTerms) \to [0,1]$ defined by
\[
	\Kantorovich(d)(\pi,\pi') = \min_{\omega \in \Omega(\pi,\pi')} \sum_{t,t'\in \closedSTerms}d(t,t') \cdot \omega(t,t')
\]
for all $\pi,\pi' \in \Delta(\closedSTerms)$.
We call $\Kantorovich(d)$ the \emph{Kantorovich pseudometric} of $d$.
\end{defi}

In order to capture nondeterministic choices, we need to lift pseudometrics on distributions to pseudometrics on sets of distributions.

\begin{defi}[Hausdorff lifting] \label{def:HausdorffLifting}
Let $\hat{d}\colon \Delta(\closedSTerms) \times \Delta(\closedSTerms) \to [0,1]$ be a $1$-bounded pseudometric. The \emph{Hausdorff lifting} of $\hat{d}$ is a $1$-bounded pseudometric $\Hausdorff(\hat{d})\colon P(\Delta(\closedSTerms)) \times P(\Delta(\closedSTerms)) \to [0,1]$ defined by
\[
        \Hausdorff(\hat{d})(\Pi_1,\Pi_2) = \max \left\{ 
\sup_{\pi_1 \in \Pi_1}\inf_{\pi_2 \in \Pi_2} \hat{d}(\pi_1,\pi_2), 
\sup_{\pi_2\in \Pi_2}\inf_{\pi_1\in \Pi_1} \hat{d}(\pi_2,\pi_1) \right\}
\]
for all $\Pi_1,\Pi_2 \subseteq \Delta(\closedSTerms)$, with $\inf \emptyset = 1$, and $\sup \emptyset = 0$.
We call $\Hausdorff(\hat{d})$ the \emph{Hausdorff pseudometric} of $\hat{d}$.
\end{defi}

\subsubsection{Coinductive characterization}

A $1$-bounded pseudometric is a bisimulation metric if for all pairs of terms $t$ and $t'$ each transition of $t$ can be mimicked by a transition of $t'$ with the same label and the distance between the accessible distributions does not exceed the distance between $t$ and $t'$. By means of a \emph{discount factor} $\lambda \in (0,1]$, we allow to specify how much the behavioral distance of future transitions is taken into account~\cite{AHM03,DGJP04}. The discount factor $\lambda=1$ expresses no discount, meaning that the differences in the behavior between $t$ and $t'$ are considered irrespective of after how many steps they can be observed.

\begin{defi}[Bisimulation metric~\cite{DGJP04}] \label{def:bisim_metric}
A $1$-bounded pseudometric $d \colon \closedTerms \times \closedTerms \to [0,1]$ 
is a \emph{$\lambda$-bisimulation metric} with $\lambda \in (0,1]$ if for all terms $t,t'\in \closedSTerms$ with $d(t,t')<1$, if $t \trans[a] \pi$ then there exists a transition $t' \trans[a] \pi'$ for a distribution $\pi' \in \Delta(\closedTerms)$ such that  $\lambda \cdot \Kantorovich(d)(\pi,\pi') \le d(t,t')$.
\end{defi}
We refer to $\lambda \cdot \Kantorovich(d)(\pi,\pi') \le d(t,t')$ as the bisimulation transfer condition.
We call the smallest (w.r.t.\ $\sqsubseteq$) $\lambda$-bisimulation metric~\emph{$\lambda$-bisimilarity metric} \cite{DCPP06} and denote it by the symbol $\bisimddisc$.
We mean by \emph{$\lambda$-bisimulation distance} between $t$ and $t'$ the distance $\bisimddisc(t,t')$. If $\lambda$ is clear from the context, we may refer by 
bisimulation metric, bisimilarity metric and bisimulation distance to 
$\lambda$-bisimulation metric, $\lambda$-bisimilarity metric and $\lambda$-bisimulation distance.
Moreover, we may call the $1$-bisimilarity metric also non-discounting bisimilarity metric.
Bisimilarity equivalence is the kernel of the $\lambda$-bisimilarity metric \cite{DGJP04}, namely  $\bisimddisc(t,t') = 0$ iff $t$ and $t'$ are bisimilar.

\begin{example} \label{ex:coinductive_bisim}
Assume a PTS with transitions ${\trans}={\{s \trans[a] \pi_s, t \trans[a] \pi_t\}}$ whereby 
$\pi_s=0.5\delta(s)+0.5\delta(0)$ and 
$\pi_t=(0.5+\epsilon)\delta(s)+(0.5-\epsilon)\delta(0)$ for some arbitrary $\epsilon \in [0,0.5]$.
Furthermore, assume a $1$-bounded pseudometric $d$ with $d(s,s)=d(0,0)=0$ and $d(s,0)=d(0,s)=1$. 
We have $\Kantorovich(d)(\pi_s,\pi_t)=\epsilon$, by the matching $\omega \in \Omega(\pi_s,\pi_t)$ defined by $\omega(s,s) = 0.5$, $\omega(0,s) = \epsilon$ and $\omega(0,0) = 0.5 -\epsilon$.
Then, $d$ is a bisimulation metric if it satisfies the bisimulation transfer condition $d(s,t) \ge \lambda\Kantorovich(d)(\pi_s,\pi_t) = \lambda \epsilon$. Moreover, the bisimilarity metric assigns the distance $\bisimddisc(t,s)=\lambda\epsilon$.
\end{example}

\subsubsection{Fixed point characterization}

We provide now an alternative characterization of bisimulation metric in terms of prefixed points of an appropriate monotone bisimulation functional~\cite{DCPP06}. Bisimilarity metric is then the least fixed point of this functional.
Moreover, the fixed point approach allows us also to express up-to-$k$ bisimulation metrics which measure the bisimulation distance for only the first $k$ transition steps. 

\begin{defi}[Bisimulation metric functional] \label{def:metric_bisim_functional}
Let $\Bisimulation \colon [0,1]^{\closedSTerms \times \closedSTerms} \to [0,1]^{\closedSTerms \times \closedSTerms}$ be the function defined by
\[
	\Bisimulation(d)(t,t') = \sup_{a\in A} \left\{ \Hausdorff(\lambda \cdot \Kantorovich(d))(\mathit{der}(t,a), \mathit{der}(t',a)) \right\}
\]
for $d \colon \closedSTerms \times \closedSTerms \to [0,1]$ and $t,t' \in \closedSTerms$, with
$(\lambda \cdot \Kantorovich(d))(\pi,\pi')=\lambda \cdot \Kantorovich(d)(\pi,\pi')$.
\end{defi}
It is easy to show that $\Bisimulation$ is a monotone function on $([0,1]^{\closedSTerms \times \closedSTerms},\sqsubseteq)$.
The following Proposition characterizes bisimulation metrics as prefixed points of $\Bisimulation$.

\begin{prop}[\cite{DCPP06}] \label{prop:coinductive_vs_fixpoint_def_bisim_metric}
Let $d \colon \closedSTerms \times \closedSTerms \to [0,1]$ be a $1$-bounded pseudometric. Then $\Bisimulation(d) \sqsubseteq d$ iff $d$ is a bisimulation metric.
\end{prop}
Proposition~\ref{prop:coinductive_vs_fixpoint_def_bisim_metric} provides the fixed point characterization of bisimulation metrics and shows that it coincides with the coinductive characterization of Definition~\ref{def:bisim_metric}.
Since $\Bisimulation$ is a monotone function on the complete lattice $([0,1]^{\closedSTerms \times \closedSTerms},\sqsubseteq)$, we can characterize the bisimilarity metric as least fixed point of $\Bisimulation$.

\begin{prop}[\cite{DCPP06}] \label{prop:bisim_metric_lfp_D}
The bisimilarity metric $\bisimd$ is the least fixed point of $\Bisimulation$.
\end{prop}

Moreover, the fixed point approach allows us to define a notion of bisimulation distance that considers only the first $k$ trasnsition steps. 

\begin{defi}[Up-to-$k$ bisimilarity metric] \label{def:bisim_metric_uptok}
We define the \emph{up-to-$k$ bisimilarity metric} $\bisimddisc_k$ for $k \in \N$ by $\bisimddisc_k = \Bisimulation^k(\botd)$.
\end{defi}
We call $\bisimddisc_k(s,t)$ the up-to-$k$ bisimulation distance between $s$ and $t$.

If the PTS is image-finite and, moreover, for each transition $t \trans[a] \pi$ we have that the support of $\pi$ is finite, then $\Bisimulation$ is monotone and continuous, which ensures that the closure ordinal of $\Bisimulation$ is $\omega$~\cite{vB12}-Section 3. As a consequence, up-to-$k$ bisimulation distances converge to the bisimulation distances when $k \to \infty$, which opens the door to show properties of the bisimulation metric by using a simple inductive argument \cite{vB12}. 
\begin{prop}[\cite{vB12}]
Assume an image-finite PTS s.t.\ for each transition $t \trans[a] \pi$ we have that the distribution $\pi$ has finite support. Then $\bisimddisc = \lim_{k \to \infty} \bisimddisc_k$.
\end{prop}

\subsubsection{Properties of bisimulation metrics}

We give now an important property of bisimulation metrics that will be essential for the argumentation later in the technical sections.

The bisimulation distance between states $t$ and $t'$ measures the difference of the reactive behavior of $t$ and $t'$ (i.e.\ which actions can or cannot be performed) along their evolution. 
An important distinction is if two states can perform the same initial actions. 
In this case, the behavioral distance is given by the bisimulation game on the derivatives.
Otherwise, the two states get the maximal distance of $1$ assigned since there is a transition by one of these states that cannot be mimicked by the other state.

We say that states $t$ and $t'$ \emph{do not totally disagree} if $\bisimddisc(t,t') < 1$. If states do not totally disagree, then they agree on which actions they can perform immediately. 

\begin{prop} \label{prop:immidiate_reactive_behavior_vars_distance_less_one}
Let $d \colon \closedSTerms \times \closedSTerms \to [0,1]$ be a $1$-bounded pseudometric. Then
\begin{enumerate}
	\item \label{prop:immidiate_reactive_behavior_vars_distance_less_one:bisim_functional}
		$\Bisimulation(d)(t,t') < 1$ implies $t \trans[a] \ \Leftrightarrow\  t '\trans[a]$ for all $a \in \Act$,
	\item \label{prop:immidiate_reactive_behavior_vars_distance_less_one:bisim_metric}
		$d(t,t')<1$ implies $t \trans[a] \ \Leftrightarrow\  t '\trans[a]$ for all $a \in \Act$, if $d$ is a bisimulation metric.
\end{enumerate}
\end{prop}
\begin{proof}
We start with Proposition~\ref{prop:immidiate_reactive_behavior_vars_distance_less_one}.\ref{prop:immidiate_reactive_behavior_vars_distance_less_one:bisim_functional} and reason as follows.
\begin{align*}
&  \Bisimulation(d)(t,t') < 1 \\
\Leftrightarrow\ \ & \forall a \in \Act. \Hausdorff(\lambda \cdot \Kantorovich(d))(\mathit{der}(t,a), \mathit{der}(t',a)) < 1 \\
\Rightarrow\ \ & \forall a \in \Act.((\mathit{der}(t,a) = \emptyset = \mathit{der}(t',a)) \,\vee\, (\mathit{der}(t,a) \neq \emptyset \neq \mathit{der}(t',a))) \\
\Leftrightarrow\ \ & \forall a \in \Act.(t \trans[a] \ \Leftrightarrow\ t' \trans[a]).
\end{align*}
Now we show Proposition~\ref{prop:immidiate_reactive_behavior_vars_distance_less_one}.\ref{prop:immidiate_reactive_behavior_vars_distance_less_one:bisim_metric}. By Proposition~\ref{prop:coinductive_vs_fixpoint_def_bisim_metric} we get that $d(t,t')<1$ implies $\Bisimulation(d)(t,t') < 1$. The thesis follows now from Proposition~\ref{prop:immidiate_reactive_behavior_vars_distance_less_one}.\ref{prop:immidiate_reactive_behavior_vars_distance_less_one:bisim_functional}.
\end{proof}
Moreover, if $\lambda<1$ the implications in both cases also hold in the other direction.

\begin{remark} \label{rem:agreementDisagreementActionDistance}
The bisimulation distance $\bisimddisc(t,t')$ between terms $t$ and $t'$ is in $[0,\lambda] \cup \{1\}$.
If $\lambda \in (0,1)$, then: 
\begin{enumerate}
\item
$\bisimddisc(t,t') = 1$ iff $t$ can perform an action which $t'$ cannot (or vice versa), i.e.\ $\der(t,a) \neq\emptyset$ and $\der(t',a)=\emptyset$ for some action $a \in \Act$;
\item
$\bisimddisc(t,t') = 0$ iff $t$ and $t'$ have the same reactive behavior (are bisimilar); and
\item
$\bisimddisc(t,t') \in (0,\lambda]$ iff $t$ and $t'$ have the same set of initial moves, i.e.\ $\der(t,a) = \der(t',a)$, and have different reactive behavior after performing the same initial actions.
\end{enumerate}
Notice that in the first case the discount $\lambda$ does not apply since the different behaviors are observed immediately. 
If $\lambda = 1$ then the first and last case collapse, i.e.\
$\bisimddisc(t,t') = 0$ iff $t$ and $t'$ have the same reactive behavior (are bisimilar), and
$\bisimddisc(t,t') \in (0,1]$ iff $t$ and $t'$ have different reactive behavior. 
\end{remark}

\subsubsection{Properties of the Kantorovich lifting}

The Kantorovich pseudometric satisfies important properties that will be essential to prove our technical results.
In detail, 
the Kantorovich lifting functional is monotone, 
the Dirac operator is an isometric embedding of the metric space of states into the metric space of distributions,
and probabilistic choice distributes over the Kantorovich lifting.

\begin{prop}[\cite{Pan09}] \label{prop:kantorovich_lifting}
Let $d$ and $d'$ be any $1$-bounded pseudometrics. Then
\begin{enumerate}
	\item \label{prop:kantorovich_lifting:monotonicity} 
		$\Kantorovich(d) \sqsubseteq \Kantorovich(d')$ if $d \sqsubseteq d'$;
	\item \label{prop:kantorovich_lifting:dirac}
		$\Kantorovich(d)(\delta(t),\delta(t')) = d(t,t')$ for all $t,t'\in \closedSTerms$;
	\item \label{prop:kantorovich_lifting:splitting}
		$\Kantorovich(d)(\sum_{i \in I} p_i \pi_i, \sum_{i \in I} p_i \pi'_i) \le \sum_{i \in I} p_i \cdot \Kantorovich(d)(\pi_i,\pi'_i)$ for all $\pi_i, \pi'_i \in\Delta(\closedSTerms)$ and $p_i\in[0,1]$ with $\sum_{i \in I} p_i=1$.
\end{enumerate}
\end{prop}\smallskip

\noindent Now we will show a very important new result stating that the Kantorovich lifting preserves concave moduli of continuity of language operators. In other words, moduli of continuity of language operators distribute over probabilistic choices.

\begin{thm}\label{thm:Kantorovich_lifting}
Let $d \colon \closedSTerms \times \closedSTerms \to [0,1]$ be any $1$-bounded pseudometric. Assume an $n$-ary operator $f \in \Sigma$ and a concave\footnote{A function $z \colon [0,1]^n \to [0,1]$ is called concave if, for any $x_1,\ldots,x_n,y_1,\ldots,y_n \in [0,1]$ and any $\lambda \in [0,1]$,
$z((1-\lambda)x_1+\lambda y_1,\ldots,(1-\lambda)x_n+\lambda y_n) \ge (1-\lambda)z(x_1,\ldots,x_n)+\lambda z(y_1,\ldots,y_n)$.} function $z \colon [0,1]^n \to [0,1]$ with
\[ 	
	d(f(t_1,\ldots,t_n),f(t_1',\ldots,t_n'))
		\le
	z(d(t_1,t_1'),\ldots,d(t_n,t_n'))	
\]
for all terms $t_1,t_1',\ldots,t_n,t_n' \in \closedSTerms$. Then we have
\[
	\Kantorovich(d)(f(\pi_1,\ldots,\pi_{n}),f(\pi'_1,\ldots,\pi'_{n}))
		\le
	z(\Kantorovich(d)(\pi_1,\pi'_1),\ldots,\Kantorovich(d)(\pi_n,\pi'_n))
\]
for all probability distributions $\pi_1,\pi'_1,\ldots,\pi_n,\pi'_n \in \Delta(\closedSTerms)$.
\end{thm}
\begin{proof}
We assume $\omega_i \in \Omega(\pi_i,\pi'_i)$ to be an optimal matching such that
$\Kantorovich(d)(\pi_i,\pi'_i) = \sum_{t,t' \in \closedSTerms} d(t,t') \cdot \omega_i(t,t')$, i.e.\ a matching between $\pi_i$ and $\pi_i'$ which yields the Kantorovich distance $\Kantorovich(d)(\pi_i,\pi'_i)$. We define a new distribution over the product space 
$\omega \in \Delta(\closedSTerms \times \closedSTerms)$ by
\[
	\omega(f(t_1,\ldots,t_n),f(t_1',\ldots,t_n')) = \prod_{i=1}^n \omega_i(t_i,t'_i)
\]
for all $t_1,t'_1,\ldots,t_n,t'_n \in \closedTerms$.
First, we show that $\omega$ is a joint probability distribution with left marginal $f(\pi_1,\dots,\pi_{n})$ and right marginal $f(\pi'_1,\dots,\pi'_{n})$.
The left marginal is
\begin{align*}
  & \sum_{t' \in \closedSTerms} \omega(f(t_1,\ldots,t_n),t') \\
= & \sum_{t_1',\ldots,t_n' \in \closedSTerms} \omega(f(t_1,\ldots,t_n),f(t_1',\ldots,t_n')) \\
= & \sum_{t_1',\ldots,t_n' \in \closedSTerms} \prod_{i=1}^{n} \omega_i(t_i,t'_i) \\
= & \prod_{i=1}^{n} \sum_{t'_i \in \closedSTerms} \omega_i(t_i,t'_i) \\
= & \prod_{i=1}^{n} \pi_i(t_i) \\
= & f(\pi_1,\ldots,\pi_{n})(f(t_1,\ldots,t_n))
\end{align*}
with $\sum_{t_1',\ldots,t_n' \in \closedSTerms} \prod_{i=1}^{n} \omega_i(t_i,t'_i) 
= \prod_{i=1}^{n} \sum_{t'_i \in \closedSTerms} \omega_i(t_i,t'_i)$ by induction over $n$ with induction step
\begin{align*}
	& \sum_{t_1',\ldots,t_{n+1}' \in \closedSTerms} \prod_{i=1}^{n+1} \omega_i(t_i,t'_i) \\
= & \sum_{t_1',\ldots,t_{n}' \in \closedSTerms} \sum_{t_{n+1}' \in \closedSTerms} 
	\omega_{n+1}(t_{n+1},t_{n+1}') \prod_{i=1}^{n} \omega_i(t_i,t'_i) \\
= & \sum_{t_{n+1}' \in \closedSTerms} \omega_{n+1}(t_{n+1},t_{n+1}') \sum_{t_1',\ldots,t_{n}' \in \closedSTerms} \prod_{i=1}^{n} \omega_i(t_i,t'_i) \\
= & \sum_{t_{n+1}' \in \closedSTerms} \omega_{n+1}(t_{n+1},t_{n+1}') \prod_{i=1}^{n} \sum_{t_i' \in \closedSTerms} \omega_i(t_i,t'_i) \\
= &\prod_{i=1}^{n+1} \sum_{t'_i \in \closedSTerms} \omega_i(t_i,t'_i).
\end{align*}
The right marginal is computed analogously. Hence, $\omega \in \Omega(f(\pi_1,\dots,\pi_{n}), f(\pi'_1,\dots,\pi'_{n}))$, i.e.\ $\omega$ is a matching for distributions $f(\pi_1,\dots,\pi_{n})$ and $f(\pi'_1,\dots,\pi'_{n})$.

The proof obligation can be derived now by
\begin{align*}
 & \ \Kantorovich(d)(f(\pi_1,\dots,\pi_{n}),f(\pi'_1,\dots,\pi'_{n})) \\
\le &\  \sum_{{t_1,\ldots,t_n \atop t_1',\ldots,t_n'} \in \closedSTerms} d(f(t_1,\ldots,t_n),f(t_1',\ldots,t_n')) \cdot \omega(f(t_1,\ldots,t_n),f(t_1',\ldots,t_n')) \\
= &\ \sum_{{t_1,\ldots,t_n \atop t_1',\ldots,t_n'} \in \closedSTerms}  d(f(t_1,\dots,t_{n}),f(t'_1,\dots,t'_{n})) \cdot  \prod_{i=1}^n \omega_i(t_i,t'_i) \\
\le &\ \sum_{{t_1,\ldots,t_n \atop t_1',\ldots,t_n'} \in \closedSTerms} z(d(t_1,t_1'),\ldots,d(t_n,t_n')) 
\cdot  \prod_{i=1}^n \omega_i(t_i,t'_i) \\ 
\le & \ z\left(
	\sum_{{t_1,\ldots,t_n \atop t_1',\ldots,t_n'} \in \closedSTerms} (d(t_1,t_1'),\ldots,d(t_n,t_n')) \cdot  \prod_{i=1}^n \omega_i(t_i,t'_i)
           \right) \\
= & \ z\left( 
	\sum_{{t_1,\ldots,t_n \atop t_1',\ldots,t_n'} \in \closedSTerms} 
		\left(d(t_1,t_1')\cdot  \prod_{i=1}^n \omega_i(t_i,t'_i),\ldots,d(t_n,t_n')\cdot  \prod_{i=1}^n \omega_i(t_i,t'_i)\right) 
    \right) \\
= & \ z\left( \left(
		\sum_{{t_1,\ldots,t_n \atop t_1',\ldots,t_n'}\in \closedSTerms} d(t_1,t_1')\cdot  \prod_{i=1}^n \omega_i(t_i,t'_i),\ldots,\sum_{{t_1,\ldots,t_n \atop t_1',\ldots,t_n'} \in \closedSTerms} d(t_n,t_n') \cdot \prod_{i=1}^n \omega_i(t_i,t'_i) 
    \right) \right) \\
= & \ z \left( \left( \sum_{t_1, t'_1 \in \closedSTerms}  d(t_1,t'_1) \omega_1(t_1,t'_1),  \ldots , \sum_{t_n,t'_n \in \closedSTerms} d(t_n,t_n') \omega_n(t_n,t'_n) \right) \right)
\\
= &\ z(\Kantorovich(d)(\pi_1,\pi'_1),\ldots,\Kantorovich(d)(\pi_n,\pi'_n))
\end{align*}
whereby the reasoning steps are derived as follows:
step 1 from the fact that $\omega$ is a matching for distributions $f(\pi_1,\dots,\pi_{n})$ and $f(\pi'_1,\dots,\pi'_{n})$, 
step 2 by the definition of $\omega$, 
step 3 by the assumption $d(f(t_1,\ldots,t_n),f(t_1',\ldots,t_n')) \le z(d(t_1,t_1'),\ldots,d(t_n,t_n'))$, 
step 4 by using Jensen's inequality for the concave function $z$, 
step 7 by $\sum_{{t_1,\ldots,t_n \atop t_1',\ldots,t_n'}\in \closedSTerms} d(t_1,t_1')\cdot  \prod_{i=1}^n \omega_i(t_i,t'_i) = \sum_{t_1, t'_1 \in \closedSTerms}  d(t_1,t'_1) \omega_1(t_1,t'_1)$, and
step 8 by the definition of $\Kantorovich$.
\end{proof}

\subsection{PGSOS Specifications}

We will specify the operational semantics of operators by SOS rules in the probabilistic GSOS format~\cite{Bar04,LGD12,DGL15}. 
The probabilistic GSOS format, PGSOS format for short, is the quantitative generalization of the classical nondeterministic GSOS format~\cite{BIM95}.
It is more general than earlier formats \cite{LT05,LT09} which consider transitions of the form $t \trans[a,q] t'$ modeling that term $t$ reaches through action $a$ the term $t'$ with probability $q$.
The probabilistic GSOS format allows us to specify probabilistic nondeterministic process algebras, such as probabilistic CCS~\cite{JYL01,Bar04,DD07}, probabilistic CSP~\cite{JYL01,Bar04,DGHMZ07b,DL12} and probabilistic ACP~\cite{And99,And02}.

\begin{defi}[PGSOS rule, \cite{Bar04,LGD12}]\label{def:pgsos}
A \emph{PGSOS rule} $r$ has the form:
\[
	\SOSrule{\{ x_i \trans[a_{i,k}] \mu_{i,k} \mid i \in I, k \in K_i \} \qquad
			 \{ x_i \ntrans[b_{i,l}] \mid i \in I, l \in L_i \}}
			{ f(x_1,\ldots,x_n) \trans[a] \theta}
\]
with $f \in \Sigma$ an operator with rank $n$,
$I = \{1,\dots,n\}$ indices for the arguments of $f$,
$K_i,L_i$ finite index sets, 
$a_{i,k},b_{i,l}, a \in \Act$ actions,
$x_i \in \SVar$ state variables, 
$\mu_{i,k} \in \DVar$ distribution variables, and 
$\theta\in\openDTerms$ a distribution term. Furthermore, the following constraints need to be satisfied:
\begin{enumerate}
	\item \label{item:pgsos:cond_muik_different}
		all $\mu_{i,k}$ for $i \in I, k \in K_i$ are pairwise different; 
	\item \label{item:pgsos:cond_xi_different} 
		all $x_1,\ldots,x_n$ are pairwise different;
	\item \label{item:target_no_free_vars} 
		$\VarTerm(\theta) \subseteq \{\mu_{i,k} \mid i\in I, k \in K_i \} \cup \{ x_1 \ldots,x_n\}$.
\end{enumerate}
\end{defi}
The PGSOS constraints~\ref{item:pgsos:cond_muik_different}--\ref{item:target_no_free_vars} are precisely the constraints of the nondeterministic GSOS format~\cite{BIM95} where the variables in the right-hand side of the literals are replaced by distribution variables.

\begin{notation}[Notations for rules]
Let $r$ be a PGSOS rule.
The expressions $x_i \trans[a_{i,k}] \mu_{i,k}$, $x_i \ntrans[b_{i,l}]$ and $f(x_1,\dots,x_n) \trans[a] \theta$ are called, resp., \emph{positive premises}, \emph{negative premises} and \emph{conclusion}. 
The set of all premises is denoted by $\prem{r}$ and the conclusion by $\conc{r}$.
The term $f(x_1,\dots,x_n)$ is called the \emph{source},
the variables $x_1,\dots,x_n$ are called \emph{source variables},
and
the distribution term $\theta$ is called the \emph{target}.

Given a set of rules $R$ 
we denote by $R_f$ the rules specifying operator $f$, i.e.\ all rules of $R$ with source $f(x_1,\dots,x_n)$, 
and  by $R_{f,a}$ the rules specifying an $a$-labelled transition for operator $f$, i.e.\ all rules of $R_f$ with a conclusion that is $a$-labelled.
\end{notation}

\begin{defi}[PTSS]
A \emph{probabilistic transition system specification} (PTSS) in PGSOS format is a triple $P = (\Sigma, \Act, R)$, where 
\begin{itemize}
	\item $\Sigma$ is a signature, 
	\item $\Act$ is a countable set of actions, 
	\item $R$ is a countable set of PGSOS rules, and 
	\item $R_{f,a}$ is finite for all $f \in \Sigma$ and $a \in \Act$.
\end{itemize}
\end{defi}
The last property ensures that the supported model (Defintion~\ref{def:supported_model}) is image-finite such that the fixed point characterization of bisimulation metrics coincides with the coinductive characterization (Proposition~\ref{prop:bisim_metric_lfp_D}).

The operational semantics of terms is given by inductively applying the respective PGSOS rules. 
Then, a supported model of a PTSS describes the operational semantics of all terms.
In other words, a supported model of a PGSOS specification $P$ is a PTS $M$ with transition relation $\trans$ such that $\trans$ contains all and only those transitions for which the rules of $P$ offer a justification.

\begin{defi}[Supported transition]
Let $P = (\Sigma, \Act, R)$ be a PTSS and $r \in R$ be a rule. 
Given a PTS $M=(\closedSTerms,\Act,{\trans})$ and a closed substitution $\sigma$, we say that the $\sigma$-instance of $r$ is \emph{satisfied} in $M$ and allows to derive $t \trans[a] \pi$, formally $M \models_r^\sigma t \trans[a] \pi$, if
\begin{itemize}
	\item $\sigma(x_i) \trans[a_{i,k}] \sigma(\mu_{i,k}) \in {\trans}$ for all $x_i \trans[a_{i,k}] \mu_{i,k} \in \prem{r}$,
	\item $\sigma(x_i) \trans[b_{i,l}] \pi \not\in {\trans}$ for any $\pi \in \Delta(\closedSTerms)$, for all $x_i \ntrans[b_{i,l}] \in \prem{r}$, and
	\item $t \trans[a] \pi \in {\trans}$ for $t \trans[a] \pi = \sigma(\conc{r})$.
\end{itemize}
We call a transition $t \trans[a] \pi$ in $M$ \emph{supported} by $P$, notation $M \models_P t \trans[a] \pi$, if there is some $r \in R$ and a closed substitution $\sigma$ such that $M \models_r^\sigma t \trans[a] \pi$.
\end{defi}

The supported transitions of a PTSS $P$ form the supported model of $P$.

\begin{defi}[Supported model] \label{def:supported_model}
Let $P = (\Sigma, \Act, R)$ be a PTSS. A PTS $M=(\closedSTerms,\Act,{\trans})$ is a \emph{supported model} if 
\[
	t \trans[a] \pi \text{ iff } M \models_P t \trans[a] \pi
\]
for all ${t \trans[a] \pi} \in {\trans}$.
\end{defi}
Each PTSS in PGSOS format has a supported model which is moreover unique~\cite{BIM95,Bar04}.
We call the single supported PTS of a PTSS $P$ also the \emph{induced model} of $P$.

Intuitively, a term $f(t_1,\ldots,t_n)$ represents the composition of terms $t_1,\ldots,t_n$ by  operator $f$. A rule $r$ specifies some transition $f(t_1,\ldots,t_n) \trans[a] \pi$ that represents the evolution of the composed term $f(t_1,\ldots,t_n)$ by action $a$ to the distribution $\pi$.

\begin{defi}[Disjoint extension~\cite{ABV94}] \label{def:disjoint_extension}
Let $P_1 = (\Sigma_1, \Act, R_1)$ and $P_2 = (\Sigma_2, \Act, R_2)$ be two PGSOS PTSSs. $P_2$ is a \emph{disjoint extension} of $P_1$, notation $P_1 \dext P_2$, iff $\Sigma_1 \subseteq \Sigma_2$, $R_1 \subseteq R_2$ and $R_2$ introduces no new rule for any operator in $\Sigma_1$. 
\end{defi}

\section{Non-recursive processes} \label{sec:nonrec_procs}

\begin{table}[!t]
\begin{gather*}
	\SOSrule{}
	 		{\varepsilon \trans[\tick] \delta(0)}
	\qquad
	\SOSrule{}
	 		{a.\bigoplus_{i=1}^n[p_i]x_i \trans[a] \sum_{i=1}^n p_i \delta(x_i)}
\\[1ex]
	\SOSrule{x \trans[a] \mu \quad\ a\neq\tick}
			{x; y \trans[a] \mu;\delta(y)}
	\qquad
	\SOSrule{x\trans[\tick]\mu \quad y\trans[a]\nu}
			{{x; y}\trans[a]\nu}
	\qquad
	\SOSrule{x \trans[a] \mu}
			{x + y \trans[a] \mu}
	\qquad
	\SOSrule{y \trans[a] \nu}
			{x + y \trans[a] \nu}
\\[1ex]
	\SOSrule{x \trans[a] \mu \quad y \trans[a] \nu \quad a \neq \tick}
			{x \mid y \trans[a] \mu \mid \nu}
	\qquad
                \SOSrule{x \trans[\tick] \mu \quad y \trans[\tick] \nu}
			{x \mid y \trans[\tick] \delta(0)}
\\[1ex]
	\SOSrule{x\trans[a]\mu \quad\ a\neq\tick}
			{x \mid\mid\mid y \trans[a] \mu \mid\mid\mid \delta(y)}
	\qquad
	\SOSrule{y\trans[a]\nu \quad\ a\neq\tick}
			{x \mid\mid\mid y \trans[a] \delta(x) \mid\mid\mid \nu}
	\qquad
	\SOSrule{x\trans[\tick]\mu \quad y\trans[\tick]\nu}
			{x \mid\mid\mid y \trans[\tick] \delta(0)}
\\[1ex]
	\SOSrule{x\trans[a]\mu \quad y\trans[a]\nu \quad a\in{B\setminus\{\tick\}}}
			{{x\,\parop{B}\,y}\trans[a] \mu\,\parop{B}\,\nu}
	\qquad
                 \SOSrule{x\trans[\tick]\mu \quad y\trans[\tick]\nu}
			{{x\,\parop{B}\,y}\trans[\tick] \delta(0)}
\\[1ex]
                \SOSrule{x\trans[a]\mu \quad a\notin B\cup\{\tick\}}
			{{x\,\parop{B}\,y}\trans[a] \mu\,\parop{B}\,\delta(y)}
                \qquad
                \SOSrule{y\trans[a]\nu \quad a\notin B\cup\{\tick\}}
			{{x\,\parop{B}\,y}\trans[a] \delta(x)\,\parop{B}\,\nu}
\end{gather*}
\caption{Standard non-recursive process combinators}
\label{tab:prob_algebra_nonprob_process_combinators}
\end{table}

We start by discussing compositional reasoning over probabilistic processes that are composed by non-recursive process combinators. First we introduce the most common non-recursive process combinators, then study the distance between processes composed by these combinators, and conclude by analyzing their compositionality properties. Our study of compositionality properties generalizes earlier results of~\cite{DGJP04,DCPP06} which considered only a small set of process combinators and only the compositionality property of non-expansiveness. The development of tight bounds on the distance between composed processes (necessary for effective metric assume-guarantee performance validation) is novel.

\subsection{Non-recursive process combinators} \label{sec:nonrec_probabilistic_process_algebras}

We introduce now a probabilistic process algebra that comprises many of the probabilistic process combinators from CCS~\cite{JYL01,Bar04,DD07} and CSP~\cite{JYL01,Bar04,DGHMZ07b,DL12}. 
Assume a set of actions $\Act$, with $\tick \in \Act$ denoting the successful termination action.
Let $\nonrecPASig$ be the signature with the following operators:
\begin{itemize}
	\item constants $0$ (stop process)
                 and $\varepsilon$ (skip process);
	\item a family of $n$-ary probabilistic prefix operators $a.([p_1]\_ \oplus \ldots \oplus [p_n]\_ )$ with 
		$a \in \Act$, 
		$n \ge 1$, 
		$p_1,\ldots,p_n \in (0,1]$ and 
		$\sum_{i=1}^n p_i = 1$;
\item binary operators 
\begin{itemize}
\item	    $\_ \,; \_$ (sequential composition), 
\item		$\_ + \_$ (alternative composition),
\item		$\_ +_p \_$ (probabilistic alternative composition), with $p \in (0,1)$,
\item	    $\_ \mid \_$ (synchronous parallel composition), 
\item	    $\_ \mid\mid\mid \_$ (asynchronous parallel composition), 
\item	    $\_ \mid\mid\mid_p \_$ (probabilistic parallel composition), with $p \in (0,1)$, and
\item	    $\_ \parallel_B \_$ for each for each $B \subseteq \Act$ (CSP-like parallel composition).
\end{itemize}
\end{itemize}
The PTSS $\nonrecPASpec = (\nonrecPASig, \Act, \nonrecPARules)$ is given by the set of PGSOS rules $\nonrecPARules$ in Table~\ref{tab:prob_algebra_nonprob_process_combinators} and 
Table~\ref{tab:prob_algebra_prob_process_combinators}.

The probabilistic prefix operator expresses that the process $a.([p_1]t_1 \oplus \ldots \oplus [p_n]t_n)$ can perform action $a$ and evolves to process $t_i$ with probability $p_i$. Sometimes we write $a.\bigoplus_{i=1}^n[p_i]t_i$ for $a.([p_1]t_1\oplus \ldots \oplus [p_n]t_n)$ and $a.t$ for $a.([1]t)$ (deterministic prefix operator). 
The sequential composition and the alternative composition are as usual. The synchronous parallel composition $t \mid t'$ describes the simultaneous evolution of processes $t$ and $t'$, while the asynchronous parallel composition $t \mid\mid\mid t'$ describes the interleaving of $t$ and $t'$ where both processes can progress by alternating at any rate the execution of their actions.  
The CSP-like parallel composition $t \parallel_B t'$ describes multi-party synchronization where $t$ and $t'$ synchronize on actions in $B$ and evolve independently for all other actions.

The probabilistic variants of the alternative composition and the asynchronous parallel composition replace the nondeterministic choice of their non-probabilistic variant by a probabilistic choice.
The probabilistic alternative composition $t +_p t'$ evolves to the probabilistic choice between a distribution reached by $t$ (with probability $p$) and a distribution reached by $t'$ (with probability $1-p$) for actions which can be performed by both processes. For actions that can be performed by either only $t$ or only $t'$, the probabilistic alternative composition $t +_p t'$ behaves just like the nondeterministic alternative composition $t + t'$. 
Similarly, the probabilistic parallel composition $t \mid\mid\mid_p t'$ evolves to a probabilistic choice (with respectively the probability $p$ and $1-p$) between the two nondeterministic choices of the nondeterministic parallel composition $t \mid\mid\mid t'$ for actions which can be performed by both $t$ and $t'$.
For actions that can be performed by either only $t$ or only $t'$, the probabilistic parallel composition $t \mid\mid\mid_p t'$ behaves just like the nondeterministic parallel composition $t \mid\mid\mid t'$. 

\begin{table}[!t]
\begin{gather*}
	\SOSrule{x \trans[a] \mu \quad y \ntrans[a]}
			{x +_p y \trans[a] \mu}
	\qquad
	\SOSrule{x \ntrans[a] \quad y \trans[a] \nu}
			{x +_p y \trans[a] \nu}
	\qquad
	\SOSrule{x \trans[a] \mu \quad y \trans[a] \nu}
			{x +_p y \trans[a] \mu \oplus_p \nu}
\\[2ex]
	\SOSrule{x\trans[a]\mu \quad y\ntrans[a] \quad\ a\neq \tick}
	 		{x \mid\mid\mid_p y \trans[a] \mu \mid\mid\mid_p \delta(y)}
	\qquad
	\SOSrule{x\ntrans[a] \quad y\trans[a]\nu \quad\ a\neq\tick}
			{x \mid\mid\mid_p y \trans[a] \delta(x) \mid\mid\mid_p \nu}
\\[2ex]
	\SOSrule{x\trans[a] \mu \quad y\trans[a]\nu \quad\ a\neq\tick}
			{x \mid\mid\mid_p y \trans[a] \mu \mid\mid\mid_p \delta(y) \oplus_p \delta(x) \mid\mid\mid_p \nu}
	\qquad
	\SOSrule{x\trans[\tick] \mu \quad y\trans[\tick]\nu}
			{x \mid\mid\mid_p y \trans[\tick] \delta(0)}
\end{gather*}
\caption{Standard non-recursive probabilistic process combinators}
\label{tab:prob_algebra_prob_process_combinators}
\end{table}

\subsection{Distance between processes combined by non-recursive process combinators} \label{sec:distance_nonrec_procs}

We develop now tight bounds on the distance between processes combined by the non-recursive process combinators presented in Table~\ref{tab:prob_algebra_nonprob_process_combinators} and 
Table~\ref{tab:prob_algebra_prob_process_combinators}. This will allow us to derive the compositionality properties of those operators. 
As we will discuss two different compositionality properties for non-recursive process combinators (non-extensiveness, Definition~\ref{def:non_extensiveness}, and non-expansiveness, Definition~\ref{def:non_expansiveness}), we split in this section the discussion on the distance bounds accordingly. 
We use disjoint extensions of the specification of the process combinators in order to reason over the composition of arbitrary processes.

We will express the bound on the distance between composed processes $f(s_1,\dots,s_n)$ and $f(t_1,\dots,t_n)$ in terms of the distance between their respective components $s_i$ and $t_i$. Intuitively, given a probabilistic process $f(s_1,\dots,s_n)$ we provide a bound on the distance to the respective probabilistic process $f(t_1,\dots,t_n)$ where each component $s_i$ is replaced by the component $t_i$. 

We start with those process combinators that satisfy the later discussed compositionality property of non-extensiveness (Definition~\ref{def:non_extensiveness}). 

\begin{prop} \label{prop:distance_nonextensively_composed_procs}	
Let $P=(\Sigma,\Act,R)$ be any PTSS with $\nonrecPASpec \dext P$. For all terms
$s_i,t_i \in \closedTerms$ it holds:
\begin{enumerate}[label=\({\alph*}]
	\item \label{prop:distance_nonextensively_composed_procs:action_prefix}
	$\bisimddisc(a.\bigoplus_{i=1}^n[p_i]s_i,a.\bigoplus_{i=1}^n[p_i]t_i) \le \lambda \cdot \sum_{i=1}^n p_i \bisimddisc(s_i,t_i)$;
	\item \label{prop:distance_nonextensively_composed_procs:nondet_alt_comp}
	$\bisimddisc(s_1 + s_2, t_1 + t_2) \le \max(\bisimddisc(s_1,t_1),\bisimddisc(s_2,t_2))$;
	\item \label{prop:distance_nonextensively_composed_procs:prob_alt_comp}
	$\bisimddisc(s_1 +_p s_2, t_1 +_p t_2) \le \max(\bisimddisc(s_1,t_1),\bisimddisc(s_2,t_2))$.
\end{enumerate}
\end{prop}
\proof
First we consider the probabilistic prefix operator (Proposition~\ref{prop:distance_nonextensively_composed_procs}.\ref{prop:distance_nonextensively_composed_procs:action_prefix}).
The only transitions from $a.\bigoplus_{i=1}^n[p_i]s_i$ and $a.\bigoplus_{i=1}^n[p_i]t_i$ are 
$a.\bigoplus_{i=1}^n[p_i]s_i \trans[a] \sum_{i=1}^{n} p_i \delta(s_i)$ and 
$a.\bigoplus_{i=1}^n[p_i]t_i \trans[a] \sum_{i=1}^{n} p_i \delta(t_i)$. 
Hence we need to show that $\lambda \cdot \Kantorovich(\bisimddisc) (\sum_{i=1}^{n} p_i \delta(s_i),\sum_{i=1}^{n} p_i \delta(t_i)) \le \lambda \cdot \sum_{i=1}^n  p_i \bisimddisc(s_i,t_i)$.
This property can be derived by Proposition~\ref{prop:kantorovich_lifting} as follows:
\begin{align*}
	& \Kantorovich(\bisimddisc) \left(\sum_{i=1}^{n} p_i \delta(s_i),\sum_{i=1}^{n} p_i \delta(t_i)\right) \\
\le & \sum_{i=1}^n p_i \Kantorovich(\bisimddisc)(\delta(s_i),\delta(t_i)) & \text{(Proposition~\ref{prop:kantorovich_lifting}.\ref{prop:kantorovich_lifting:splitting})} \\
= & \sum_{i=1}^n p_i \bisimddisc(s_i,t_i) & \text{(Proposition~\ref{prop:kantorovich_lifting}.\ref{prop:kantorovich_lifting:dirac})} \\
\end{align*}

We proceed with the alternative composition operator (Proposition~\ref{prop:distance_nonextensively_composed_procs}.\ref{prop:distance_nonextensively_composed_procs:nondet_alt_comp}).
If either $\bisimddisc(s_1,t_1)=1$ or $\bisimddisc(s_2,t_2)=1$ then the statement is trivial since $\bisimddisc$ is a $1$-bounded pseudometric. Hence, we assume $\bisimddisc(s_1,t_1)<1$ and $\bisimddisc(s_2,t_2)<1$. 
We consider now the two different rules specifying the alternative composition operator and show that in each case whenever $s_1 + s_2 \trans[a] \pi$ is derivable by some of the rules then there is a transition $t_1 + t_2 \trans[a] \pi'$ derivable by the same rule s.t. $\lambda \cdot \Kantorovich(\bisimddisc)(\pi,\pi') \le  \max(\bisimddisc(s_1,t_1),\bisimddisc(s_2,t_2))$.
\begin{enumerate}
	\item Assume that $s_1 + s_2 \trans[a] \pi$ is derived from $s_1 \trans[a] \pi$. Since $\bisimddisc(s_1,t_1)<1$  and $\bisimd$ satisfies the transfer condition of the bisimulation metrics, there exists a transition $t_1 \trans[a] \pi'$ for a distribution $\pi'$ with $\lambda \cdot \Kantorovich(\bisimddisc)(\pi,\pi') \le \bisimddisc(s_1,t_1) \le \max(\bisimddisc(s_1,t_1),\bisimddisc(s_2,t_2))$. Finally, from $t_1 \trans[a] \pi'$  we derive $t_1 + t_2 \trans[a] \pi'$.
	\item Assume that $s_1 + s_2 \trans[a] \pi$ is derived from $s_2 \trans[a] \pi$. The argument is  the same of the previous case.
\end{enumerate}

\noindent We conclude with the probabilistic alternative composition operator (Proposition~\ref{prop:distance_nonextensively_composed_procs}.\ref{prop:distance_nonextensively_composed_procs:prob_alt_comp}).
If either $\bisimddisc(s_1,t_1)=1$ or $\bisimddisc(s_2,t_2)=1$ then the statement is trivial since $\bisimddisc$ is a $1$-bounded pseudometric. Hence, we assume $\bisimddisc(s_1,t_1)<1$ and $\bisimddisc(s_2,t_2)<1$.  We consider now the three different rules specifying the probabilistic alternative composition operator and show that in each case whenever $s_1 + s_2 \trans[a] \pi$ is derivable by some of the rules then there is a transition $t_1 + t_2 \trans[a] \pi'$ derivable by the same rule s.t. $\lambda \cdot \Kantorovich(\bisimddisc)(\pi,\pi') \le \max(\bisimddisc(s_1,t_1),\bisimddisc(s_2,t_2))$.
\begin{enumerate}
	\item 
Assume that $s_1 +_p s_2 \trans[a] \pi$ is derived from $s_1 \trans[a] \pi$ and $s_2 \ntrans[a]$. Since $\bisimddisc(s_1,t_1)<1$ and $\bisimd$ satisfies 
the transfer condition of the bisimulation metrics, there exists a transition $t_1 \trans[a] \pi'$ with $\lambda \cdot \Kantorovich(\bisimddisc)(\pi,\pi') \le \bisimddisc(s_1,t_1) \le \max(\bisimddisc(s_1,t_1),\bisimddisc(s_2,t_2))$.
Since $\bisimddisc(s_2,t_2) < 1$, by Proposition~\ref{prop:immidiate_reactive_behavior_vars_distance_less_one}.\ref{prop:immidiate_reactive_behavior_vars_distance_less_one:bisim_metric} the processes $s_2$ and $t_2$ agree on the actions they can perform immediately. Thus $t_2 \ntrans[a]$. 
Hence we can derive the transition $t_1 +_p t_2 \trans[a] \pi'$. 
	\item Assume that $s_1 +_p s_2 \trans[a] \pi$ is derived from $s_1 \ntrans[a]$ and $s_2 \trans[a] \pi$. The argument is the same of the previous case.
	\item Assume that $s_1 +_p s_2 \trans[a] \pi$ with $\pi = p(\pi_1) + (1-p)\pi_2$ is derived from $s_1 \trans[a] \pi_1$ and $s_2 \trans[a] \pi_2$.
Then, since $\bisimddisc(s_1,t_1)<1$ and $\bisimddisc(s_2,t_2)<1$ and $\bisimd$ satisfies the transfer condition of the bisimulation metrics, there exist transitions $t_1 \trans[a] \pi'_1$ with $\lambda \cdot \Kantorovich(\bisimddisc)(\pi_1,\pi_1') \le \bisimddisc(s_1,t_1)$ and $t_2 \trans[a] \pi'_2$ with $\lambda \cdot \Kantorovich(\bisimddisc)(\pi_2,\pi_2') \le \bisimddisc(s_2,t_2)$.  Therefore we derive $t_1 +_p t_2 \trans[a] p\pi'_1 + (1-p) \pi'_2$, with
\begin{align*}
	& \lambda \cdot \Kantorovich(\bisimddisc)(p\pi_1 + (1-p)\pi_2, p\pi_1' + (1-p)\pi_2') \\
\le & \lambda \cdot (p\Kantorovich(\bisimddisc)(\pi_1,\pi_1') + (1-p)\Kantorovich(\bisimddisc)(\pi_2,\pi_2')) & \text{(Proposition~\ref{prop:kantorovich_lifting}.\ref{prop:kantorovich_lifting:splitting})} \\
\le & \lambda \cdot \max(\Kantorovich(\bisimddisc)(\pi_1,\pi_1'), \Kantorovich(\bisimddisc)(\pi_2,\pi_2')) \\
\le & \max(\bisimddisc(s_1,t_1),\bisimddisc(s_2,t_2)). 
\rlap{\hbox to 271 pt{\hfill\qEd}}
\end{align*}
\end{enumerate}

\noindent We note that the distance between action prefixed processes (Proposition~\ref{prop:distance_nonextensively_composed_procs}.\ref{prop:distance_nonextensively_composed_procs:action_prefix}) is discounted by $\lambda$ since the processes $a.\bigoplus_{i=1}^n[p_i]s_i$ and $a.\bigoplus_{i=1}^n[p_i]t_i$ perform first the action $a$ before the processes $s_i$ and $t_i$ may evolve and their distance is observed.
The distances between processes composed by either the nondeterministic alternative composition operator or by the probabilistic alternative composition operator are both bounded by the maximum of the distances between their respective arguments (Propositions~\ref{prop:distance_nonextensively_composed_procs}.\ref{prop:distance_nonextensively_composed_procs:nondet_alt_comp} and~\ref{prop:distance_nonextensively_composed_procs}.\ref{prop:distance_nonextensively_composed_procs:prob_alt_comp}). The distance bounds for these operators coincide since the first two rules specifying the probabilistic alternative composition define the same operational behavior as the nondeterministic alternative composition and the third rule defining a convex combination of these transitions applies only for those actions that can be performed by both processes $s_1$ and $s_2$ and resp.\ $t_1$ and $t_2$.
If the probabilistic alternative composition would be defined by only the third rule of Table~\ref{tab:prob_algebra_prob_process_combinators}, then $\bisimddisc(s_1 +_p s_2, t_1 +_p t_2) \le p\bisimddisc(s_1,t_1) + (1-p)\bisimddisc(s_2,t_2)$.

Finally, we note that the processes $s_i$ and $t_i$ in Propositions~\ref{prop:distance_nonextensively_composed_procs} are obtained by using arbitrary operators in $\Sigma$ (not necessarily only operators in $\nonrecPASig$).

We proceed with those process combinators that satisfy the later discussed compositionality property of non-expansiveness (Definition~\ref{def:non_expansiveness}). 

\begin{prop} \label{prop:distance_nonexpansively_composed_procs}
Let $P=(\Sigma,\Act,R)$ be any PTSS with $\nonrecPASpec \dext P$. For all terms
$s_i,t_i \in \closedTerms$ it holds:
\begin{enumerate}[label=\({\alph*}]
	\item \label{prop:distance_nonexpansively_composed_procs:sequential_comp} 
		$ \bisimddisc(s_1 ; s_2, t_1 ; t_2) \le 
                                      \begin{cases}
			1 & \text{if } \bisimddisc(s_1,t_1)=1 \\
			\max(d^1_{1,2},\bisimddisc(s_2,t_2)) & \text{if } \bisimddisc(s_1,t_1) \in [0,1)
		\end{cases}$ \\
	\item \label{prop:distance_nonexpansively_composed_procs:sync_parallel_comp} 
		$\bisimddisc(s_1 \mid s_2, t_1 \mid t_2) \le d^s$ 
	\item \label{prop:distance_nonexpansively_composed_procs:async_parallel_comp}  
		$\bisimddisc(s_1 \mid\mid\mid s_2, t_1 \mid\mid\mid t_2) \le d^a$
	\item \label{prop:distance_nonexpansively_composed_procs:csp_parallel_comp}  
		$\bisimddisc(s_1 \parallel_B s_2, t_1 \parallel_B t_2) \le
		\begin{cases}
            d^s & \text{if } B \setminus \{\tick\} \neq \emptyset \\
			d^a & \text{otherwise}	
		\end{cases}$ \\[0ex]
            \item \label{prop:distance_nonexpansively_composed_procs:prob_async_parallel}  
			$\bisimddisc(s_1 \mid\mid\mid_p s_2, t_1 \mid\mid\mid_p t_2) \le d^a$
\end{enumerate}
with
\begin{align*}
	d^s &= 
	\begin{cases}
				1 & \text{if } \bisimddisc(s_1,t_1) = 1 \\
				1 & \text{if } \bisimddisc(s_2,t_2) = 1 \\
				d^0_{1,2} & \text{otherwise}
	\end{cases}\\[1.0ex]
	d^a &= 
	\begin{cases}
				1 & \text{if } \bisimddisc(s_1,t_1) = 1 \\
				1 & \text{if } \bisimddisc(s_2,t_2) = 1 \\
				\max(d^2_{1,2} \,,\, d^2_{2,1}) & \text{otherwise}
	\end{cases}\\[1.0ex]
	d^n_{1,2} &= \bisimddisc(s_1,t_1) + \lambda^n (1 - \bisimddisc(s_1,t_1)/\lambda)\bisimddisc(s_2,t_2) \\
	d^n_{2,1} &= \bisimddisc(s_2,t_2) + \lambda^n (1 - \bisimddisc(s_2,t_2)/\lambda)\bisimddisc(s_1,t_1) \\
\end{align*}
\end{prop}
\begin{proof}
We will prove only Proposition~\ref{prop:distance_nonexpansively_composed_procs}.\ref{prop:distance_nonexpansively_composed_procs:csp_parallel_comp} (CSP-like parallel composition $\parallel_B$). 
The synchronous and asynchronous parallel composition operators (Propositions~\ref{prop:distance_nonexpansively_composed_procs}.\ref{prop:distance_nonexpansively_composed_procs:sync_parallel_comp} and~\ref{prop:distance_nonexpansively_composed_procs}.\ref{prop:distance_nonexpansively_composed_procs:async_parallel_comp}) are special cases, since $\mid$ coincides with $\parallel_{\Act}$ and $\mid\mid\mid$ coincides with $\parallel_{\emptyset}$. The proofs for the probabilistic parallel composition operator $\mid\mid\mid_p$ (Proposition~\ref{prop:distance_nonexpansively_composed_procs}.\ref{prop:distance_nonexpansively_composed_procs:prob_async_parallel}) and the sequential composition $;$ (Proposition~\ref{prop:distance_nonexpansively_composed_procs}.\ref{prop:distance_nonexpansively_composed_procs:sequential_comp}) are analogous.

We prove the case $B \setminus \{\tick\} \neq \emptyset$ (the case $B \setminus \{\tick\} = \emptyset$ is similar).
First we need to introduce the notion of congruence closure for $\lambda$-bisimilarity metric $\bisimd$ as the quantitative analogue of the well-known concept of congruence closure of a process equivalence. 
We define the metric congruence closure of $\bisimddisc$ for operator $\parallel_B$ w.r.t.\ the bound provided in Proposition~\ref{prop:distance_nonexpansively_composed_procs}.\ref{prop:distance_nonexpansively_composed_procs:csp_parallel_comp} as a function $d\colon \closedTerms \times \closedTerms \to [0,1]$ defined by
\[
	d(t,t')=
	\begin{cases}
		 \min(\lambda[1-(1-d(t_1,t'_1)/\lambda)(1-d(t_2,t'_2)/\lambda)], \bisimddisc(t,t')) & \text{if } 
			\left[\begin{array}{l}
				t=t_1\parallel_B t_{2} \,\land \\
				t'=t'_1 \parallel_B t'_{2} \,\land \\
				\bisimddisc(t_1,t'_1) < 1 \,\land \\
				\bisimddisc(t_2,t'_2) < 1
			\end{array} \right.\\
		\bisimddisc(t,t') & \text{otherwise}
	\end{cases}
\]

We note that $d$ satisfies by construction $d(s_1 \parallel_B s_2, t_1 \parallel_B t_2) \le d^s$ since $\lambda [1-(1-d(s_1,t_1)/\lambda)(1-d(s_2,t_2)/\lambda)]=d(s_1,t_1) + (1 - d(s_1,t_1)/\lambda)d(s_2,t_2)$.
We note also that $d$ satisfies by construction $d \sqsubseteq \bisimd$.
It remains to show that $\bisimddisc \sqsubseteq d$, thus giving $\bisimddisc = d$, and Proposition~\ref{prop:distance_nonexpansively_composed_procs}.\ref{prop:distance_nonexpansively_composed_procs:csp_parallel_comp} holds. 
Since $\bisimddisc$ is the least  prefixed point of $\Bisimulation$, to show $\bisimddisc \sqsubseteq d$ it is enough to prove that $d$ is a prefixed point of $\Bisimulation$.

To prove that $\Bisimulation(d) \sqsubseteq d$ we need to show that $d$ satisfies the transfer condition of the bisimulation metrics, namely
\begin{equation}\label{bound_parallel_proof_obligation}
 \text{for all } t \trans[a] \pi \text{ there exists a transition } t' \trans[a] \pi' \text{ with } \lambda \cdot \Kantorovich(d)(\pi,\pi') \le d(t,t')
\end{equation}
for all terms $t,t' \in \closedSTerms$ with $d(t,t') < 1$.

We prove Equation~\ref{bound_parallel_proof_obligation} by induction over the overall number $k$ of occurrences of operator $\parallel_B$ occurring in $t$ and $t'$.

Consider the base case $k=0$. 
By definition of $d$, we have that $d(t,t') = \bisimd(t,t')$. Since $\bisimd(t,t') < 1$ we are sure that the transition $t \trans[a] \pi$ is mimicked by some transition $t' \trans[a] \pi'$ for some distribution $\pi' \in \Delta(\closedSTerms)$ such that $\lambda \cdot \Kantorovich(\bisimd)(\pi,\pi') \le \bisimd(t,t')$. By Proposition~\ref{prop:kantorovich_lifting}.\ref{prop:kantorovich_lifting:monotonicity} from $d \sqsubseteq \bisimd$ we infer $\Kantorovich(d) \sqsubseteq \Kantorovich(\bisimd)$. Therefore we conclude
\[
\lambda \cdot \Kantorovich(d)(\pi,\pi') \le \lambda \cdot \Kantorovich(\bisimd)(\pi,\pi') \le \bisimd(t,t') = d(t,t')
\]
which confirms that Equation~\ref{bound_parallel_proof_obligation} holds for $t$ and $t'$.

Consider the inductive step $k > 0$.
If either $t$ is not of the form $t = t_1 \parallel_B t_2$, or $t'$ is not of the form $t '= t'_1 \parallel_B t'_2$, then by definition of $d$ we have $d(t,t') = \bisimd(t,t')$ and Equation~\ref{bound_parallel_proof_obligation} follows precisely as in the base case $k=0$.
If both $t = t_1 \parallel_B t_2$ and $t '= t'_1 \parallel_B t'_2$, then we distinguish two cases, namely $d(t,t') = \bisimd(t,t')$ (either $\bisimd(t_1,t'_1)=1$ or $\bisimd(t_2,t'_2)=1$ or $\bisimd(t,t') < \lambda[1-(1-d(t_1,t'_1)/\lambda)(1-d(t_2,t'_2)/\lambda)]$) and $d(t,t') = \lambda[1-(1-d(t_1,t'_1)/\lambda)(1-d(t_2,t'_2)/\lambda)]$ (both $\bisimd(t_1,t'_1)<1$ and $\bisimd(t_2,t'_2)<1$ and $\bisimd(t,t') \ge \lambda[1-(1-d(t_1,t'_1)/\lambda)(1-d(t_2,t'_2)/\lambda)]$). In case $d(t,t') = \bisimd(t,t')$ Equation~\ref{bound_parallel_proof_obligation} follows precisely as in the base case $k=0$. Consider the case $d(t,t') = \lambda[1-(1-d(t_1,t'_1)/\lambda)(1-d(t_2,t'_2)/\lambda)]$. We have four different subcases:
\begin{enumerate}
	\item $t_1 \trans[a] \pi_1$, $t_2 \trans[a] \pi_2$, $a \in B \setminus \{\tick\}$ and $\pi = \pi_1 \parallel_B \pi_2$;
	\item $t_1 \trans[a] \pi_1$, $t_2 \ntrans[a]$, $a \not\in B \cup \{\tick\}$ and $\pi = \pi_1 \parallel_B \delta(t_2)$;
	\item $t_2 \trans[a] \pi_2$, $t_1 \ntrans[a]$, $a \not\in B \cup \{\tick\}$ and $\pi = \delta(t_1) \parallel_B \pi_2$;
                \item $t_1 \trans[a] \pi_1$, $t_2 \trans[a] \pi_2$, $a = \tick$ and $\pi = \delta(0)$.
\end{enumerate}

We start with the first case. 
By $\bisimd(t_1,t'_1)<1$ and $\bisimd(t_2,t'_2) < 1$  and $d \sqsubseteq \bisimd$, we get $d(t_1,t'_1)<1$ and $d(t_2,t'_2) < 1$.
By the inductive hypothesis we get that there are also transitions $t'_1 \trans[a] \pi'_1$ and $t'_2 \trans[a] \pi'_2$ with $\lambda \cdot \Kantorovich(d)(\pi_1,\pi'_1) \le d(t_1,t'_1)$ and $\lambda \cdot \Kantorovich(d)(\pi_2,\pi'_2) \le d(t_2,t'_2)$. 
Hence, there is also the transition $t_1' \parallel_B t'_2 \trans[a] \pi'_1 \parallel_B \pi'_2$.
Then 
\begin{align*}
& \lambda \cdot \Kantorovich(d)(\pi_1 \parallel_B \pi_2, \pi'_1 \parallel_B \pi'_2) \\
\le & \lambda^2[1 - (1-\Kantorovich(d)(\pi_1,\pi'_1)/\lambda)(1- \Kantorovich(d)(\pi_2,\pi'_2)/\lambda) ] \\
\le & \lambda^2[1 - (1-d(t_1,t'_1)/\lambda^2)(1- d(t_2,t'_2)/\lambda^2) ] \\
\le & \lambda[1 - (1-d(t_1,t'_1)/\lambda)(1- d(t_2,t'_2)/\lambda) ] \\
= & d(t_1 \parallel_B t_2, t'_1 \parallel_B t'_2)
\end{align*}
with the first step by Theorem~\ref{thm:Kantorovich_lifting} (using the fact that the candidate modulus of continuity of operator $\parallel_B$ given by $z(\epsilon_1,\epsilon_2)=\lambda[1 - (1-\epsilon_1/\lambda)(1-\epsilon_2/\lambda)]$ is concave) and the second step by the inductive hypothesis $\lambda \cdot \Kantorovich(d)(\pi_i,\pi'_i) \le d(t_i,t'_i)$. 
Thus, the metric bisimulation transfer condition (Equation~\ref{bound_parallel_proof_obligation}) is satisfied for $d$ in this case. 

Consider now the second case. By $\bisimd(t_1,t'_1)<1$ and $d \sqsubseteq \bisimd$, we get $d(t_1,t'_1)<1$.
By the inductive hypothesis we get that there is also a transitions $t'_1 \trans[a] \pi'_1$ with $\lambda \cdot \Kantorovich(d)(\pi_1,\pi'_1) \le d(t_1,t'_1)$.
By Proposition~\ref{prop:immidiate_reactive_behavior_vars_distance_less_one}.\ref{prop:immidiate_reactive_behavior_vars_distance_less_one:bisim_metric} we have that $t'_2 \ntrans[a]$, therefore we can derive the transition
$t_1' \parallel_B t'_2 \trans[a] \pi'_1 \parallel_B \delta(t_2')$.
Then 
\begin{align*}
& \lambda \cdot \Kantorovich(d)(\pi_1 \parallel_B \delta(t_2), \pi'_1 \parallel_B \delta(t_2')) \\
\le & \lambda^2 [1 - (1-\Kantorovich(d)(\pi_1,\pi'_1)/\lambda)(1- \Kantorovich(d)(\delta(t_2),\delta(t'_2))/\lambda)] \\
\le & \lambda^2 [1 - (1-d(t_1,t'_1)/\lambda^2)(1-d(t_2,t'_2)/\lambda)] \\
\le & \lambda [1 - (1-d(t_1,t'_1)/\lambda)(1-d(t_2,t'_2)/\lambda)] \\
= & d(t_1 \parallel_B t_2, t'_1 \parallel_B t'_2)
\end{align*}
with step 1 again from Theorem~\ref{thm:Kantorovich_lifting} like in the first case and the second step by the inductive hypothesis $\lambda \cdot \Kantorovich(d)(\pi_1,\pi'_1) \le d(t_1,t'_1)$ and Proposition~\ref{prop:kantorovich_lifting}.\ref{prop:kantorovich_lifting:dirac}. 
Hence, the metric bisimulation transfer condition (Equation~\ref{bound_parallel_proof_obligation}) is satisfied for $d$ in this case.

The third case is analogous to the second one.

Consider now the fourth case.
By $\bisimd(t_1,t'_1)<1$ and $\bisimd(t_2,t'_2) < 1$  and $d \sqsubseteq \bisimd$, we get $d(t_1,t'_1)<1$ and $d(t_2,t'_2) < 1$.
By the inductive hypothesis we get that there are also transitions $t'_1 \trans[\tick] \pi'_1$ and $t'_2 \trans[\tick] \pi'_2$. 
Hence, there is also the transition $t_1' \parallel_B t'_2 \trans[\tick] \delta(0)$.
Then $\lambda \cdot \Kantorovich(d)(\delta(0), \delta(0)) = 0 \le d(t_1 \parallel_B t_2, t'_1 \parallel_B t'_2)$.
Thus, the metric bisimulation transfer condition (Equation~\ref{bound_parallel_proof_obligation}) is satisfied for $d$ also in this case. 
\end{proof}

The expression $d^s$ in Proposition~\ref{prop:distance_nonexpansively_composed_procs} captures the distance bound between the synchronously evolving processes $s_1$ and $s_2$ on the one hand and the synchronously evolving processes $t_1$ and $t_2$ on the other hand. We remark that the distances $\bisimddisc(s_1,t_1)$ and $\bisimddisc(s_2,t_2)$ contribute symmetrically to $d^s$ since
$d^0_{1,2} = \bisimddisc(s_1,t_1) + (1 - \bisimddisc(s_1,t_1)/\lambda)\bisimddisc(s_2,t_2) =
\bisimddisc(s_2,t_2) + (1 - \bisimddisc(s_2,t_2)/\lambda)\bisimddisc(s_1,t_1) = d^0_{2,1}$. 
The expressions $d^n_{1,2}, d^n_{2,1}$ with $n > 0$ cover different scenarios of the asynchronous evolution of those processes. The expression $d^n_{1,2}$ (resp.\ $d^n_{2,1}$) denotes the distance bound between the asynchronously evolving processes $s_1$ and $s_2$ on the one hand and the asynchronously evolving processes $t_1$ and $t_2$ on the other hand, at which the first $n$ transitions are performed by the processes $s_1$ and $t_1$ (resp.\ the first $n$ transitions are performed by processes $s_2$ and $t_2$).

If $\bisimddisc(s_1,t_1)=1$ or $\bisimddisc(s_2,t_2)=1$, then the processes $s_1$ and $t_1$ and the processes $s_2$ and $t_2$ may disagree on the initial actions they can perform, and also the composed processes may disagree on their initial actions and have then also the maximal distance of $1$ (cf.\ Proposition~\ref{prop:immidiate_reactive_behavior_vars_distance_less_one} and Remark~\ref{rem:agreementDisagreementActionDistance}).
We analyze the bound for the process combinators in details assuming both $\bisimddisc(s_1,t_1)<1$ and $\bisimddisc(s_2,t_2)<1$.

The distance between the sequentially composed processes $s_1 ; s_2$ and $t_1 ; t_2$ (Proposition~\ref{prop:distance_nonexpansively_composed_procs}.\ref{prop:distance_nonexpansively_composed_procs:sequential_comp}) is given if $\bisimddisc(s_1,t_1) \in [0,1)$ as the maximum of
\begin{enumerate}[label=(\roman*)]
	\item \label{case:sequential_comp_case1}
		distance $d^1_{1,2} = \bisimddisc(s_1,t_1) + \lambda(1 - \bisimddisc(s_1,t_1)/\lambda)\bisimddisc(s_2,t_2)$, which captures the case that first the processes $s_1$ and $t_1$ evolve followed by $s_2$ and $t_2$, and
	\item \label{case:sequential_comp_case2} 
		distance $\bisimddisc(s_2,t_2)$, which captures the case that the processes $s_2$ and $t_2$ evolve immediately because both $s_1$ and $t_1$ terminate successfully at their first computation step.
\end{enumerate}	
The distance $d^1_{1,2}$ weights the distance $\bisimddisc(s_2,t_2)$ between $s_2$ and $t_2$ by $\lambda (1-\bisimddisc(s_1,t_1)/\lambda)$. 
The discount $\lambda$ expresses that processes $s_2$ and $t_2$ are delayed by at least one transition step whenever $s_1$ and $t_1$ perform at least one transition step before terminating.
Additionally, note that the difference between $s_2$ and $t_2$ can only be observed when $s_1$ and $t_1$ agree to terminate. When processes $s_1$ and $t_1$ evolve by one step, they disagree by  $\bisimddisc(s_1,t_1)/\lambda$ on their behavior. Hence they agree by $(1-\bisimddisc(s_1,t_1)/\lambda)$. Thus, the distance between processes $s_2$ and $t_2$ needs to be additionally weighted by $(1-\bisimddisc(s_1,t_1)/\lambda)$. In case~(\ref{case:sequential_comp_case2}) the distance between $s_2$ and $t_2$ is not discounted since both processes start immediately.

The distance bound between synchronous parallel composed processes $s_1 \mid s_2$ and $t_1 \mid t_2$ (Proposition~\ref{prop:distance_nonexpansively_composed_procs}.\ref{prop:distance_nonexpansively_composed_procs:sync_parallel_comp}) is the expression $d^s$, which is
$d^0_{1,2} = \bisimddisc(s_1,t_1) + (1 - \bisimddisc(s_1,t_1)/\lambda)\bisimddisc(s_2,t_2) =
 \bisimddisc(s_2,t_2) + (1 - \bisimddisc(s_2,t_2)/\lambda)\bisimddisc(s_1,t_1) = d^0_{2,1}$, when both $\bisimd(s_1,t_1) < 1$ and $\bisimd(s_2,t_2) < 1$.
Hence the distance between $s_1 \mid s_2$ and $t_1 \mid t_2$ is bounded by the sum of the distance between $s_1$ and $t_1$, which is the degree of dissimilarity between $s_1$ and $t_1$, and the distance between $s_2$ and $t_2$ weighted by the probability that $s_1$ and $t_1$ agree on their behavior, which is the degree of dissimilarity between $s_2$ and $t_2$ under equal behavior of $s_1$ and $t_1$. 
Alternatively, by $d^0_{1,2} = d^0_{2,1} = \lambda(1-(1 - \bisimddisc(s_1,t_1)/\lambda)(1 - \bisimddisc(s_2,t_2)/\lambda))$, the bound to the distance between $s_1 \mid s_2$ and $t_1 \mid t_2$ can be  understood as composing processes on the behavior they agree upon, i.e.\ $s_1 \mid s_2$ and $t_1 \mid t_2$ agree on their behavior if $s_1$ and $t_1$ agree (probability of similarity $1 - \bisimddisc(s_1,t_1)/\lambda$) and if $s_2$ and $t_2$ agree (probability of similarity $1 - \bisimddisc(s_2,t_2)/\lambda$). The resulting distance is then the probability of dissimilarity of the respective behavior $1-(1 - \bisimddisc(s_1,t_1)/\lambda)(1 -\bisimddisc(s_2,t_2)/\lambda)$ multiplied by the discount factor $\lambda$.

The distance bound between asynchronous parallel composed processes $s_1 \mid\mid\mid s_2$ and $t_1 \mid\mid\mid t_2$ is the expression $d^a$ (Proposition~\ref{prop:distance_nonexpansively_composed_procs}.\ref{prop:distance_nonexpansively_composed_procs:async_parallel_comp}).
Hence the distance bound is the maximum of $d^2_{1,2}$, namely the distance observable when first processes $s_1$ and $t_1$ evolve by at least two transition steps and then $s_2$ and $t_2$, and  $d^2_{2,1}$, namely the distance observable when first processes $s_2$ and $t_2$ evolve by at least two transition steps and then $s_1$ and $t_1$.
Notice that at least two transition steps by the faster processes are necessary to observe their distance before the slower processes start.
The behaviors where either $s_1$ and $t_1$ perform the first transition step and $s_2$ and $t_2$ perform the second transition step, or $s_2$ and $t_2$ perform the first transition step and $s_1$ and $t_1$ perform the second transition step, give rise to a lower distance wrt. that expressed by the maximum between $d^2_{1,2}$ and $d^2_{2,1}$. The reason is that the observation of the different behaviors is delayed by more transition steps and, therefore, more discounted.
Notice that both $d^2_{1,2}$ and $d^2_{2,1}$ differ from the distance $d^s$ of the synchronously evolving processes $s_1 \mid s_2$ and $t_1 \mid t_2$ only by the discount factor $\lambda^2$ that is applied to the distance of the delayed processes. 
Moreover, $d^2_{1,2}$ differs from the distance $d^1_{1,2}$ of the sequential composed processes $s_1 ; s_2$ and $t_1 ; t_2$ by the different discount factor that is applied to the distance of the  processes $s_2$ and $t_2$.
The discount factor in case $d^2_{1,2}$ is $\lambda^2$ since $s_2$ and $t_2$ are delayed by at least two transition steps after the distance between $s_1$ and $t_1$ is observed, whereas the discount factor in case $d^1_{1,2}$ is $\lambda$ since the distance between $s_1$ and $t_1$ observed at their second transition step may be realized by the ability/inability of performing action $\tick$, which let $s_2$ and $t_2$ start immediately (namely already in this second transition step).

Processes that are composed by the CSP-like parallel composition operator $\_ \parallel_B \_$ evolve synchronously for actions in $B \setminus \{\tick\}$, evolve asynchronously for actions in $\Act \setminus (B \cup \{\tick\})$, and the action $\tick$ leads always to the stop process if both processes can perform $\tick$. Since $d^s \ge d^a$, the distance between processes $s_1 \mid\mid\mid s_2$ and $t_1 \mid\mid\mid t_2$ (Proposition~\ref{prop:distance_nonexpansively_composed_procs}.\ref{prop:distance_nonexpansively_composed_procs:csp_parallel_comp}) is bounded by $d^s$ if there is at least one action $a \in B$ with $a \neq \tick$ for which the composed processes can evolve synchronously, and otherwise by $d^a$.

The distance between processes composed by the probabilistic parallel composition operator $s_1 \mid\mid\mid_p s_2$ and $t_1 \mid\mid\mid_p t_2$ (Proposition~\ref{prop:distance_nonexpansively_composed_procs}.\ref{prop:distance_nonexpansively_composed_procs:prob_async_parallel}) is bounded by 
the expression $d^a$ since the first two rules specifying the probabilistic parallel composition define the same operational behavior as the nondeterministic parallel composition, and the third rule defining a convex combination of these transitions applies only for those actions that can be performed by both processes $s_1$ and $s_2$ and resp.\ $t_1$ and $t_2$.

The distance bounds on the distance between processes composed by non-recursive process combinators 
(Proposition~\ref{prop:distance_nonextensively_composed_procs} and~\ref{prop:distance_nonexpansively_composed_procs}) are tight.

\begin{prop} \label{prop:optimality_distances_nonext_nonexp_composed_procs}
Let $\epsilon_i \in [0,1]$. There are processes $s_i,t_i \in \T(\nonrecPASig)$ with $\bisimddisc(s_i,t_i)=\epsilon_i$ such that the inequalities in Propositions~\ref{prop:distance_nonextensively_composed_procs} and~\ref{prop:distance_nonexpansively_composed_procs} become equalities.
\end{prop}
\begin{proof}
We start with Proposition~\ref{prop:distance_nonextensively_composed_procs}. 
Let $\Act = \{a_1, \ldots, a_n \} \cup \{\tick\}$. 
We define now the witness processes
\begin{itemize}
	\item $s_i = t_i = a_i.\varepsilon$, if $\epsilon_i = 0$;
	\item $s_i = a_i.([1-\epsilon_i/\lambda]\varepsilon\oplus [\epsilon_i/\lambda]0)$ and $t_i = a_i.\varepsilon$, if $\epsilon_i \in (0,\lambda)$;
	\item $s_i = a_i.0$ and $t_i = a_i.\varepsilon$, if $\epsilon_i = \lambda < 1$;
	\item $s_i = 0$ and $t_i = a_i.\varepsilon$, if $\epsilon_i = 1$.
\end{itemize}
It is easy to see that these processes yield for all process combinators of Proposition~\ref{prop:distance_nonextensively_composed_procs} exactly the stated upper bound.

We proceed now with Propositions~\ref{prop:distance_nonexpansively_composed_procs}.\ref{prop:distance_nonexpansively_composed_procs:sequential_comp}, \ref{prop:distance_nonexpansively_composed_procs}.\ref{prop:distance_nonexpansively_composed_procs:sync_parallel_comp} and \ref{prop:distance_nonexpansively_composed_procs}.\ref{prop:distance_nonexpansively_composed_procs:csp_parallel_comp}, case $B \setminus \{\tick\} \neq \emptyset$.
Let $\Act = \{a,\tick\}$ with $a \in B$. We define now the witness processes 
\begin{itemize}
	\item $s_i = t_i = a.\varepsilon$, if $\epsilon_i = 0$;
	\item $s_i = a.([1-\epsilon_i/\lambda]\varepsilon \oplus [\epsilon_i/\lambda]0)$ and $t_i = a.\varepsilon$, if $\epsilon_i \in (0,\lambda)$;
	\item $s_i = a.0$ and $t_i = a.\varepsilon$, if $\epsilon_i = \lambda<1$;
	\item $s_i = 0$ and $t_i = a.\varepsilon$ if $\epsilon_i = 1$.
\end{itemize} 
These processes yield for all process combinators of Propositions~\ref{prop:distance_nonexpansively_composed_procs}.\ref{prop:distance_nonexpansively_composed_procs:sequential_comp}, \ref{prop:distance_nonexpansively_composed_procs}.\ref{prop:distance_nonexpansively_composed_procs:sync_parallel_comp} and \ref{prop:distance_nonexpansively_composed_procs}.\ref{prop:distance_nonexpansively_composed_procs:csp_parallel_comp}, case $B \setminus \{\tick\} \neq \emptyset$, exactly the stated upper bound.

Finally, we conclude with Propositions~\ref{prop:distance_nonexpansively_composed_procs}.\ref{prop:distance_nonexpansively_composed_procs:async_parallel_comp}, \ref{prop:distance_nonexpansively_composed_procs}.\ref{prop:distance_nonexpansively_composed_procs:prob_async_parallel} and \ref{prop:distance_nonexpansively_composed_procs}.\ref{prop:distance_nonexpansively_composed_procs:csp_parallel_comp}, case $B \setminus \{\tick\} = \emptyset$.
Let $\Act = \{a_1,a_2,a\} \cup \{\tick\}$. We define now the witness processes
\begin{itemize}
	\item $s_i = t_i = a_i.a.0$, if $\epsilon_i = 0$;
	\item $s_i = a_i.([1-\epsilon_i/\lambda]a.0\oplus [\epsilon_i/\lambda]0)$ and $t_i = a_i.a.0$, if $\epsilon_i \in (0,\lambda)$;
	\item $s_i = a_i.0$ and $t_i = a_i.a.0$, if $\epsilon_i = \lambda < 1$;
	\item $s_i = 0$ and $t_i = a_i.\varepsilon$, if $\epsilon_i = 1$.
\end{itemize}
These processes yield for all process combinators of Propositions~\ref{prop:distance_nonexpansively_composed_procs}.\ref{prop:distance_nonexpansively_composed_procs:async_parallel_comp}, \ref{prop:distance_nonexpansively_composed_procs}.\ref{prop:distance_nonexpansively_composed_procs:prob_async_parallel} and \ref{prop:distance_nonexpansively_composed_procs}.\ref{prop:distance_nonexpansively_composed_procs:csp_parallel_comp}, case $B \setminus \{\tick\} = \emptyset$, exactly the stated upper bound.
\end{proof}

\subsection{Compositional reasoning over non-recursive processes} \label{sec:comp_reasoning_nonrec_procs}

In order to specify and verify systems in a compositional manner, it is necessary that the behavioral semantics is compatible with all operators of the language that describe these systems.
There are multiple proposals which properties of process combinators facilitate compositional reasoning. In this section we discuss non-extensiveness~\cite{BBLM13b} and non-expansiveness~\cite{DJGP02,DGJP04,DCPP06,CGPX14}),
which are compositionality properties based on the $p$-norm. They allow for compositional reasoning over probabilistic processes that are built of non-recursive process combinators. Non-extensiveness and non-expansiveness are very strong forms of uniform continuity. For instance, a non-expansive operator ensures that the distance between the composed processes is at most the sum of the distances between their parts. 
Later in Section~\ref{sec:recProcs:compReasoning} we will propose uniform continuity as generalization of these properties 
that allows also for compositional reasoning over recursive processes.

\begin{defi}[Non-extensive process combinator]\label{def:non_extensiveness}
A process combinator $f \in \Sigma$ is \emph{non-extensive} w.r.t.\ $\lambda$-bisimilarity metric $\bisimddisc$ if 
\[ 
	\bisimddisc(f(s_1,\dots,s_n), f(t_1,\dots,t_n)) 
		\le
	\max_{i=1}^n \bisimddisc(s_i,t_i) 
\]
for all closed process terms $s_i,t_i \in \closedSTerms$.
\end{defi}

Probabilistic action prefix, nondeterministic alternative composition, and probabilistic alternative composition are non-extensive  w.r.t.\ $\bisimddisc$.

\begin{thm}\label{thm:non_extensiveness}
The process combinators 
\begin{itemize}
	\item probabilistic action prefix $a.\bigoplus_{i=1}^n[p_i]\_$\ 
	\item nondeterministic alternative composition $\_ + \_$\ 
	\item probabilistic alternative composition $\_ +_p \_$
\end{itemize}
are non-extensive w.r.t.\ $\lambda$-bisimilarity metric $\bisimddisc$ \ for any $\lambda\in(0,1]$.
\end{thm}
\begin{proof}
Follows directly from Proposition~\ref{prop:distance_nonextensively_composed_procs}.
\end{proof}

All other operators of $\nonrecPASig$ are not non-extensive.

\begin{prop} \label{prop:not_non_extensive}
None of the process combinators
\begin{itemize}
	\item sequential composition $\_ \,; \_$
	\item synchronous parallel composition $\_ \mid \_$
	\item asynchronous parallel composition $\_ \mid\mid\mid \_$
	\item CSP-like parallel composition $\_ \parallel_B \_$
	\item probabilistic parallel composition $\_ \mid\mid\mid_p \_$
\end{itemize}
is non-extensive w.r.t.\ $\lambda$-bisimilarity metric $\bisimddisc$ \ for any $\lambda\in(0,1]$.
\end{prop}
\proof
Follows directly from Propositions~\ref{prop:distance_nonexpansively_composed_procs} and~\ref{prop:optimality_distances_nonext_nonexp_composed_procs}.
\qed

We proceed now with the compositionality property of non-expansiveness.

\begin{defi}[Non-expansive process combinator]\label{def:non_expansiveness}
A process combinator $f \in \Sigma$ is \emph{non-expansive} w.r.t.\ $\lambda$-bisimilarity metric $\bisimddisc$ if 
\[ 
	\bisimddisc(f(s_1,\dots,s_n), f(t_1,\dots,t_n)) 
		\le  
	\sum_{i=1}^n \bisimddisc(s_i,t_i) 
\]
for all closed process terms $s_i,t_i \in \closedSTerms$.
\end{defi}

It is clear that if a process combinator $f$ is non-extensive, then $f$ is non-expansive.
Moreover, the two notions coincide when $f$ is unary.

\begin{thm}\label{thm:non_expansiveness}
All non-recursive process combinators of $\nonrecPASig$

are non-expansive w.r.t.\ $\lambda$-bisimilarity metric $\bisimddisc$ \ for any $\lambda\in(0,1]$.
\end{thm}
\begin{proof}
Follows directly from Propositions~\ref{prop:distance_nonextensively_composed_procs} and~\ref{prop:distance_nonexpansively_composed_procs} and the observation that 
$d^a , d^1_{1,2} \le d^s \le \bisimd(s_1,t_1) + \bisimd(s_2,t_2)$.
\end{proof}

Theorem~\ref{thm:non_expansiveness} generalizes a similar result of~\cite{DGJP04} which considered only PTSs without nondeterministic branching and only a small set of process combinators. The analysis which operators are non-extensive (Theorem~\ref{thm:non_extensiveness}) and the tight distance bounds (Propositions~\ref{prop:distance_nonextensively_composed_procs}, and~\ref{prop:distance_nonexpansively_composed_procs} and~\ref{prop:optimality_distances_nonext_nonexp_composed_procs}) are novel.

\section{Recursive processes} \label{sec:rec_procs}

Recursion is necessary to express infinite (non-terminating) behavior in terms of finite process expressions. Moreover, recursion allows us to express repetitive finite behavior in a compact way. 
We will discuss now compositional reasoning over probabilistic processes that are composed by recursive process combinators. We will see that the compositionality properties of non-extensiveness and non-expansiveness used for non-recursive process combinators (Section~\ref{sec:comp_reasoning_nonrec_procs}) fall short for recursive process combinators. We will propose the more general property of uniform continuity (Section~\ref{sec:recProcs:compReasoning}) that captures the inherent nature of compositional reasoning over probabilistic processes. In fact, it allows us to reason compositionally over processes that are composed by both recursive and non-recursive process combinators. In the next section we apply these results to reason compositionally over a communication protocol and derive its respective performance properties. To the best of our knowledge this is the first study which explores systematically compositional reasoning over recursive processes in the context of bisimulation metric semantics. 
We remark that recursive process combinators are indispensable for effective modeling and verification of safety critical systems, network protocols, and systems biology. 

\subsection{Recursive process combinators} \label{sec:rec_probabilistic_process_algebras}

We define $\recPASpec$
as disjoint extension of $\nonrecPASpec$ with the following operators: 
\begin{itemize}
	\item finite iteration $\_^n$,
	\item infinite iteration $\_^\omega$,
	\item binary Kleene-star iteration $\_^*\_$,
	\item probabilistic Kleene-star iteration $\_^{*_p}\_$,
	\item finite replication $!^n\_$,
	\item infinite replication (bang) operator $!\_$, and
	\item probabilistic bang operator $!_p\_$.
\end{itemize}
The operational semantics of these operators is specified by the rules in Table~\ref{tab:prob_algebra_rec_process_combinators}. 

The finite iteration $t^n$ (resp.\ infinite iteration $t^\omega$) of process $t$ expresses that $t$ is performed $n$ times (resp.\ infinitely often) in sequel. 
The binary Kleene-star expresses for ${t_1}^*t_2$ that either $t_1$ is performed infinitely often in sequel, or $t_1$ is performed a finite number of times in sequel, followed by $t_2$.
The bang operator expresses for $! t$ (resp.\ finite replication $!^n t$) that infinitely many copies (resp.\ $n$ copies) of $t$ evolve asynchronously. 
The probabilistic Kleene-star iteration~\cite[Section~5.2.4(vi)]{Bar04} expresses that ${t_1}^{*_p}t_2$ evolves to a probabilistic choice (with respectively the probability $p$ and $1-p$) between the two nondeterministic choices of the Kleene star operation ${t_1}^*{t_2}$ for actions which can be performed by both $t_1$ and $t_2$.
For actions that can be performed by either only $t_1$ or only $t_2$, ${t_1}^{*_p}t_2$ behaves just like ${t_1}^*{t_2}$.
The probabilistic bang replication \cite[Fig.~1]{MS13} expresses that $!_pt$ replicates the argument process $t$ with probability $1-p$ and behave just like $t$ with probability $p$.

\begin{table}[!t]
\begin{gather*}
	\SOSrule{x \trans[a] \mu \quad\ a\neq\tick}{x^{n+1} \trans[a] \mu;\delta(x^n)}
	\quad\quad
	\SOSrule{x \trans[\tick] \mu}{x^{n+1} \trans[\tick] \mu}
	\quad\quad 
	\SOSrule{}{x^{0} \trans[\tick] \delta(0)}
	\quad\quad
	\SOSrule{x \trans[\tick] \mu \quad x \trans[a] \nu \quad a \neq \tick \quad n>m}{x^{n} \trans[a] \nu;\delta(x^m)}
\\[2ex]
	\SOSrule{x \trans[a] \mu \quad\ a\neq\tick}{x^{\omega} \trans[a] \mu;\delta(x^\omega)}
	\qquad	
	\SOSrule{x \trans[a] \mu \quad\ a\neq\tick}{x^*y \trans[a] \mu;\delta(x^*y)}
	\qquad
	\SOSrule{y \trans[a] \nu}{x^*y \trans[a] \nu}
\\[2ex]
	\SOSrule{x \trans[a] \mu \ \ \ \ y \trans[a] \nu \ \ \ \ a \neq \tick}
			{x^{*_p}y \trans[a] \nu \oplus_p \mu;\delta(x^{*_p}y)}
	\quad\quad
	\SOSrule{x \trans[a] \mu \ \ \ \ y \ntrans[a] \ \ \ \ a \neq \tick}
			{x^{*_p}y \trans[a] \mu;\delta(x^{*_p} y)}
	\quad\quad
	\SOSrule{x \ntrans[a] \ \ \ \ y \trans[a] \nu \ \ \ \ a \neq \tick }
			{x^{*_p}y \trans[a] \nu}
	\quad\quad
	\SOSrule{y \trans[\tick] \nu}
			{x^{*_p}y \trans[\tick] \nu}			
\\[2ex]
	\SOSrule{x \trans[a] \mu \quad a \neq \tick}{!^{n+1} x \trans[a] \mu \mid\mid\mid \delta(!^n x)}
	\qquad
	\SOSrule{x \trans[\tick] \mu}{!^{n+1} x \trans[\tick] \mu}
	\qquad
	\SOSrule{}{!^{0} x \trans[\tick] \delta(0)}
\\[2ex]
	\SOSrule{x \trans[a] \mu \quad\ a\neq\tick}{! x \trans[a] \mu \mid\mid\mid \delta(! x)}
	\qquad
	\SOSrule{x \trans[a] \mu \quad\ a\neq\tick}{!_p x \trans[a]  \mu \oplus_p (\mu \mid\mid\mid \delta(!_p x))} 
\end{gather*}
\caption{Standard recursive process combinators}
\label{tab:prob_algebra_rec_process_combinators}
\end{table}

\subsection{Distance between processes combined by recursive process combinators} \label{sec:distance_rec_procs}
We develop now tight bounds on the distance between processes combined by the recursive process combinators presented in Table~\ref{tab:prob_algebra_rec_process_combinators}. 

\begin{prop}\label{prop:distance_finite_recursion}
Let $P=(\Sigma,\Act,R)$ be any PTSS with $\recPASpec \dext P$. For all terms $s,s_i,t,t_i \in \closedTerms$ it holds:
\begin{enumerate}[label=\({\alph*}]
	\item \label{prop:distance_finite_recursion:finite_iteration}
		$\bisimddisc(s^n,t^n) \le d^n$
	\item \label{prop:distance_finite_recursion:finite_replication} 
		$\bisimddisc(!^n s,!^n t) \le d^{!^n}$
	\item \label{prop:distance_infinite_recursion:infinite_iteration} 
		$\bisimddisc(s^\omega,t^\omega) \le d^\omega$
	\item \label{prop:distance_infinite_recursion:infinite_replication}
                                  $\bisimddisc(! s,! t) \le d^!$
	\item \label{prop:distance_infinite_recursion:kleene_star}
		$\bisimddisc({s_1}^*s_2,{t_1}^*t_2) \le \max (\bisimddisc({s_1}^\omega,{t_1}^\omega) , \bisimddisc(s_2,t_2) )$\\[-2ex]
	\item \label{prop:distance_probabilistic_kleene}
               $\bisimddisc({s_1}^{*_p} s_2,{t_1}^{*_p} t_2) \le \bisimddisc({s_1}^*s_2,{t_1}^*t_2)$ 
	\item \label{prop:distance_probabilistic_bang}
            $\displaystyle\bisimddisc(!_p s, !_p t) \le \begin{cases}
    \bisimddisc(s,t)\frac{1}{1-(1-p)(\lambda^2 - \lambda\bisimddisc(s,t))} & \text{if } \bisimddisc(s,t) \in (0,1) \\
\bisimddisc(s,t) & \text{if } \bisimddisc(s,t) \in \{0,1\}
\end{cases}$
\end{enumerate}
with
\begin{align*}
	d^n &= 
	\begin{cases} 
		\bisimddisc(s,t) \frac{1-(\lambda - \bisimddisc(s,t))^n}{1-(\lambda - \bisimddisc(s,t))} & \text{if } \bisimddisc(s,t) \in (0,1) \\
		\bisimddisc(s,t) & \text{if } \bisimddisc(s,t) \in \{0,1\}
	\end{cases} \\[1ex]
                d^{!^n} &= 
	\begin{cases} 
		\bisimddisc(s,t) \frac{1-(\lambda^2 - \lambda\bisimddisc(s,t))^n}{1-(\lambda^2 - \lambda\bisimddisc(s,t))} & \text{if } \bisimddisc(s,t) \in (0,1) \\
		\bisimddisc(s,t) & \text{if } \bisimddisc(s,t) \in \{0,1\}
	\end{cases} \\[1ex]
	d^\omega &= 
	\begin{cases} 
	   \bisimddisc(s,t)  \frac{1}{1-(\lambda - \bisimddisc(s,t))} & \text{if } \bisimddisc(s,t) \in (0,1) \\
		\bisimddisc(s,t) & \text{if } \bisimddisc(s,t) \in \{0,1\}
	\end{cases}\\[1ex]
	d^! &= 
	\begin{cases} 
	   \bisimddisc(s,t)  \frac{1}{1-(\lambda^2 - \lambda\bisimddisc(s,t))} & \text{if } \bisimddisc(s,t) \in (0,1) \\
		\bisimddisc(s,t) & \text{if } \bisimddisc(s,t) \in \{0,1\}
	\end{cases}\\[1ex]
\end{align*}
\end{prop}
\begin{proof}
First of all we observe that $\frac{1-(\lambda - \bisimddisc(s,t))^n}{1-(\lambda - \bisimddisc(s,t))} = \sum_{k=0}^{n-1}  (\lambda - \bisimd(s,t))^k$ and $\frac{1-(\lambda^2 - \lambda \bisimddisc(s,t))^n}{1-(\lambda^2 - \lambda\bisimddisc(s,t))} = \sum_{k=0}^{n-1}  (\lambda^2 - \lambda\bisimd(s,t))^k$.

Consider first the finite iteration operator $\_^n$. The cases  $\bisimddisc(s,t)$ = $0$ and $\bisimddisc(s,t) =1$ are immediate. Consider the case $0 < \bisimddisc(s,t) < 1$. The proof obligation can be rewritten as $\bisimddisc(s^n,t^n) \le \bisimddisc(s,t) \sum_{k=0}^{n-1}  (\lambda - \bisimd(s,t))^k$. We reason by induction over $n$. The base case $n=0$ is immediate. Let us consider the inductive step $n+1$. By the rules in Tables \ref{tab:prob_algebra_nonprob_process_combinators}--\ref{tab:prob_algebra_rec_process_combinators}, we infer that $s^{n+1}$ is bisimilar to $s ; s^n$ (i.e.\ they are in bisimulation distance $0$) and that $t^{n+1}$ is bisimilar to $t ; t^n$.
Hence $\bisimddisc(s^{n+1},t^{n+1}) =  \bisimddisc(s;s^{n},t;t^{n})$. By Proposition~\ref{prop:distance_nonexpansively_composed_procs}.\ref{prop:distance_nonexpansively_composed_procs:sequential_comp} we have $\bisimddisc(s;s^{n},t;t^{n}) \le \bisimddisc(s,t) + \bisimddisc(s^n,t^n)(\lambda - \bisimddisc(s,t) )$ = (by the inductive hypothesis over $n$)  $\bisimddisc(s,t) + (\bisimddisc(s,t) \sum_{k=0}^{n-1} (\lambda - \bisimddisc(s,t))^k)(\lambda -\bisimddisc(s,t))$ = $\bisimddisc(s,t)\sum_{k=0}^{n}  (\lambda - \bisimddisc(s,t))^k$. Summarizing, $\bisimddisc(s^{n+1},t^{n+1}) \le\bisimddisc(s,t)\sum_{k=0}^{n}  (\lambda - \bisimddisc(s,t))^k$, thus confirming the thesis.

Consider now the finite replication operator $!^n\_$. The cases $\bisimddisc(s,t) =1$ and $\bisimddisc(s,t) =0$ are immediate. Consider the case $0 < \bisimddisc(s,t) < 1$. The proof obligation can be rewritten as $\bisimddisc(!^ns,!^nt) \le \bisimddisc(s,t) \sum_{k=0}^{n-1}  (\lambda^2 - \lambda\bisimd(s,t))^k$. We reason by induction over $n$. The base case $n=0$ is immediate. Let us consider the inductive step $n+1$. By the rules in Tables \ref{tab:prob_algebra_nonprob_process_combinators}--\ref{tab:prob_algebra_rec_process_combinators}, we infer that $!^{n+1}s$ is bisimilar to $s \mid\mid\mid !^{n}s$ and that $!^{n+1}t$ is bisimilar to $t \mid\mid\mid !^{n}t$.
Hence $\bisimddisc(!^{n+1} s, !^{n+1} t) = \bisimddisc(s \mid\mid\mid !^{n}s, t \mid\mid\mid !^{n}t)$. By Proposition~\ref{prop:distance_nonexpansively_composed_procs}.\ref{prop:distance_nonexpansively_composed_procs:async_parallel_comp} we get $\bisimddisc(s \mid\mid\mid !^{n}s, t \mid\mid\mid !^{n}t)
\le \bisimddisc(s,t) + (\lambda^2 -  \lambda\bisimddisc(s,t))\bisimddisc(!^{n}s,!^{n}t) 
\le \text{ (inductive hypothesis over $n$) } \bisimddisc(s,t) + (\lambda^2 -  \lambda\bisimddisc(s,t))\bisimddisc(s,t) (\sum_{k=0}^{n-1}  (\lambda^2 - \lambda\bisimd(s,t))^k)
= \bisimddisc(s,t) \sum_{k=0}^{n}  (\lambda^2 - \lambda\bisimd(s,t))^k$.
Summarizing, we have $\bisimddisc(!^{n+1}s,!^{n+1}t) \le \bisimddisc(s,t)\sum_{k=0}^{n}  (\lambda^2 - \lambda\bisimddisc(s,t))^k$. This confirms the thesis.

Consider the infinite iteration operator $\_^\omega$. The cases  $\bisimddisc(s,t)=1$ and $\bisimddisc(s,t)=0$ are immediate. Consider the case $0 < \bisimddisc(s,t) < 1$. 
By the rules in Tables~\ref{tab:prob_algebra_nonprob_process_combinators}--\ref{tab:prob_algebra_rec_process_combinators}, we infer that 
$s^{\omega}$ is bisimilar to $s ; s^\omega$ and that $t^{\omega}$ is bisimilar to $t ; t^\omega$. 
Hence $\bisimddisc(s^{\omega},t^{\omega}) =  \bisimddisc(s;s^{\omega},t;t^{\omega})$.
By Proposition \ref{prop:distance_nonexpansively_composed_procs}.\ref{prop:distance_nonexpansively_composed_procs:sequential_comp} we get $\bisimddisc(s;s^{\omega},t;t^{\omega}) \le \bisimddisc(s,t) + (\lambda - \bisimddisc(s,t))\bisimddisc(s^\omega,t^\omega)$. 
Hence we have
 $\bisimddisc(s^\omega,t^\omega) \le \bisimddisc(s,t) + (\lambda - \bisimddisc(s,t))\bisimddisc(s^\omega,t^\omega)$, from which we infer $\bisimddisc(s^\omega,t^\omega)  \le \bisimddisc(s,t)\frac{1}{1-(\lambda - \bisimddisc(s,t))} = d^{\omega}$.

Consider now the bang operator $!\_$. The cases  $\bisimddisc(s,t) =1$ and $\bisimddisc(s,t) =0$ are immediate. Consider the case $0 < \bisimddisc(s,t) < 1$. 
By the rules in Tables~\ref{tab:prob_algebra_nonprob_process_combinators}--\ref{tab:prob_algebra_rec_process_combinators}, we infer that $!s$ is bisimilar to $s \mid\mid\mid !s$ and that $!t$ is bisimilar to $t \mid\mid\mid !t$. Hence $\bisimddisc(!s, !t) = \bisimddisc(s \mid\mid\mid !s, t \mid\mid\mid !t)$. By Proposition~\ref{prop:distance_nonexpansively_composed_procs}.\ref{prop:distance_nonexpansively_composed_procs:async_parallel_comp}  we get  $\bisimddisc(s \mid\mid\mid !s, t \mid\mid\mid !t)  \le \bisimddisc(s,t) + (\lambda^2 - \lambda\bisimddisc(s,t))\bisimddisc(!s,!t)$.
Hence we have $\bisimddisc(!s, !t) \le \bisimddisc(s,t) + (\lambda^2 - \lambda\bisimddisc(s,t))\bisimddisc(!s,!t)$, from which we infer
$\bisimddisc(!s,!t)  \le \bisimddisc(s,t)\frac{1}{1-(\lambda^2 - \lambda \bisimddisc(s,t))} = d^{!}$.

Consider the binary Kleene star operator $\_^*\_$.
Observe that the term ${s_1} ^* s_2$ is bisimilar to $(s_1 ; ({s_1} ^* s_2))+s_2$ and that the term ${t_1}^*t_2$ is bisimilar to $(t_1 ; ({t_1}^*t_2))+t_2$.
Proposition~\ref{prop:distance_nonextensively_composed_procs}.\ref{prop:distance_nonextensively_composed_procs:nondet_alt_comp} shows
$\bisimddisc({s_1}^*s_2,{t_1}^*t_2) 
= \bisimddisc((s_1 ; ({s_1}^*s_2))+s_2,(t_1 ; ({t_1}^*t_2))+t_2 )
= \max\{\bisimddisc((s_1 ; ({s_1}^*s_2)),(t_1 ; ({t_1}^*t_2))), \bisimddisc(s_2,t_2)\}$.
If $\max\{\bisimddisc((s_1 ; ({s_1}^*s_2)),(t_1 ; ({t_1}^*t_2))), \bisimddisc(s_2,t_2)\} = \bisimddisc((s_1 ; ({s_1}^*s_2)),(t_1 ; ({t_1}^*t_2))$, we get $\bisimddisc({s_1}^*s_2,{t_1}^*t_2) = \bisimddisc((s_1 ; ({s_1}^*s_2)),(t_1 ; ({t_1}^*t_2)))$, where, by
Proposition~\ref{prop:distance_nonexpansively_composed_procs},\ref{prop:distance_nonexpansively_composed_procs:sequential_comp},
$\bisimddisc((s_1 ; ({s_1}^*s_2)),(t_1 ; ({t_1}^*t_2)))$
= $\bisimddisc(s_1,t_1) + (\lambda - \bisimddisc(s_1,t_1)) \bisimddisc({s_1}^*s_2,{t_1}^*t_2)$, thus giving
$\bisimddisc({s_1}^*s_2,{t_1}^*t_2) = \bisimddisc(s_1,t_1)  \frac{1}{1-(\lambda - \bisimddisc(s_1,t_1))}$.
Therefore we conclude that
$\bisimddisc({s_1}^*s_2,{t_1}^*t_2)
 = \max\{\bisimddisc(s_1,t_1)  \frac{1}{1-(\lambda - \bisimddisc(s_1,t_1))}, \bisimddisc(s_2,t_2)\}
= \max\{\bisimddisc(s_1^\omega,t_1^\omega) , \bisimddisc(s_2,t_2)\}$. This confirms the thesis.

Consider now the probabilistic Kleene star operator. 
The second, third and fourth rule specifying the probabilistic Kleene star operator define the same operational behavior as the nondeterministic Kleene star operator. Since the target of the first rule for the probabilistic Kleene star operator is a convex combination of the targets of the second and the third rule, the thesis follows. 

Consider now the probabilistic bang operator. The bound on the distance of processes composed by the probabilistic bang operator can be understood by observing that the term $!_p s$ behaves as $!^{n+1} s$ with probability $p(1-p)^n$. Hence, by Proposition~\ref{prop:distance_finite_recursion}.\ref{prop:distance_finite_recursion:finite_replication} we get 
$\bisimddisc(!_p s, !_p t) \le
\sum_{n=0}^{\infty} p(1-p)^n \bisimddisc(!^{n+1}s,!^{n+1}t) \le 
\sum_{n=0}^{\infty} p(1-p)^n d^{!^{n+1}} =
\bisimddisc(s,t)/(1-(1-p)(\lambda^2 - \lambda\bisimddisc(s,t)))$.
\end{proof}

The bounds for the combinators in Proposition~\ref{prop:distance_finite_recursion} are immediate when the distance between the process arguments is either 0 or 1.
We explain those bounds by assuming that the distance between the process arguments is neither 0 nor 1.

First we explain the distance bounds for the nondeterministic recursive process combinators. 
To understand the distance bound between processes that iterate finitely often (Proposition~\ref{prop:distance_finite_recursion}.\ref{prop:distance_finite_recursion:finite_iteration}), observe that $s^n$ and $s;\ldots;s$, with $s;\ldots;s$ denoting $n$ sequentially composed instances of $s$, denote the same PTSs (up to renaming of states). 
Recursive application of the distance bound for operator $\_ \; ; \_$ (Proposition~\ref{prop:distance_nonexpansively_composed_procs}.\ref{prop:distance_nonexpansively_composed_procs:sequential_comp}) 
yields
$\bisimddisc(s^n,t^n) =  
\bisimddisc(s;\ldots;s , t;\ldots;t)
\le  \bisimddisc(s,t) \sum_{k=0}^{n-1}(\lambda - \bisimddisc(s,t))^k
= d^n$.
The same reasoning applies to the finite replication operator (Proposition~\ref{prop:distance_finite_recursion}.\ref{prop:distance_finite_recursion:finite_replication}) by observing that $!^n s$ and $s \mid\mid\mid \ldots \mid\mid\mid s$, with  $s \mid\mid\mid \ldots \mid\mid\mid s$ denoting $n$ occurrences of $s$ that evolve asynchronously, denote the same PTSs (up to renaming of states), thus giving 
$\bisimddisc(!^ns,!^nt) =  
\bisimddisc(s \mid\mid\mid \ldots \mid\mid\mid s , t \mid\mid\mid \ldots \mid\mid\mid t)
\le  \bisimddisc(s,t) \sum_{k=0}^{n-1}(\lambda^2 - \lambda \bisimddisc(s,t))^k
= d^{!^n}$.

The distance between processes that may iterate infinitely many times (Proposition~\ref{prop:distance_finite_recursion}.\ref{prop:distance_infinite_recursion:infinite_iteration}), and the distance between processes that may spawn infinitely many copies that evolve asynchro\-nous\-ly (Proposition~\ref{prop:distance_finite_recursion}.\ref{prop:distance_infinite_recursion:infinite_replication}) are the limit of the respective finite iteration and replication bounds. 
The distance between the Kleene-star iterated processes ${s_1}^*s_2$ and ${t_1}^*t_2$ (Proposition~\ref{prop:distance_finite_recursion}.\ref{prop:distance_infinite_recursion:kleene_star}) is bounded by the maximum of the distance $\bisimddisc({s_1}^\omega,{t_1}^\omega)$ (infinite iteration of $s_1$ and $t_1$ s.t.\ $s_2$ and $t_2$ never evolve), and the distance $\bisimddisc(s_2,t_2)$ ($s_2$ and $t_2$ evolve immediately). The case where $s_1$ and $t_1$ iterate $n$-times and then $s_2$ and $t_2$ evolve leads always to a distance $\bisimddisc({s_1}^n,{t_1}^n) + (\lambda - \bisimddisc(s_1,t_1))^n\bisimddisc(s_2,t_2) \le \max (\bisimddisc({s_1}^\omega,{t_1}^\omega) , \bisimddisc(s_2,t_2) )$.

Now we explain the bounds for the probabilistic recursive process combinators.
The distance between processes composed by the probabilistic Kleene star is bounded by the distance between those processes composed by the nondeterministic Kleene star (Proposition~\ref{prop:distance_finite_recursion}.\ref{prop:distance_probabilistic_kleene}), since the second, the third and the fourth rule specifying the probabilistic Kleene star define the same operational behavior as the nondeterministic Kleene star, and the first rule which defines a convex combination of these transitions applies only for those actions that both of the combined processes can perform. In fact, $\bisimddisc({s_1}^{*_p} s_2,{t_1}^{*_p} t_2) = \bisimddisc({s_1}^{*} s_2,{t_1}^{*} t_2)$ if the initial actions that can be performed by processes $s_1,t_1$ are disjoint from the initial actions that can be performed by processes $s_2,t_2$ (and hence the first rule defining $\_^{*_p}\_$ cannot be applied). Thus, the distance bound of the probabilistic Kleene star coincides with the distance bound of the nondeterministic Kleene star.
The bound on the distance of processes composed by the probabilistic bang operator can be understood by observing that $!_p s$ behaves as $!^{n+1} s$ with probability $p(1-p)^n$. Hence, by Proposition~\ref{prop:distance_finite_recursion}.\ref{prop:distance_finite_recursion:finite_replication} we get 
$\bisimddisc(!_p s, !_p t) \le 
\sum_{n=0}^{\infty} p(1-p)^n \bisimddisc(!^{n+1}s,!^{n+1}t) \le 
\sum_{n=0}^{\infty} p(1-p)^n d^{!^{n+1}} =
\bisimddisc(s,t)/(1-(1-p)(\lambda^2 - \lambda\bisimddisc(s,t)))$.

The distance bounds on the distance between processes composed by recursive process combinators (Proposition~\ref{prop:distance_finite_recursion}) are tight.

\begin{prop} \label{prop:optimality_distances_rec_proc_combinators}
Let $\epsilon_i \in [0,1]$. There are processes $s_i,t_i \in \T(\nonrecPASig)$ with $\bisimddisc(s_i,t_i)=\epsilon_i$ such that the inequalities in Proposition~\ref{prop:distance_finite_recursion} become equalities.
\end{prop}
\proof
The witness processes of Proposition~\ref{prop:optimality_distances_nonext_nonexp_composed_procs} that were used to show that the inequality in 
 Proposition~\ref{prop:distance_nonexpansively_composed_procs}.\ref{prop:distance_nonexpansively_composed_procs:sequential_comp}
becomes an equality,
suffice for Propositions~\ref{prop:distance_finite_recursion}.\ref{prop:distance_finite_recursion:finite_iteration}, \ref{prop:distance_finite_recursion}.\ref{prop:distance_infinite_recursion:infinite_iteration}, \ref{prop:distance_finite_recursion}.\ref{prop:distance_infinite_recursion:kleene_star},  \ref{prop:distance_finite_recursion}.\ref{prop:distance_probabilistic_kleene}.
The witness processes of Proposition~\ref{prop:optimality_distances_nonext_nonexp_composed_procs} that were used to show that the inequality in
Proposition~\ref{prop:distance_nonexpansively_composed_procs}.\ref{prop:distance_nonexpansively_composed_procs:async_parallel_comp}
becomes an equality,
suffice for Propositions~\ref{prop:distance_finite_recursion}.\ref{prop:distance_finite_recursion:finite_replication}, \ref{prop:distance_finite_recursion}.\ref{prop:distance_infinite_recursion:infinite_replication}, \ref{prop:distance_finite_recursion}.\ref{prop:distance_probabilistic_bang}.
\qed

\subsection{Compositional reasoning over recursive processes} \label{sec:recProcs:compReasoning}

From Propositions~\ref{prop:distance_finite_recursion} and~\ref{prop:optimality_distances_rec_proc_combinators} it follows that none of the recursive process combinators discussed in this section satisfies the compositionality property of non-expansiveness. 
\begin{prop}\label{prop:not-non-expansive}
None of the recursive process combinators of $\recPASig$ (unbounded recursion and bounded recursion with $n \ge 2$) is non-expansive w.r.t.\ $\lambda$-bisimilarity metric $\bisimddisc$ \ for any $\lambda\in(0,1]$.
\end{prop}
\begin{proof}
Follows directly from Propositions~\ref{prop:distance_finite_recursion} and~\ref{prop:optimality_distances_rec_proc_combinators} and the observation that $d^{\omega} \ge d^{!} , d^n \ge d^{!^n}  > \bisimd(s,t)$ whenever $0 < \bisimd(s,t) <1$.
\end{proof}

However, a weaker property suffices to facilitate compositional reasoning. To reason compositionally over probabilistic processes it is enough if the distance between the composed processes can be related to the distance between their parts. In essence, compositional reasoning over probabilistic processes is possible whenever a small variance in the behavior of the parts leads to a bounded small variance in the behavior of the composed processes. 

We introduce uniform continuity as the compositionality property for both recursive and non-recursive process combinators. Uniform continuity generalizes the 
properties non-extensiveness and non-expansiveness for non-recursive process combinators. 

\begin{defi}[Uniformly continuous process combinator]\label{def:continuous_operator}
A process combinator $f \in \Sigma$ is \emph{uniformly continuous} w.r.t.\ $\lambda$-bisimilarity metric $\bisimddisc$ if for all $\epsilon>0$ there are $\delta_1,\ldots,\delta_n>0$ such that 
\begin{align*}
	\forall i=1,\ldots,n.\ \bisimddisc(s_i,t_i) < \delta_i
		\ \implies
	\bisimddisc(f(s_1,\dots,s_n), f(t_1,\dots,t_n)) < \epsilon
\end{align*}
for all closed process terms $s_i,t_i \in \closedSTerms$.
\end{defi}
Note that by definition each non-expansive operator is also uniformly continuous (by $\delta_i = \epsilon / n$).
A uniformly continuous combinator $f$ ensures that for any non-zero bisimulation distance $\epsilon$ there are appropriate non-zero bisimulation distances $\delta_i$
s.t. for any composed process $f(s_1,\dots,s_n)$ the distance to the composed process where each $s_i$ is replaced by any $t_i$ with $\bisimddisc(s_i,t_i) < \delta_i$ is $\bisimddisc(f(s_1,\dots,s_n), f(t_1,\dots,t_n)) < \epsilon$. We consider the uniform notion of continuity (technically, the $\delta_i$ depend only on $\epsilon$ and are independent of the concrete states $s_i$) because we aim at universal compositionality guarantees.

A particular case of uniform continuity is Lipschitz continuity, which requires that there is a constant $K \in \Rgez$ such that $\delta_i = \epsilon / (n\cdot K)$.
Intuitively, this ensures that the distance between the composed processes is limited in how fast it can change due to the change of the distance between the components.
\begin{defi}[Lipschitz continuous process combinator]\label{def:Lipschitz-continuous-operator}
A process combinator $f \in \Sigma$ is \emph{Lipschitz continuous} w.r.t.\ $\lambda$-bisimilarity metric $\bisimddisc$ if there exists a constant $K \in \Rgez$ with 
\[ 
	\bisimddisc(f(s_1,\dots,s_n), f(t_1,\dots,t_n)) 
		\le  
	K \sum_{i=1}^n \bisimddisc(s_i,t_i) 
\]
for all closed process terms $s_i,t_i \in \closedSTerms$.
\end{defi}

We refer to the constant $K$ in Definition~\ref{def:Lipschitz-continuous-operator} as the \emph{Lipschitz factor} for combinator $f$, and we may say that $f$ is \emph{$K$-Lipschitz continuous}.
Note that by definition a non-expansive operator is Lipschitz continuous (by $K=1$) and a Lipschitz continuous operator is uniformly continuous (by $\delta_i = \epsilon / (n\cdot K)$).

The distance bounds of Section~\ref{sec:distance_rec_procs} allow us to derive that finitely recursing process combinators are Lipschitz continuous (and therefore also uniformly continuous) w.r.t.\ both non-discounted and discounted bisimilarity metric (Theorem~\ref{thm:continuous_ops_nondiscounting_bisim}). On the contrary, unbounded recursing process combinators are Lipschitz continuous and uniformly continuous only w.r.t.\ discounted bisimilarity metric (Theorem~\ref{thm:continuous_ops_discounting_bisim} and Proposition~\ref{prop:infinite_recursion_not_continuous_w.r.t_d1}).

\begin{thm} \label{thm:continuous_ops_nondiscounting_bisim}
The process combinators 
\begin{itemize}
	\item finite iteration $\_^n$
	\item finite replication $!^n\_$
	\item probabilistic replication (bang) $!_p\_$
\end{itemize}
are Lipschitz continuous w.r.t.\ $\lambda$-bisimilarity metric $\bisimddisc$ for any $\lambda \in (0,1]$.
\end{thm}
\begin{proof}
For finite iteration operator, this follows directly from Propositions~\ref{prop:distance_finite_recursion}.\ref{prop:distance_finite_recursion:finite_iteration} and the observation that $\frac{1-(\lambda - \bisimddisc(s,t))^n}{1-(\lambda - \bisimddisc(s,t))} \le n = K$.
For finite  replication operator, this follows directly from Propositions~\ref{prop:distance_finite_recursion}.\ref{prop:distance_finite_recursion:finite_replication} and the observation that $\frac{1-(\lambda^2 - \lambda \bisimddisc(s,t))^n}{1-(\lambda^2 - \lambda \bisimddisc(s,t))} \le n = K$. 
For the probabilistic bang operator it follows from Proposition~\ref{prop:distance_finite_recursion}.\ref{prop:distance_probabilistic_bang} and the observation that $\frac{1}{1-(1-p)(\lambda^2 - \lambda \bisimddisc(s,t))} \le \frac{1}{1-(1-p)\lambda^2} = K$. 
\end{proof}

Note that the probabilistic bang operator is Lipschitz continuous w.r.t.\ non-discounted bisimilarity metric $\bisimd$ with $\lambda=1$ because in each step there is a non-zero probability that the process is not copied. On the contrary, the process ${s_1}^{*_p} s_2$ applying the probabilistic Kleene star creates with probability $1$ a copy of $s_1$ for actions that $s_1$ can and $s_2$ cannot perform. Hence, the probabilistic Kleene star operator $\_^{*_p}\_$ is uniformly continuous only for discounted bisimilarity metric with $\lambda<1$.

\begin{thm} \label{thm:continuous_ops_discounting_bisim}
The process combinators 
\begin{itemize}
	\item infinite iteration $\_^\omega$
	\item nondeterministic Kleene-star iteration $\_^*\_$
	\item probabilistic Kleene-star iteration $\_^{*_p}\_$, and
	\item infinite replication (bang) $!\_$
\end{itemize}
are Lipschitz continuous w.r.t.\ discounted $\lambda$-bisimilarity metric $\bisimddisc$ for any $\lambda \in (0,1)$.
\end{thm}
\begin{proof}
For infinite iteration, nondeterministic Kleene star iteration and probabilistic Kleene star iteration this follows by Proposition~\ref{prop:distance_finite_recursion}.\ref{prop:distance_infinite_recursion:infinite_iteration}, ~\ref{prop:distance_finite_recursion}.\ref{prop:distance_infinite_recursion:kleene_star},  \ref{prop:distance_finite_recursion}.\ref{prop:distance_probabilistic_kleene} 
and the observation that $\frac{1}{1-(\lambda - \bisimddisc(s,t))} \le \frac{1}{1-\lambda} = K$.
For infinite replication this follows by Proposition~\ref{prop:distance_finite_recursion}.\ref{prop:distance_infinite_recursion:infinite_replication}
and the observation that $\frac{1}{1-(\lambda^2 - \lambda \bisimddisc(s,t))} \le \frac{1}{1-\lambda^2} = K$.
\end{proof}

\begin{prop} \label{prop:infinite_recursion_not_continuous_w.r.t_d1}
None of the process combinators 
\begin{itemize}
	\item infinite iteration $\_^\omega$
	\item nondeterministic Kleene-star iteration $\_^*\_$
	\item probabilistic Kleene-star iteration $\_^{*_p}\_$, and
	\item infinite replication (bang) $!\_$
\end{itemize}
is uniformly continuous w.r.t.\ the non-discounted $\lambda$-bisimilarity metric $\bisimd$ with $\lambda=1$.
\end{prop}
\begin{proof}
Follows directly from Propositions~\ref{prop:distance_finite_recursion} and~\ref{prop:optimality_distances_rec_proc_combinators}. We will reason in detail for the first case of infinite iteration operator. Let $\epsilon$ be any fixed real with $0 < \epsilon < 1$. We will show that there is no $\delta>0$ s.t.\ for all $s,t \in \closedTerms$ with $\bisimd(s,t)<\delta$ we have $\bisimd(s^\omega,t^\omega)<\epsilon$. We will show this by contradiction. Assume there is some $\delta>0$. Consider $s = a.([1-\delta/2]\varepsilon \oplus [\delta/2]0)$ and $t = a.\varepsilon$. We have $\bisimd(s,t)=\delta/2 < \delta$ and $\bisimd(s^\omega,t^\omega) = 1 >\epsilon$. Contradiction. Similar reasoning applies also to the other process combinators. 
\end{proof}
Note that the processes used in the proof of Proposition~\ref{prop:infinite_recursion_not_continuous_w.r.t_d1} are witnesses that these combinators are not continuous at all.

Given any discount factor $\lambda$, all process combinators discussed so far that are uniformly continuous wrt.\ $\lambda$-bisimilarity metric $\bisimd$ are also Lipschitz continuous wrt.\ $\bisimd$.
We conclude this section by discussing the copy operator $\copyOp$ of~\cite{BIM95,FvGdW12} as an example of an operator being uniformly continuous but not Lipschitz continuous wrt.\ discounted $\lambda$-bisimilarity metric $\bisimd$ with any $\lambda \in (0,1)$.

The copy operator $\copyOp$  is defined by the rules
\begin{gather*}
	\SOSrule{x \trans[a] \mu}
			{\copyOp(x) \trans[a] \mu} (a \not\in \{l,r\})
\qquad\qquad
	\SOSrule{x \trans[l] \mu \quad x \trans[r] \nu}
			{\copyOp(x) \trans[s] \copyOp(\mu) \mid \copyOp(\nu)}
\end{gather*}
The copy operator $\copyOp$ specifies the fork operation of operating systems. 
Actions $l$ and $r$ are the left and right \emph{forking actions}, and $s$ is the resulting \emph{split action}.
The fork of $t$ is the process $\copyOp(t)$ evolving by $t$ to the parallel composition of the left fork ($l$-derivative of $t$) and the right fork ($r$-derivative of $t$). For all other actions $a \not\in \{l,r\}$ the process $\copyOp(t)$ mimics the behavior of $t$.

\begin{prop}
The copy operator $\copyOp$ is not Lipschitz continuous wrt.\  $\lambda$-bisimilarity metric $\bisimddisc$ for any $\lambda \in (0,1]$.
\end{prop}
\begin{proof}
Assume any discount factor $\lambda \in (0,1]$.
For any constant $L \in \Rgez$, we provide suitable CCS processes $s$ and $t$ s.t.\ $\bisimddisc(\copyOp(s),\copyOp(t)) > L \bisimddisc(s,t)$.
Let 
$s_{1} = l.([1-\epsilon]a \oplus [\epsilon] 0) + r. ([1-\epsilon]a \oplus [\epsilon] 0)$ and $t_1 = l.a + r.a$, and 
$s_{k+1} = l. s_k  +  r . s_k$ and $t_{k+1} = l.t_k + r. t_k$.
Clearly $\bisimddisc(s_k,t_k) = \lambda^k \epsilon$. 
Then $\bisimddisc(\copyOp(s_k),\copyOp(t_k)) = \lambda^{k}(1-(1-\epsilon)^{2^k})$. 
Hence, for any $k$ with $2^k > L$, $\bisimddisc(\copyOp(s),\copyOp(t))/\bisimddisc(s,t) = (1-(1-\epsilon)^{2^k}) / \epsilon > L$ holds for $s = s_k$, $t = t_k$ and all $0 < \epsilon < (2^k - L) / (2^{k-1}(2^k-1))$. Thus, the copy operator is not Lipschitz continuous  wrt.\  $\lambda$-bisimilarity metric $\bisimddisc$.
\end{proof}

To prove that the copy operator $\copyOp$ is uniformly continuous wrt.\ discounted  $\lambda$-bisimilarity metric $\bisimddisc$ with any $\lambda \in (0,1)$, we need some preliminary results.
First we show that 
the behavioral distance between two arbitrary terms $s$ and $t$ can be divided in the distance observable by the first $k$ steps and the distance observable after step $k$. 
The step discount $\lambda$ allows us to give the upper bound $\lambda^k$ on the distance observable after step $k$.

\begin{prop} \label{prop:bisim_distance_vs_k_projection_bisim_distance}
Let $P=(\Sigma,\Act,R)$ be a PTSS and $s,t \in \closedSTerms$ arbitrary closed terms. Then 
\[
\bisimddisc(s,t) \le \bisimddisc_k(s,t) + \lambda^{k}
\] 
for all $k \in \N$.
\end{prop}
\begin{proof}
By induction.
Case $k=0$ is trivial since $\lambda^{0}=1$. 
Let $(\bisimddisc-\epsilon) \colon \closedTerms \times \closedTerms \to [0,\epsilon]$ with $\epsilon\in[0,1]$ be the function defined by $(\bisimddisc-\epsilon)(s,t)=\max(\bisimddisc(s,t)-\epsilon,0)$. 
For the induction step, assume $\bisimddisc_k \sqsupseteq \bisimddisc-\lambda^k$.
It remains to show $\bisimddisc_{k+1} \sqsupseteq \bisimddisc-\lambda^{k+1}$. We reason as follows:
\begin{align*}
& \bisimddisc_{k+1}(s,t) \\
= & \sup_{a\in A} \left\{ \Hausdorff(\lambda \cdot \Kantorovich(\bisimddisc_{k}))(\mathit{der}(s,a), \mathit{der}(t,a)) \right\} \\
\ge & \sup_{a\in A} \left\{ \Hausdorff(\lambda \cdot \Kantorovich(\bisimddisc-\lambda^k))(\mathit{der}(s,a), \mathit{der}(t,a)) \right\} \\
\ge & \sup_{a\in A} \left\{ \Hausdorff(\lambda \cdot \Kantorovich(\bisimddisc))(\mathit{der}(s,a), \mathit{der}(t,a)) \right\} - \lambda^{k+1}\\
= & \bisimddisc(s,t) - \lambda^{k+1}
\end{align*}
by using the properties

\begin{equation} \label{eq:dist_props_Kantorovich_Hausdroff}
\begin{aligned}
	\Kantorovich(d) &\sqsupseteq \Kantorovich(d') \qquad\text{if } d \sqsupseteq d' \\
	\Hausdorff(d) &\sqsupseteq \Hausdorff(d') \qquad\text{if } d \sqsupseteq d' \\
	\Kantorovich(d - \epsilon)(\pi,\pi') &\ge \Kantorovich(d)(\pi,\pi') - \epsilon \\
	\Hausdorff(d - \epsilon)(\pi,\pi') &\ge \Hausdorff(d)(\pi,\pi') - \epsilon
\end{aligned}
\end{equation}
for any pseudometrics $d,d'$ and any $\epsilon \in [0,1]$,
definition of $\bisimddisc_{k+1}$ applied in step $1$,
induction hypothesis applied in step $2$, 
the fixpoint property of bisimulation metric $\bisimddisc(s,t) = \sup_{a\in A} \{ \Hausdorff(\lambda \cdot \Kantorovich(\bisimddisc))(\mathit{der}(s,a), \mathit{der}(t,a)) \}$ applied in step $4$,
and properties of Equation~\ref{eq:dist_props_Kantorovich_Hausdroff} applied in steps $2$ and $3$.
\end{proof}

Now we show that an operator is uniformly continuous w.r.t.\ the discounted $\lambda$-bisimilarity metric $\bisimddisc$ if this operator is Lipschitz continuous wrt.\ all up-to-$k$ $\lambda$-bisimilarity metrics $\bisimddisc_k$. 

\begin{thm} \label{thm:finite_projection_Lipschitz_implies_uniform_continuity}
Let $P=(\Sigma,\Act,R)$ be a PTSS and $\lambda < 1$.
If an operator $f \in \Sigma$ is Lipschitz continuous wrt.\ $\bisimddisc_k$ \ for each $k \in N$, then $f$ is uniformly continuous wrt.\ $\bisimddisc$.
\end{thm}
\begin{proof}
Assume that $f \in \Sigma$ is any $n$-ary operator. 
We prove that for any $\epsilon > 0$ there exist $\delta_1,\ldots,\delta_n >0$ such that  $\bisimddisc(f(s_1,\ldots,s_n),f(t_1,\ldots,t_n)) < \epsilon$ whenever $\bisimddisc(s_i,t_i) < \delta_i$ for all $i=1,\ldots,n$.
Let $L_k \in \Rgez$ be the Lipschitz factor for $f$ wrt.\ $\bisimddisc_k$, i.e.\ 
\[
\bisimddisc_k(f(s_1,\ldots,s_n),f(t_1,\ldots,t_n)) \le L_k \sum_{i=1}^n \bisimddisc_k(s_i,t_i).
\]
Together with Proposition~\ref{prop:bisim_distance_vs_k_projection_bisim_distance} and property $\bisimddisc_k \sqsubseteq \bisimddisc$ we get
\begin{equation}\label{eq:rel_d_dk}
	\bisimddisc(f(s_1,\ldots,s_n),f(t_1,\ldots,t_n)) \le L_k \sum_{i=1}^n \bisimddisc(s_i,t_i) + \lambda^{k}
\end{equation}
for all $k \in \N$.
Since $\lambda < 1$, there is some $m \in \N$ s.t. $\lambda^m < \epsilon$. Let $\delta_i \in (0,1]$ be such that 
\[
	\delta_i < \frac{\epsilon - \lambda^m}{n \cdot L_m}
\] 
If we take $\bisimddisc(s_i,t_i) < \delta_i$ for all $i=1,\ldots,n$ then we get
\begin{align*}
& \bisimddisc(f(s_1,\ldots,s_n),f(t_1,\ldots,t_n))  \\
\le  & L_m \sum_{i=1}^n \bisimddisc(s_i,t_i) + \lambda^{m} & \text{(Equation~\ref{eq:rel_d_dk})} \\
<  & L_m \sum_{i=1}^n \delta_i + \lambda^{m} \\
\le & L_m \sum_{i=1}^n \frac{\epsilon - \lambda^m}{n \cdot L_m} + \lambda^{m} \\
= & \epsilon
\end{align*}
thus concluding that that $f$ is uniformly continuous w.r.t.\ $\bisimddisc$.
\end{proof}

Now we show that the copy operator $\copyOp$ is Lipschitz-continuous wrt.\ the (not necessarily discounted) up-to-k $\lambda$-bisimilarity metric $\bisimddisc_k$ \ for any $k \ge 0$ and $\lambda \in (0,1]$.
Together with Theorem~\ref{thm:finite_projection_Lipschitz_implies_uniform_continuity} this allows us to derive that $\copyOp$ is uniformly continuous wrt. the discounted $\lambda$-bisimilarity metric $\bisimddisc$\ for any $\lambda \in (0,1)$.

\begin{prop}\label{prop:copy-Lipschitz-continuous}
The copy operator $\copyOp$ is Lipschitz continuous wrt.\ the up-to-k $\lambda$-bisimilarity metric $\bisimddisc_k$ \ for any $k \ge 0$ and $\lambda \in (0,1]$.
\end{prop}
\begin{proof}
For all $k \ge 0$, we show that the operator $\copyOp$ is $2^k$-Lipschitz continuous wrt.\ the up-to-$k$ $\lambda$-bisimilarity metric $\bisimddisc_k$, namely 
\[ 
\bisimddisc_{k}(\copyOp(s),\copyOp(t)) \le 2^{k} \bisimddisc_{k}(s,t)
\] 
holds for arbitrary terms $s,t \in \closedTerms$.
We proceed by induction over $k$.
The base case $k=0$ is immediate.
Consider the inductive step $k+1$.
The subcase $\bisimd_{k+1}(s,t) = 1$ is immediate. Consider the subcase $\bisimd_{k+1}(s,t) < 1$.
We consider now the two different rules specifying the copy operator and show that in each case whenever $\copyOp(s) \trans[a] \pi$ is derivable by some of the rules then there is a transition $\copyOp(t) \trans[a] \pi'$ derivable by the same rule s.t. $\lambda \cdot \Kantorovich(\bisimddisc_{k})(\pi,\pi') \le  2^{k+1} \bisimddisc_{k+1}(s,t)$, thus confirming the thesis.
\begin{enumerate}
\item
Assume that $\copyOp(s) \trans[a] \pi$  is derived by $s \trans[a] \pi$ with $a \in \Act \setminus\{l,r\}$.
Since $\bisimddisc_{k+1}(s,t)<1$ and $\bisimd_{k+1}$ satisfies the transfer condition of the bisimulation metrics, there exists a transition $t \trans[a] \pi'$  for a distributions $\pi'$ with $\lambda \cdot \Kantorovich(\bisimddisc_k)(\pi,\pi') \le \bisimddisc_{k+1}(s,t)$. Finally, from $t \trans[a] \pi'$ we derive $\copyOp(t) \trans[a] \pi'$.
\item
Assume that $\copyOp(s) \trans[a] \pi$  is derived by $s \trans[l] \pi_1$ and $s \trans[r] \pi_2$ with $a = s$ and $\pi = \copyOp(\pi_1) \mid \copyOp(\pi_2)$.
Since $\bisimddisc_{k+1}(s,t)<1$ and $\bisimd_{k+1}$ satisfies the transfer condition of the bisimulation metrics, there exist transitions $t \trans[l] \pi'_1$ and $t \trans[l] \pi'_2$ for distributions $\pi'_1,\pi_2'$ with $\lambda \cdot \Kantorovich(\bisimddisc_k)(\pi_1,\pi'_1) \le \bisimddisc_{k+1}(s,t)$ and  $\lambda \cdot \Kantorovich(\bisimddisc_k)(\pi_2,\pi'_2) \le \bisimddisc_{k+1}(s,t)$.
From $t \trans[l] \pi'_1$ and $t \trans[r] \pi'_2$ we derive $\copyOp(t) \trans[s] \copyOp(\pi'_1) \mid \copyOp(\pi'_2)$. 
Finally we have
\begin{align*}
& \lambda \Kantorovich(\bisimd_k)(\copyOp(\pi_1) \mid \copyOp(\pi_2),\copyOp(\pi_1') \mid \copyOp(\pi_2')) \\
\le & \lambda (  1 - (1 - \Kantorovich(\bisimd_k)(\copyOp(\pi_1) ,\copyOp(\pi_1'))) (1 - \Kantorovich(\bisimd_k)(\copyOp(\pi_2) ,\copyOp(\pi_2')))) \\
\le & \lambda (\Kantorovich(\bisimd_k)(\copyOp(\pi_1) ,\copyOp(\pi_1')) + \Kantorovich(\bisimd_k)(\copyOp(\pi_2) ,\copyOp(\pi_2'))) \\
\le & \lambda (2^k \Kantorovich(\bisimd_k)(\pi_1 ,\pi_1') + 2^k \Kantorovich(\bisimd_k)(\pi_2 ,\pi_2')) \\
\le & \lambda (2^k \bisimd_{k+1}(s,t)/\lambda + 2^k \bisimd_{k+1}(s,t)/\lambda) \\
= & 2^{k+1} \bisimd_{k+1}(s,t)
\end{align*}
with the first step by the inductive hypothesis and Theorem~\ref{thm:Kantorovich_lifting} (using the fact that the candidate modulus of continuity of operator $\mid$ given by $z(\epsilon_1,\epsilon_2)=\lambda[1 - (1-\epsilon_1/\lambda)(1-\epsilon_2/\lambda)]$ is concave), the third step again by the inductive hypothesis and  by Theorem~\ref{thm:Kantorovich_lifting} (using the fact that the candidate modulus of continuity of operator $\copyOp$ given by $z(\epsilon)=2^k \epsilon$ is concave).
\end{enumerate}
\end{proof}

\begin{thm}
The copy operator $\copyOp$ is uniformly continuous wrt.\ the discounted $\lambda$-bisimilarity metric $\bisimddisc$ \ for any $\lambda \in (0,1)$.
\end{thm}
\begin{proof}
Directly by Proposition~\ref{prop:copy-Lipschitz-continuous} and Theorem~\ref{thm:finite_projection_Lipschitz_implies_uniform_continuity}. 
\end{proof}

\section{Application} \label{sec:application}

\begin{figure*}[t!]
\hfill 
\begin{align*}
\procBRP(N,T,p,q) &= \procRC(N,T,p,q) \parallel_{B} \procTV, \text{ where } B = \{c(d,b) \mid d \in D, b \in \{0,1\}\} \cup \{\actAck,\actLost \} \\[1.0 ex]
\procRC(N,T,p,q) &= \Bigg[ \sum_{1 \le n \le N, n = 2k} i(n) . 
                                 \bigg(\procCH(0,T,p,q) \, ; \, \procCH(1,T,p,q) \bigg)^{\frac{n}{2}}  \\
                            &  \qquad  \qquad + \\
                            &  \quad \; \sum_{1 \le n \le N, n = 2k+1} i(n) .
                               \bigg(\bigg( \procCH(0,T,p,q) \, ; \, \procCH(1,T,p,q) \bigg)^{\frac{n-1}{2}} \! ; 
                                 \procCH(0,T,p,q) \bigg) \Bigg] ; \\[1.0 ex]
                           & \quad \quad \actResOk . \varepsilon \\[2.0 ex]
\procCH(b,t,p,q) &= \sum_{d \in D} i(d).\procCHprime(d,b,t,p,q) \\[2.0 ex]
\procCHprime(d,b,t,p,q) &= \begin{cases}
	(\bot. \, \procCHprime(d,b,t-1,p,q)) \; \oplus_p \; (c(d,b). \procCHH(d,b,t,p,q)) & \text{if } t>0 \\
	\actResNok & \text{if } t=0 
\end{cases} \\[2.0 ex]
\procCHH(d,b,t,p,q) &= \begin{cases}
(\actLost. \procCH'(d,b,t-1,p,q))
\; \oplus_q \; (\actAck. \varepsilon) & \text{if } t>0 \\
\actResNok & \text{if } t=0 
\end{cases} \\[2.0 ex]
TV &= \bigg[ \bigg( \bigg(\sum_{d \in D} c(d,1). (\actAck. \varepsilon+ \actLost. \varepsilon) \bigg)^*
  \bigg(\sum_{d \in D} c(d,0). o(d). (\actAck. \varepsilon + \actLost. \varepsilon)\bigg) \bigg) \, ;\\
     & \; \; \quad \bigg( \bigg(\sum_{d \in D} c(d,0). (\actAck. \varepsilon + \actLost. \varepsilon) \bigg)^*
   \bigg(\sum_{d \in D} c(d,1). o(d). (\actAck. \varepsilon + \actLost. \varepsilon)  \bigg) \bigg) \bigg]^\omega 
\end{align*}
\hfill 
\caption{Specification of the Bounded Retransmission Protocol}
\label{fig:brp_spec}
\end{figure*}

To advocate both uniform continuity as adequate property for compositional reasoning as well as bisimulation metric semantics as a suitable distance measure for performance validation of communication protocols, we exemplify the discussed compositional reasoning method by analyzing  the bounded retransmission protocol (BRP) as a case study.

The BRP allows us to transfer streams of data from a sender (e.g.\ a remote control RC) to a receiver (e.g.\ a TV).
The RC tries to send to the TV a stream of $n$ data, $d_0,\dots, d_{n-1}$, with each $d_i$ a member of the finite data domain $D$. The length $n$ of the stream is bounded by a given $N$. Each 
datum $d_i$ is sent separately and has probability $p$ to get lost. When the TV receives 
a datum $d_i$, it sends back 
an acknowledgment message to the RC, which may also get lost, with probability $q$. 
If the RC does not receive the acknowledgment for datum $d_i$ within a given time bound, it assumes that $d_i$ got lost and retries to transmit it. However, the maximal number of attempts for $d_i$ is a given $T$, meaning that $T$ failures for any datum $d_i$ imply the failure of the whole transmission. Since also the acknowledgment message may get lost, it may happen that the RC sends more than once the same datum $d_i$ notwithstanding that it was correctly received by the TV. Therefore, the RC attaches a control bit $b$ to each datum $d_i$ that it sends to the TV, s.t. the TV can recognize if this datum is original or already received. Data items at even positions, i.e.\ $d_{2k}$ for some $k\in \N$, get control bit $0$ attached, and data items at odd positions, i.e\ $d_{2k+1}$ for some $k\in \N$, get control bit $1$ attached.

The BRP 
is specified in Figure~\ref{fig:brp_spec}. Our specification adapts the nondeterministic process algebra specification of~\cite{Fok07} by refining the configuration of lossy channels. While in the nondeterministic setting a lossy channel (nondeterministically) either successfully transmits a datum $d_i$ or loses it, we attached a success and failure probability to this choice. The protocol specification $\procBRP(N,T,p,q)$ is parametrized by the quadruple $(N,T,p,q)$, with $N$ denoting the maximum length of the data stream, $T$ denoting how often a single datum may be retransmitted, $p$ the probability that a single attempt to transmit a datum may fail, and $q$ the probability that the acknowledgment may fail. 
The term $\procBRP(N,T,p,q)$ represents a system consisting of the RC interface to the TV modeled as process $\procRC(N,T,p,q)$, the TV interface to the RC modeled as process $\procTV$, and the channels $\procCH(b, t, p, q)$ for data transmission and $\procCHH(d,b,t,p,q)$ for acknowledgment.

The processes $\procRC(N,T,p,q)$ and $\procTV$ synchronize over the actions: 
\begin{enumerate}[label=(\roman*)]
\item $c(d,b)$, with $d \in D$ and $b \in \{0,1\}$, 
modeling the correct transmission of datum $d \in D$ and control bit $b \in \{0,1\}$ 
 from the RC to the TV;
\item $\actAck$, modeling the correct transmission of the acknowledgment message from the TV to the RC, and 
\item $\actLost$, used to model the timeout due to loss of the acknowledgment message.
\end{enumerate}
Timeout due to the loss of pair $(d,b)$ is modeled by action $\bot$ by the RC.

The process $\procRC(N,T,p,q)$ starts by receiving the size $n \le N$ of the data stream by some other RC component, by means of action $i(n)$. Then, for $n$ times it reads the datum $d_i$ from some other RC components by means of action $i(d)$ and tries to send it to the $\procTV$. If all $n$ data are sent successfully, then the other RC components are notified by means of action $\actResOk$. 
In case of $T$ failures for one datum, the whole
transmission fails and 
the other RC components are notified by means of action $\actResNok$. 
If the process $\procTV$ receives a pair $(d,b)$ from $\procRC(N,T,p,q)$ by action $c(d,b)$, then, if the datum $d$ is original, namely $b$ is the expected control bit, then $d$ is sent to the other TV components by means of action $o(d)$, otherwise $(d,b)$ is ignored.

To
advocate bisimulation metric semantics as a suitable distance measure for performance validation of communication protocols we 
translate performance properties of a BRP implementation with lossy channels $\procBRP(N,T,p,q)$ to the bisimulation distance between such an implementation and the specification  with perfect channels $\procBRP(N,T,0,0)$. 
In the following we assume that $\lambda=1$, namely no discount.

\begin{prop}\label{prop:relation_performance_prop_bisim_dist}
Let $N,T \in \N$ and $p,q \in [0,1]$.
\begin{enumerate}
\item \label{prop:relation_performance_prop_bisim_dist:BRP}
The bisimulation distance $\bisimd(\procBRP(N,T,0,0),\procBRP(N,T,p,q))=\epsilon$ relates 
as follows to the protocol performance properties:
\begin{enumerate}[label=\({\alph*}] 
	\item \label{prop:relation_performance_prop_bisim_dist:BRP_no_retry}
The likelihood that $N$ data items are sent and acknowledged without
any retry (this means $\procBRP(N,T,p,q)$ behaves as $\procBRP(N,T,0,0)$) is $1-\epsilon$.
	\item \label{prop:relation_performance_prop_bisim_dist:BRP_k_retry} The likelihood that $N$ data items are sent and acknowledged with exactly $k$ retries, for some $0 \le k \le N\cdot T$, is $(1-\epsilon) (1 - (1 - \epsilon)^{1/N})^k$.
	\item \label{prop:relation_performance_prop_bisim_dist:BRP_atmostk_retry}  The likelihood that $N$ data items are sent and acknowledged with at most $k \le N\cdot T$ retries is $(1- \epsilon) \frac{1 - (1 - (1 - \epsilon)^{1/N})^{k+1}}{(1-\epsilon)^{1/N}}$.
	\item \label{prop:relation_performance_prop_bisim_dist:BRP_n_items} 
The likelihood that at least $n \le N$ of the $N$ data items are sent and acknowledged is $(1- \epsilon) \frac{1 - (1 - (1 - \epsilon)^{1/n})^{nT+1}}{(1-\epsilon)^{1/n}}$.
	\item \label{prop:relation_performance_prop_bisim_dist:BRP_N_items}  The likelihood that all $N$ items are sent and acknowledged is $(1- \epsilon) \frac{1 - (1 - (1 - \epsilon)^{1/N})^{N\cdot T + 1}}{(1-\epsilon)^{1/N}}$.
\end{enumerate}

\item \label{prop:relation_performance_prop_bisim_dist:CH} 
The bisimulation distance $\bisimd(\procCH(b,T,0,0),\procCH(b,T,p,q))=\delta$ relates as follows to the channel performance properties:
\begin{enumerate}[label=\({\alph*}]
	\item The likelihood that one datum is sent and acknowledged without any retry is $1 - \delta$.
	\item \label{prop:relation_performance_prop_bisim_dist:CH_k_retries} The likelihood that one datum is sent and acknowledged with exactly $k$ retries, for some $k \le T$, is $(1 - \delta) \cdot \delta^{k}$.
	\item \label{prop:relation_performance_prop_bisim_dist:CH_<=k_retries}  The likelihood that one datum is sent and acknowledged with at most $k$ retries, for some $k \le T$, is $1 - \delta^{k+1}$.
%
                   \item The likelihood that one datum is sent and acknowledged is $1 - \delta^{T+1}$.
\end{enumerate}
\end{enumerate}
\end{prop}

\begin{proof}

\begin{enumerate}
\item First we note that $((1-p)(1-q))^N$ is the likelihood that $N$ data items are sent and acknowledged without any retry.
\begin{enumerate}[label=\({\alph*}] 
	\item \label{proof_prop:relation_performance_prop_bisim_dist:BRP_no_retry}
The result can be understood by observing that $\epsilon = 1-((1-p)(1-q))^N$ is the likelihood that at least one retry is needed to transmit the stream of $N$ data.
	\item \label{proof_prop:relation_performance_prop_bisim_dist:BRP_k_retry} 
The result can be understood by observing that  $(1-\epsilon) (1 - (1 - \epsilon)^{1/N})^k$ is the conjunct probability to have exactly $k$ failures in sending or acknowledging a datum (probability $(1 - (1 - \epsilon)^{1/N})^k$), and to have $N$ successes (probability $(1-\epsilon)$).
	\item \label{proof_prop:relation_performance_prop_bisim_dist:BRP_atmostk_retry}  
The result can be understood by observing that $(1- \epsilon) \frac{1 - (1 - (1 - \epsilon)^{1/N})^{k+1}}{(1-\epsilon)^{1/N}} = \sum_{i=0}^k (1-\epsilon) (1 - (1 - \epsilon)^{1/N})^i$, where $(1-\epsilon) (1 - (1 - \epsilon)^{1/N})^i$ is the likelihood to send the $N$ data with exactly $i$ retries (see item \ref{proof_prop:relation_performance_prop_bisim_dist:BRP_k_retry}).
	\item \label{proof_prop:relation_performance_prop_bisim_dist:BRP_n_items} 
This is item~\ref{prop:relation_performance_prop_bisim_dist:BRP_atmostk_retry} with $N$ instantiated with $n$ and $k$ instantiated with $n\cdot T$.
	\item \label{proof_prop:relation_performance_prop_bisim_dist:BRP_N_items} 
This is item~\ref{prop:relation_performance_prop_bisim_dist:BRP_atmostk_retry} with $k$ instantiated with $N\cdot T$.
\end{enumerate}

\item 
First we note that the likelihood that a single datum requires no retry is $(1-p)(1-q)$.
\begin{enumerate}[label=\({\alph*}]
	\item The result can be understood by observing that $\delta = 1-(1-p)(1-q)$ is the likelihood that a single datum requires at least one retry to be successfully transmitted and acknowledged.
	\item The result can be understood by observing that $(1-\delta) \cdot \delta^k = (1-p)(1-q) \cdot (1-(1-p)(1-q))^k$ is the conjunct probability to have $k$ failures (probability $(1-(1-p)(1-q))^k$) followed by a successful transmission (probability $(1-p)(1-q)$). 
	\item The result can be understood by observing that  $1 - \delta^{k+1} = \sum_{i=0}^{k} (1-\delta) \cdot \delta^i$, where $(1-\delta) \cdot \delta^i$ is the likelihood that one datum is sent and acknowledged with exactly $i$ retries (see item \ref{prop:relation_performance_prop_bisim_dist:CH_k_retries}).
               \item This is item \ref{prop:relation_performance_prop_bisim_dist:CH_<=k_retries} istantiated with $k = T$.
\end{enumerate}
\end{enumerate}
\end{proof}

Now we show that by applying the compositionality results given in the previous sections  (Propositions~\ref{prop:distance_nonextensively_composed_procs},~\ref{prop:distance_nonexpansively_composed_procs},~\ref{prop:distance_finite_recursion})
we can relate the bisimulation distance between the specification 
with perfect channels $\procBRP(N,T,0,0)$  and some implementation with lossy channel $\procBRP(N,T,p,q)$ of the entire protocol 
with the distances between the specification and some implementation of its respective components.
On the one hand, this allows us to derive from specified performance properties of the entire protocol individual performance requirements of its components (compositional verification). 
On the other hand, this allows us to infer from performance properties of the protocol components suitable performance guarantees on the entire protocol (compositional specification).
We show also that the same compositionality results allow us to relate the distance between the specification and some implementation with lossy channel of the entire protocol or some components to the parameters of the system.

\begin{prop} \label{prop:brp_distance_spec_impl}
Let $N,T \in \N$ and $p,q \in [0,1]$. For all 
$b \in \{0,1\}$ it holds:
\begin{enumerate}[label=\({alph*}]
	\item \label{prop:brp_distance_spec_impl:brp}
	  $\displaystyle\bisimd(\procBRP(N,T,0,0),\procBRP(N,T,p,q)) \le 1-(1-\bisimd(\procCH(b,T,0,0),\procCH(b,T,p,q)))^N$; 
	\item \label{prop:brp_distance_spec_impl:ch} 
                                $\bisimd(\procCH(b,T,0,0),\procCH(b,T,p,q)) \le 1-(1-p)(1-q)$.
                 \item \label{prop:brp_distance_spec_impl:brp2} $\displaystyle\bisimd(\procBRP(N,T,p,q), \procBRP(N,T,0,0)) \le 1-((1-p)(1-q))^N$
\end{enumerate}
\end{prop}
\begin{proof}
Consider case~(\ref{prop:brp_distance_spec_impl:brp}).
By Proposition~\ref{prop:distance_nonexpansively_composed_procs}.\ref{prop:distance_nonexpansively_composed_procs:csp_parallel_comp} we obtain $\displaystyle\bisimd(\procBRP(N,T,0,0),\procBRP(N,T,p,q)) \le  \bisimd(\procRC(N,T,0,0),\procRC(N,T,p,q)) +(1-\bisimd(\procRC(N,T,0,0),\procRC(N,T,p,q))) \bisimd(\procTV,\procTV)$.
By $\bisimd(\procTV,\procTV) =0$ we get
$\displaystyle\bisimd(\procBRP(N,T,0,0),\procBRP(N,T,p,q)) \le  \bisimd(\procRC(N,T,0,0),\procRC(N,T,p,q))$.
Then, by applying Propositions~\ref{prop:distance_nonextensively_composed_procs}.\ref{prop:distance_nonextensively_composed_procs:action_prefix},
~\ref{prop:distance_nonextensively_composed_procs}.\ref{prop:distance_nonextensively_composed_procs:nondet_alt_comp},
~\ref{prop:distance_nonexpansively_composed_procs}.\ref{prop:distance_nonexpansively_composed_procs:sequential_comp},
 and~\ref{prop:distance_finite_recursion}.\ref{prop:distance_finite_recursion:finite_iteration} we infer $\bisimd(\procRC(N,T,0,0),\procRC(N,T,p,q)) \le 1-(1-\bisimd(\procCH(b,T,0,0),\procCH(b,T,p,q)))^N$.

Case~(\ref{prop:brp_distance_spec_impl:ch}) follows directly from Proposition~\ref{prop:distance_nonextensively_composed_procs}.
More precisely, by Proposition~\ref{prop:distance_nonextensively_composed_procs}.\ref{prop:distance_nonextensively_composed_procs:action_prefix} we infer both inequalities $\bisimd(\procCH(b,t,p,q),\procCH(b,t,0,0)) \le p + (1-p) \bisimddisc(\procCHH(d,b,t,p,q),\procCHH(d,b,t,0,0))$ and $\bisimddisc(\procCHH(d,b,t,p,q),\procCHH(d,b,t,0,0)) \le q$, which give $\bisimd(\procCH(b,T,0,0),\procCH(b,T,p,q)) \le p + (1-p)q = 1-(1-p)(1-q)$.

Case~(\ref{prop:brp_distance_spec_impl:brp2}) follows directly from cases~(\ref{prop:brp_distance_spec_impl:brp}) and~(\ref{prop:brp_distance_spec_impl:ch}).
\end{proof}

To advocate uniform continuity as adequate property for compositional reasoning, we show that the uniform continuity of  
process combinators in $\procBRP(N,T,p,q)$ 
allows us to 
relate the distance between this implementation and the specification $\procBRP(N,T,0,0)$ (which relates by Proposition~\ref{prop:relation_performance_prop_bisim_dist} to performance properties of the entire protocol) to the concrete parameters $p,q$ and $N$ of the system. In detail, 
by Theorems~\ref{thm:non_extensiveness}, \ref{thm:non_expansiveness}, \ref{thm:continuous_ops_nondiscounting_bisim} 
we can derive that 
$\bisimd(\procBRP(N,T,p,q),\procBRP(N,T,0,0)) \le N/2 \cdot (\bisimd(\procCH(0,T,p,q), \procCH(0,T,0,0))  + \bisimd(\procCH(1,T,p,q), \procCH(1,T,0,0)))$
(see the proof of Proposition~\ref{prop:continuous_reasoning_brp_to_params} below).
Then, by Proposition~\ref{prop:brp_distance_spec_impl} we can derive 
 $N/2 \cdot (\bisimd(\procCH(0,T,p,q), \procCH(0,T,0,0))  + \bisimd(\procCH(1,T,p,q), \procCH(1,T,0,0))) \le N(1-(1-p)(1-q))$.
Summarizing, we can conclude that $\bisimd(\procBRP(N,T,p,q),\procBRP(N,T,0,0)) \le N(1-(1-p)(1-q))$, which allows us to infer an upper bound to $\bisimd(\procBRP(N,T,p,q),\procBRP(N,T,0,0))$ from suitable constraints for $p$ and $q$, as formalized in the following result. 

\
\begin{prop} \label{prop:continuous_reasoning_brp_to_params}
Let $N,T \in \N$ and $p,q \in [0,1]$. For all $\epsilon \ge 0$, $p+q-pq < \epsilon/N$ ensures
\[
	\bisimd(\procBRP(N,T,p,q),\procBRP(N,T,0,0)) < \epsilon
\]	
\end{prop}
\begin{proof}
Assume $N$ is even. Then:
\begin{align*}
	& \bisimd(\procBRP(N,T,p,q),\procBRP(N,T,0,0)) \\
\le & \bisimd(\procRC(N,T,p,q),\procRC(N,T,0,0)) + \bisimd(TV,TV) & \text{(Theorem~\ref{thm:non_expansiveness})} \\
= & \bisimd(\procRC(N,T,p,q),\procRC(N,T,0,0)) \\
\le & \bisimd((\procCH(0,T,p,q);\procCH(1,T,p,q))^{N/2},(\procCH(0,T,0,0);\procCH(1,T,0,0))^{N/2}) & \text{(Theorem~\ref{thm:non_extensiveness})} \\
\le & N/2 \cdot \bisimd(\procCH(0,T,p,q);\procCH(1,T,p,q),\procCH(0,T,0,0);\procCH(1,T,0,0)) & \text{(Theorem~\ref{thm:continuous_ops_nondiscounting_bisim})} \\
\le & N/2 \cdot (\bisimd(\procCH(0,T,p,q), \procCH(0,T,0,0)) + \bisimd(\procCH(1,T,p,q), \procCH(1,T,0,0))) & \text{ (Theorem~ \ref{thm:non_expansiveness})} \\
= & N(1-(1-p)(1-q))
\end{align*}
where in the third inequality 
we use the Lipschitz factor $n$ for the operator $\_^n$ that we obtained in the proof of Theorem~\ref{thm:continuous_ops_nondiscounting_bisim}. 
From $\bisimd(\procBRP(N,T,p,q),\procBRP(N,T,0,0)) \le N(1-(1-p)(1-q))$ the thesis follows.
The case that $N$ is odd is analogous.
\end{proof}

Combining Propositions~\ref{prop:relation_performance_prop_bisim_dist}~--~\ref{prop:continuous_reasoning_brp_to_params} allows us now to reason compositionally over a concrete scenario. We derive from a given performance requirement to transmit a stream of data 
the necessary performance properties of the channel components.

\begin{exa}
Consider the following scenario. We want to transmit a data stream of $N=20$ data items with at most $T=1$ retry per data item. We want to build an implementation that should satisfy the performance property `The likelihood that all 20 data items are successfully transmitted is at least $99\%$'. By applying Proposition~\ref{prop:relation_performance_prop_bisim_dist}.\ref{prop:relation_performance_prop_bisim_dist:BRP} we translate this performance property to the bisimulation distance 
$\bisimd(\procBRP(N,T,0,0),\procBRP(N,T,p,q)) \le 0.01052$ on the entire system. 
By applying Proposition~\ref{prop:brp_distance_spec_impl}.\ref{prop:brp_distance_spec_impl:brp} we derive the bisimulation distance for its channel component
$\bisimd(\procCH(b,T,0,0),\procCH(b,T,p,q) \le 0.00053$. 
By Proposition~\ref{prop:brp_distance_spec_impl}.\ref{prop:brp_distance_spec_impl:ch} this distance can be translated to appropriate parameters of the channel component, e.g.\ $p=0.0002$ and $q=0.00032$ or equivalently $p=0.020\%$ and $q=0.032\%$. Finally, Proposition~\ref{prop:relation_performance_prop_bisim_dist}.\ref{prop:relation_performance_prop_bisim_dist:CH} allows to translate the distance between the specification and implementation of the channel component back to an appropriate performance requirement, e.g.\ `The likelihood that one datum is successfully transmitted is at least 99.95\%'.
\end{exa}

\section{Conclusions} \label{sec:cmr:conclusion}

We argued that the notion of uniform continuity (Definition~\ref{def:continuous_operator}, generalizing the notions of non-expansiveness and non-extensiveness discussed by other researchers) is an appropriate property of process combinators to facilitate compositional reasoning w.r.t.\ bisimulation metric semantics. We showed that all standard (non-recursive and recursive) process algebra operators are uniformly continuous (Theorems~\ref{thm:non_extensiveness},~\ref{thm:non_expansiveness},~\ref{thm:continuous_ops_nondiscounting_bisim},~\ref{thm:continuous_ops_discounting_bisim}). In addition, we provided for all standard process algebra operators tight bounds on the distance between the composed processes (Propositions~\ref{prop:distance_nonextensively_composed_procs},~\ref{prop:distance_nonexpansively_composed_procs},~\ref{prop:distance_finite_recursion}). We exemplified how these results can be used to reason compositionally over protocols. In fact, they allow us to derive from performance requirements on the entire system appropriate performance properties of the respective components, and in reverse to induce from performance assumptions on the system components performance guarantees on the entire system.

We remark that the abstraction operator of probabilistic process algebras (that hides actions and makes them observable as non-distinguishable $\tau$-actions) is non-extensive. However, the power of abstraction and hiding can only be utilized by using also a behavioral semantics that treats the $\tau$-actions respectively as internal actions. 
We leave the development of weak and branching bisimulation metrics 
and the analysis of process algebra operators for those metrics as future work. A first analysis for weak bisimulation metric and observational congruence weak bisimulation metric (weak bisimulation metric with kernel equivalence being the largest congruence w.r.t.\ CSS operators contained in weak bisimulation equivalence) may be found in~\cite{DJGP02}.

The metric reasoning approach exemplified in Section~\ref{sec:application} is a sound method to reason compositionally over systems. 
However, 
the distance between composed systems might not be tight.
Let $C[x]$ be an open term describing a composed system with $x$ the placeholder for a subsystem. 
Given subsystems $s$ and $s'$,  
the distance $\bisimddisc(C[s],C[s'])$ might be below the composition of the compositionality properties of the operators in $C$ if some of the differences in the behaviors between $s$ and $s'$ do not induce different behaviors between $C[s]$ and $C[s']$. 
To exemplify this effect, consider the context $C[x]=x \mid b.0$ and subsystems $s=a.0$ and $s'=a.([1-\epsilon/\lambda]\varepsilon \oplus [\epsilon_i/\lambda]0)$ . Clearly $\bisimddisc(s,s')=\epsilon$. 
Then the compositional analysis gives $\bisimddisc(C[s],C[s']) \le \epsilon$.  However, $\bisimddisc(C[s],C[s'])=0$ because the behavioral distance between $s$ and $s'$ (observable only after executing action $a$) cannot be observed in the context $C[x]$ (which can only perform an action if the instances of $x$ perform action $b$).  Thus, $\bisimddisc(C[s],C[s'])=0$ since $s$ and $s'$ agree on the inability to perform action $b$.
One idea to tackle this problem is to develop the notion of context bisimulation. Given a context $C$, the $C$-bisimulation distance (bisimulation distance w.r.t.\ context $C$) between $s$ and $s'$ would measure only that degree of the bisimulation distance between $s$ and $s'$ that would induce different behavior between $x$ instantiated by $s$ and $x$ instantiated by $s'$. Using the notation $\bisimddisc_C$ for the $C$-bisimulation distance this would give the behavioral distance $\bisimddisc_C(s,s')=0$ (since $C$ derives only behavior from an initial $b$-move and $s$ and $s'$ agree on their inability to perform $b$-moves), while $\bisimddisc_C(b.0,b.([1-\epsilon/\lambda]\varepsilon \oplus [\epsilon_i/\lambda]0)=\epsilon$. It is clear that the context bisimulation distance is bounded by the bisimulation distance.
While it still allows for sound compositional metric reasoning it may lead to tighter bounds. We leave the detailed technical development and analysis as future work.

Another research direction is to generalize the analysis of concrete process algebra operators as discussed in this paper to general SOS rule and specification formats. 
The basic observation is that the compositionality results for the concrete probabilistic process algebra operators depend only on the specification rules of those operators, hence the question boils down
to develop SOS meta-theoretical results and appropriate rule and specification formats that guarantee that the specified operators are uniformly continuous.
In essence, we aim to develop the quantitative analogous of the well-established meta-theory for behavioral equivalence semantics~\cite{AFV01,MRG07}.
This approach has been already developed for notions of approximate probabilistic bisimulation \cite{Tin08,Tin10,GT13}.
Preliminary results show that in essence, a process combinator is uniformly continuous if the combined processes are copied only finitely many times along their evolution~\cite{GT14,GT15b,Geb15},
and more restrictive constraints guarantee the stronger compositional properties of Lipschitz continuity, non-expansiveness and non-extensiveness.
By following the \emph{divide and congruence} aproach \cite{FvGdW06,FvGdW12,GF12,FvG16,CGT16b}, formats for compositional properties can be obtained also through a suitable logical characterization of bisimilarity metric, like that in  \cite{CGT16a}.

Finally, we intend to explore further (as initiated in Section~\ref{sec:application}) the relation between various behavioral distance measures, e.g.\ 
convex bisimulation metric~\cite{DAMRS07}, 
trace metric~\cite{FL14}, and
total-variation distance based metrics~\cite{Mio14}
with performance properties of communication and security protocols. 
This will provide further practical means to apply process algebraic methods and compositional metric reasoning w.r.t.\ uniformly continuous process combinators.

\bibliographystyle{alpha}
\bibliography{lmcsGLT}

\newcommand{\etalchar}[1]{$^{#1}$}
\begin{thebibliography}{DAMRS07}

\bibitem[ABV94]{ABV94}
Luca Aceto, Bard Bloom, and Frits Vaandrager.
\newblock Turning {SOS} rules into equations.
\newblock {\em Information and Computation}, 111(1):1--52, 1994.

\bibitem[AFV01]{AFV01}
Luca Aceto, Wan~J. Fokkink, and Chris Verhoef.
\newblock Structural operational semantics.
\newblock In {\em Handbook of Process Algebra}, pages 197--292. Elsevier, 2001.

\bibitem[And99]{And99}
Suzana Andova.
\newblock Process algebra with probabilistic choice.
\newblock In {\em Proc.~ARTS'99}, volume 1601 of {\em LNCS}, pages 111--129.
  Springer, 1999.

\bibitem[And02]{And02}
Suzana Andova.
\newblock {\em Probabilistic process algebra}.
\newblock PhD thesis, Eindhoven University of Technology, 2002.

\bibitem[Bar04]{Bar04}
Falk Bartels.
\newblock {\em On generalised coinduction and probabilistic specification
  formats}.
\newblock PhD thesis, VU University Amsterdam, 2004.

\bibitem[BBLM13]{BBLM13b}
Giorgio Bacci, Giovanni Bacci, Kim~G. Larsen, and Radu Mardare.
\newblock Computing behavioral distances, compositionally.
\newblock In {\em Proc.~MFCS'13}, volume 8087 of {\em LNCS}, pages 74--85.
  Springer, 2013.

\bibitem[BIM95]{BIM95}
Bard Bloom, Sorin Istrail, and Albert~R. Meyer.
\newblock Bisimulation can't be traced.
\newblock {\em Journal of ACM}, 42:232--268, 1995.

\bibitem[CGPX14]{CGPX14}
Konstantinos Chatzikokolakis, Daniel Gebler, Catuscia Palamidessi, and Lili Xu.
\newblock Generalized bisimulation metrics.
\newblock In {\em Proc.~CONCUR'14}, volume 8704 of {\em LNCS}, pages 32--46.
  Springer, 2014.

\bibitem[CGT16a]{CGT16a}
Valentina Castiglioni, Daniel Gebler, and Simone Tini.
\newblock Logical characterization of bisimulation metrics.
\newblock In {\em Proc.~{QAPL} 2016}, EPTCS, 2016.

\bibitem[CGT16b]{CGT16b}
Valentina Castiglioni, Daniel Gebler, and Simone Tini.
\newblock Modal decomposition on nondeterministic probabilistic processes.
\newblock In {\em Proc.~CONCUR'16}, volume~59 of {\em LIPIcs}. Schloss Dagstuhl
  - Leibniz-Zentrum fuer Informatik, 2016.

\bibitem[DAHM03]{AHM03}
Luca De~Alfaro, Thomas~A. Henzinger, and Rupak Majumdar.
\newblock Discounting the {F}uture in {S}ystems {T}heory.
\newblock In {\em Proc.~ICALP'03}, volume 2719 of {\em LNCS}, pages 1022--1037.
  Springer, 2003.

\bibitem[DAMRS07]{DAMRS07}
Luca De~Alfaro, Rupak Majumdar, Vishwanath Raman, and Mari\"elle Stoelinga.
\newblock Game relations and metrics.
\newblock In {\em Proc.~LICS'07}, pages 99--108. IEEE, 2007.

\bibitem[DCPP06]{DCPP06}
Yuxin Deng, Tom Chothia, Catuscia Palamidessi, and Jun Pang.
\newblock Metrics for {A}ction-labelled {Q}uantitative {T}ransition {S}ystems.
\newblock In {\em Proc.~QAPL'05}, volume 153 of {\em ENTCS}, pages 79--96,
  2006.

\bibitem[DD07]{DD07}
Yuxin Deng and Wenjie Du.
\newblock Probabilistic barbed congruence.
\newblock In {\em Proc.~QAPL'07}, volume 190 of {\em ENTCS}, pages 185--203,
  2007.

\bibitem[DD09]{Den09}
Yuxin Deng and Wenjie Du.
\newblock The {K}antorovich metric in computer science: A brief survey.
\newblock In {\em Proc.~QAPL'09}, volume 253 of {\em ENTCS}, pages 73--82,
  2009.

\bibitem[DGJP04]{DGJP04}
Jos\'ee Desharnais, Vineet Gupta, Radha Jagadeesan, and Prakash Panangaden.
\newblock Metrics for labelled {M}arkov {P}rocesses.
\newblock {\em Theoretical Computer Science}, 318(3):323--354, 2004.

\bibitem[DGL15]{DGL15}
Pedro~R. D'Argenio, Daniel Gebler, and Matias~D. Lee.
\newblock A general {SOS} theory for the specification of probabilistic
  transition systems.
\newblock Accepted for I\&C. Also available at
  \url{http://www.few.vu.nl/~gebler/paper/sos-theory.pdf}, 2015.

\bibitem[DJGP02]{DJGP02}
Jos{\'e}e Desharnais, Radha Jagadeesan, Vineet Gupta, and Prakash Panangaden.
\newblock The metric analogue of weak bisimulation for probabilistic processes.
\newblock In {\em Proc.~LICS'02}, pages 413--422. IEEE, 2002.

\bibitem[DL12]{DL12}
Pedro~R. D'Argenio and Matias~D. Lee.
\newblock Probabilistic transition system specification: {C}ongruence and full
  abstraction of bisimulation.
\newblock In {\em Proc.~FoSSaCS'12}, volume 7213 of {\em LNCS}, pages 452--466.
  Springer, 2012.

\bibitem[DvGH{\etalchar{+}}07]{DGHMZ07b}
Yuxin Deng, Rob~J. van Glabbeek, Matthew Hennessy, Carroll Morgan, and Chenyi
  Zhang.
\newblock Remarks on testing probabilistic processes.
\newblock In {\em Computation, Meaning, and Logic: Articles dedicated to Gordon
  Plotkin}, volume 172 of {\em ENTCS}, pages 359--397, 2007.

\bibitem[FL14]{FL14}
Uli Fahrenberg and Axel Legay.
\newblock The quantitative linear-time-branching-time spectrum.
\newblock {\em Theoretical Computer Science}, 538:54--69, 2014.

\bibitem[Fok07]{Fok07}
Wan~J. Fokkink.
\newblock {\em Modelling distributed systems}.
\newblock Springer, 2007.

\bibitem[FvG16]{FvG16}
Wan~J. Fokkink and Rob~J. van Glabbeek.
\newblock Divide and congruence {II:} from decomposition of modal formulas to
  preservation of delay and weak bisimilarity.
\newblock {\em CoRR}, abs/1604.07530, 2016.

\bibitem[FvGdW06]{FvGdW06}
Wan~J. Fokkink, Rob~J. van Glabbeek, and Paulien de~Wind.
\newblock Compositionality of {H}ennessy-{M}ilner logic by structural
  operational semantics.
\newblock {\em Theoretical Computer Science}, 354(3):421--440, 2006.

\bibitem[FvGdW12]{FvGdW12}
Wan~J. Fokkink, Rob~J. van Glabbeek, and Paulien de~Wind.
\newblock Divide and congruence: From decomposition of modal formulas to
  preservation of branching and $\eta$-bisimilarity.
\newblock {\em Information and Computation}, 214:59--85, 2012.

\bibitem[Geb15]{Geb15}
Daniel Gebler.
\newblock {\em Robust SOS specifications of probabilistic processes}.
\newblock PhD thesis, VU University Amsterdam, 2015.

\bibitem[GF12]{GF12}
Daniel Gebler and Wan~J. Fokkink.
\newblock Compositionality of probabilistic {H}ennessy-{M}ilner logic through
  structural operational semantics.
\newblock In {\em Proc.~{CONCUR} 2012}, volume 7454 of {\em LNCS}, pages
  395--409. Springer, 2012.

\bibitem[GLT15]{GLT14}
Daniel Gebler, Kim~G. Larsen, and Simone Tini.
\newblock Compositional metric reasoning with {P}robabilistic {P}rocess
  {C}alculi.
\newblock In {\em Proc.~FoSSaCS'15}, volume 9034 of {\em LNCS}, pages 230--245.
  Springer, 2015.

\bibitem[GT13]{GT13}
Daniel Gebler and Simone Tini.
\newblock Compositionality of approximate bisimulation for probabilistic
  systems.
\newblock In {\em Proc.~EXPRESS/SOS'13}, volume 120 of {\em EPTCS}, pages
  32--46, 2013.

\bibitem[GT14]{GT14}
Daniel Gebler and Simone Tini.
\newblock Fixed-point characterization of compositionality properties of
  probabilistic processes combinators.
\newblock In {\em Proc.~EXPRESS/SOS'14}, volume 160 of {\em EPTCS}, pages
  63--78, 2014.

\bibitem[GT15]{GT15b}
Daniel Gebler and Simone Tini.
\newblock {SOS} specifications of probabilistic systems by uniformly continuous
  operators.
\newblock In {\em Proc.~CONCUR'15}, volume~42 of {\em LIPIcs}, pages 155--168.
  Schloss Dagstuhl - Leibniz-Zentrum fuer Informatik, 2015.

\bibitem[HJ94]{HJ94}
Hans Hansson and Bengt Jonsson.
\newblock A logic for reasoning about time and reliability.
\newblock {\em Formal Aspects of Computing}, 6(5):512--535, 1994.

\bibitem[JLY01]{JYL01}
Bengt Jonsson, Kim~G. Larsen, and Wang Yi.
\newblock Probabilistic extensions of {P}rocess {A}lgebras.
\newblock In {\em Handbook of Process Algebra}, pages 685--710. Elsevier, 2001.

\bibitem[Kel76]{Kel76}
Robert~M. Keller.
\newblock Formal verification of parallel programs.
\newblock {\em Communications of the ACM}, 19(7):371--384, 1976.

\bibitem[LGD12]{LGD12}
Matias~D. Lee, Daniel Gebler, and Pedro~R. D'Argenio.
\newblock Tree rules in probabilistic transition system specifications with
  negative and quantitative premises.
\newblock In {\em Proc.~EXPRESS/SOS'12}, volume~89 of {\em EPTCS}, pages
  115--130, 2012.

\bibitem[LS91]{LS91}
Kim~G. Larsen and Arne Skou.
\newblock Bisimulation through probabilistic testing.
\newblock {\em Information and Computation}, 94:1--28, 1991.

\bibitem[LT05]{LT05}
Ruggero Lanotte and Simone Tini.
\newblock Probabilistic congruence for semistochastic generative processes.
\newblock In {\em Proc.~FoSSaCS'05}, volume 3441 of {\em LNCS}, pages 63--78.
  Springer, 2005.

\bibitem[LT09]{LT09}
Ruggero Lanotte and Simone Tini.
\newblock Probabilistic bisimulation as a congruence.
\newblock {\em ACM Transactions on Computational Logic}, 10:1--48, 2009.

\bibitem[Mio14]{Mio14}
Matteo Mio.
\newblock Upper-expectation bisimilarity and {{\L}}ukasiewicz $\mu$-{C}alculus.
\newblock In {\em Proc.~FoSSaCS'14}, volume 8412 of {\em LNCS}, pages 335--350.
  Springer, 2014.

\bibitem[MRG07]{MRG07}
Mohammad~Reza Mousavi, Michel~A. Reniers, and Jan~Friso Groote.
\newblock Sos formats and meta-theory: 20 years after.
\newblock {\em Theoretical Computer Science}, 373(3):238--272, 2007.

\bibitem[MS13]{MS13}
Matteo Mio and Alex Simpson.
\newblock A proof system for compositional verification of probabilistic
  concurrent processes.
\newblock In {\em Proc.~FoSSaCS'13}, volume 7794 of {\em LNCS}, pages 161--176.
  Springer, 2013.

\bibitem[Pan09]{Pan09}
Prakash Panangaden.
\newblock {\em Labelled Markov Processes}.
\newblock Imperial College Press, 2009.

\bibitem[Seg95]{Seg95a}
Roberto Segala.
\newblock {\em Modeling and Verification of Randomized Distributed Real-Time
  Systems}.
\newblock PhD thesis, MIT, 1995.

\bibitem[SL95]{SL95}
Roberto Segala and Nancy Lynch.
\newblock Probabilistic simulations for probabilistic processes.
\newblock {\em Nordic Journal of Computing}, 2:250--273, 1995.

\bibitem[Ste94]{Ste94}
William~J. Stewart.
\newblock {\em Introduction to the numerical solution of {M}arkov Chains}.
\newblock Princeton University Press, 1994.

\bibitem[Tin08]{Tin08}
Simone Tini.
\newblock Non expansive $\epsilon$-bisimulations.
\newblock In {\em Proc.~AMAST'08}, volume 5140 of {\em LNCS}, pages 362--376.
  Springer, 2008.

\bibitem[Tin10]{Tin10}
Simone Tini.
\newblock Non-expansive $\epsilon$-bisimulations for probabilistic processes.
\newblock {\em Theoretical Computer Science}, 411:2202--2222, 2010.

\bibitem[vB12]{vB12}
Franck van Breugel.
\newblock On behavioural pseudometrics and closure ordinals.
\newblock {\em Information Processing Letters}, 112(19):715--718, 2012.

\bibitem[vBW01]{BW01b}
Franck van Breugel and James Worrell.
\newblock Towards quantitative verification of probabilistic transition
  systems.
\newblock In {\em Proc.~ICALP'01}, volume 2076 of {\em LNCS}, pages 421--432.
  Springer, 2001.

\bibitem[vBW05]{BW05}
Franck van Breugel and James Worrell.
\newblock A behavioural pseudometric for probabilistic transition systems.
\newblock {\em Theoretical Computer Science}, 331(1):115--142, 2005.

\bibitem[Vil08]{Vil08}
C{\'e}dric Villani.
\newblock {\em Optimal transport: old and new}.
\newblock Springer, 2008.

\end{thebibliography}

\end{document}